\def\year{2022}\relax

\documentclass[letterpaper]{article} 
\usepackage{aaai22}  
\usepackage{times}  
\usepackage{helvet}  
\usepackage{courier}  
\usepackage[hyphens]{url}  
\usepackage{graphicx} 
\usepackage{natbib}  
\usepackage{caption} 
\DeclareCaptionStyle{ruled}{labelfont=normalfont,labelsep=colon,strut=off} 
\usepackage{todonotes}

\usepackage{soul}
\usepackage{url}
\usepackage[utf8]{inputenc}
\usepackage{amsmath, amssymb, amsthm}
\usepackage{booktabs}
\usepackage{paralist}
\usepackage{enumitem}
\usepackage{tikz}
\usetikzlibrary{calc}
\usetikzlibrary{shapes}
\usetikzlibrary{arrows}
\usepackage{xspace}
\usepackage{subcaption}
\usepackage{stmaryrd}
\usepackage{textcomp}

\newtheorem{theorem}{Theorem}[section]
\newtheorem{remark}[theorem]{Remark}
\newtheorem{example}[theorem]{Example}
\theoremstyle{definition}
\newtheorem{definition}[theorem]{Definition}
\newtheorem{lemma}[theorem]{Lemma}
\newtheorem{corollary}[theorem]{Corollary}
\newtheorem{innernumberedlemma}{Lemma}
\newenvironment{numberedlemma}[1]
  {\renewcommand\theinnernumberedlemma{#1}\innernumberedlemma}
  {\endinnernumberedlemma}
\newtheorem{innernumberedtheorem}{Theorem}
\newenvironment{numberedtheorem}[1]
  {\renewcommand\theinnernumberedtheorem{#1}\innernumberedtheorem}
  {\endinnernumberedtheorem}

\newcommand{\lemref}[1]{Lem.~\ref{lem:#1}}
\newcommand{\lemsref}[2]{Lems.~\ref{lem:#1} and~\ref{lem:#2}}
\newcommand{\defref}[1]{Def.~\ref{def:#1}}

\newcommand{\secref}[1]{Sec.~\ref{sec:#1}}

\newcommand{\thmref}[1]{Thm.~\ref{thm:#1}}
\newcommand{\thmsref}[2]{Thms.~\ref{thm:#1} and \ref{thm:#2}}
\newcommand{\exaref}[1]{Ex.~\ref{exa:#1}}
\newcommand{\figref}[1]{Fig.~\ref{fig:#1}}

\newcommand{\criref}[1]{{\bfseries C\ref{cri:#1}}}
\renewcommand{\eqref}[1]{(\ref{eq:#1})}

\newcommand{\textaction}[1]{$\mathsf{#1}$}

\newcommand{\m}[1]{\mathsf{#1}}
\newcommand{\mc}[1]{\mathcal{#1}}
\renewcommand{\vec}[1]{\obar{#1}}
\renewcommand{\phi}{\varphi}

\renewcommand{\obar}[1]{\makebox[0pt]{$\phantom{#1}\overline{\phantom{#1}}$}#1}
\newcommand{\inquotes}[1]{\ensuremath{\text{\textgravedbl}\!{#1}\!\text{\textacutedbl}}}
\newcommand{\inn}{\,{\in}\,} 
\newcommand{\eqn}{\,{=}\,} 
\newcommand{\cseq}[1]{\mathbf{#1}}
\newcommand{\domino}[4]{\text{\tikz[baseline=-0.5ex]{\node[scale=.6, inner sep=0pt]{$%
\left[\begin{array}{@{\,}l@{=}l@{\,}}{#1}&{#2}\\{#3}&{#4}\\\end{array}\right]$}}}}

\tikzstyle{state}=[draw, circle, inner sep=1.5pt, line width=.7pt, scale=.6]
\tikzstyle{edge}=[draw, ->, line width=.5pt]
\tikzstyle{action}=[scale=.55]
\tikzstyle{caption}=[scale=.9]

\renewcommand{\AA}{\mc A} 
\newcommand{\BB}{\mc B} 
\newcommand{\CC}{\mc C} 
\newcommand{\LL}{\mc L} 
\newcommand{\NN}{\mc N} 
\newcommand{\equivBC}{\sim} 
\newcommand{\equivGC}[1][K]{\sim_{#1}}
\newcommand{\trans}[1]{\Delta_{#1}} 
\newcommand{\hist}{h} 
\newcommand{\cutoff}[2][K]{\lfloor #2\rfloor_{#1} } 
\newcommand{\Const}{{\mc K}} 
\newcommand{\MC}{\textup{MC}} 
\newcommand{\GC}{\textup{GC}} 
\newcommand{\until}{\mathrel{\mathsf{U}}} 
\newcommand{\update}{\mathit{update}} 
\newcommand{\guard}{\mathit{guard}} 
\newcommand{\vwrite}{\mathit{write}} 
\newcommand{\constr}{\mathit{constr}} 
\newcommand{\LBC}{\LL_{\BB\CC}} 
\newcommand{\last}{\lambda} 
\newcommand{\eqc}[1]{\llbracket{#1}\rrbracket} 
\newcommand{\NFApsi}[1][\psi]{{\NN}_{#1}}

\newcommand{\hasproperty}{admits a finite summary\xspace}
\newcommand{\haveproperty}{admit a finite summary\xspace}
\newcommand{\property}{finite summary\xspace}
\newcommand{\Property}{Finite summary\xspace}

\newcommand{\goto}[1]{\mathrel{\raisebox{-2pt}{$\xrightarrow{#1}$}}}

\newcommand{\pcnode}[3]{{#1}\nodepart{two}{#2}\nodepart{three}{#3}}
\newcommand{\yes}{\checkmark}
\newcommand{\no}{\textsf x}

\newcommand{\mydds}{DDSA\xspace}
\newcommand{\myddss}{DDSAs\xspace}

\newcommand{\tool}{\texttt{ada}\xspace}

\newcommand{\BLUE}[1]{{\color{blue!80!black}{#1}}}
\newcommand{\RED}[1]{{\color{red!80!black}{#1}}}
\newcommand{\GREEN}[1]{{\color{green!50!black}{#1}}}

\newcommand{\savespace}[1]{
\todo[linecolor=gray!80!black,backgroundcolor=gray!25,bordercolor=gray]{\footnotesize drop to save space?}
{\color{gray}#1}}
\renewcommand{\savespace}[1]{}

\title{Linear-Time Verification of Data-Aware Dynamic Systems with Arithmetic}

\author{Paolo Felli, Marco Montali, Sarah Winkler\thanks{This work is partially supported by the UNIBZ projects DaCoMan, QUEST, SMART-APP, VERBA, and WineId.}
}

\affiliations{Free University of Bozen-Bolzano -- Bolzano -- Italy}

\begin{document}

\maketitle

\begin{abstract}
Combined modeling and verification of dynamic systems and the data they operate on has gained momentum in AI and in several application domains. We investigate the expressive yet concise framework of data-aware dynamic systems (DDS), extending it with linear arithmetic, and provide the following contributions. 
First, we introduce a new, semantic property of ``finite summary'', which guarantees the existence of a faithful finite-state abstraction. We rely on this to show that checking whether a witness exists for a linear-time, finite-trace property is decidable for DDSs with finite summary. 
Second, we demonstrate that several decidability conditions studied in formal methods and database theory can be seen as concrete, checkable instances of this property. This also gives rise to new decidability results.
Third, we show how the abstract, uniform property of finite summary leads to modularity results: a system enjoys finite summary if it can be partitioned appropriately into smaller systems that possess the property. 
Our results allow us to analyze systems that were out of reach in earlier approaches. 
Finally, we demonstrate the feasibility of our approach in a prototype implementation.
\end{abstract}

\section{Introduction}

The analysis of complex dynamic systems is a core research topic in AI.
While process analysis has long focused on the control-flow perspective,
in recent years a multi-perspective approach gained momentum,
studying the interplay between control flow and data~\cite{Reichert12,CGM13,CDMP18,DHLV18}.
Verification in this setting is challenging, as it must deal with potentially infinitely many states.

This is aggravated in the presence of arithmetic, notwithstanding that
it is essential for practical applications \cite{DHLV18}:
model checking of transition systems operating over simple data with arithmetic constraints is known to be undecidable, as it is easy to model a two-counter system.
However, restrictions on the transition system have been shown to render certain verification tasks decidable.
In particular, decidability has been obtained by confining the constraint language, as in the case of \emph{monotonicity} constraints~\cite{DD07} (e.g. $x\:{\leq}\:y$) and 
\emph{gap-order} constraints~\cite{MT16,BP14} (e.g. $x\,{-}\,y\:{\geq}\:2$), or by limiting the control flow, as in the case of \emph{feedback freedom}~\cite{DDV12}.

In this work, we focus on the framework of data-aware dynamic systems (DDSs) \cite{LFM20}, an expressive yet concise model for process analysis, which we enrich with linear arithmetic. We call the resulting systems \emph{DDSs with arithmetic (\myddss)}, and study the verification problem for the linear-time, finite-trace temporal logic LTL$_f$~\cite{dGV13} extended with arithmetic constraints.
The following is a motivating example.
\begin{example}
\label{exa:auction}
Consider the process of an auction at an online market place. 
Its data variables are a timer $d$,
the offer $o$ by the last bidder, identified by $b$,
a threshold price $t$ for which the item can be sold immediately,
and the sum $s$.
\resizebox{\columnwidth}{!}
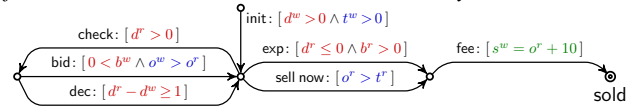
{
\begin{tikzpicture}[node distance=42mm,>=stealth']
\tikzstyle{action}=[scale=.6]
\node[state] (1)  {};
\node[state, below of=1, xshift=-0mm, yshift=26mm] (2) {};
\node[state, left of=2, xshift=-10mm] (3) {};
\node[state, right of=2, xshift=2mm] (4) {};
\node[state, right of=4, double, xshift=-0mm] (5) {};
\node[scale=.8, yshift=-3mm] at (5) {$\mathsf{sold}$};
\draw[edge] (1) to node[right,action, very near start]
{$\mathsf{init}\colon[\,\RED{d^w\,{>}\,0}\wedge \BLUE{t^w\,{>}\,0}\,]$} (2);
\draw[edge] (3) to node[above,action] {$\mathsf{bid}\colon[\,\RED{0 < b^w} \wedge \BLUE{o^w > o^r}\,]$} (2);
\draw[edge, rounded corners] (2) -- ($(2) + (-.2,.4)$) -- node[above,action]{$\mathsf{check}\colon[\,\RED{d^r > 0}\,]$} ($(3) + (.2,.4)$) -- (3);
\draw[edge, rounded corners] (3) -- ($(3) + (.2,-.4)$) to node[above,action]
{$\mathsf{dec}\colon[\,\RED{d^r\,{-}\,d^w\,{\geq}\,1}\,]$} ($(2) + (-.2,-.4)$) -- (2);
\draw[edge, rounded corners] (2) -- ($(2) + (.2,.2)$)
  -- node[above,action]{$\mathsf{exp}\colon[\,\RED{d^r \leq 0} \wedge \RED{b^r > 0}\,]$}  ($(4) + (-.2,.2)$)  -- (4);
\draw[edge, rounded corners] (2) -- ($(2) + (.2,-.2)$)
  -- node[above,action]{$\mathsf{sell\: now}\colon[\,\BLUE{o^r > t^r}\,]$}  ($(4) - (.2,.2)$)  -- (4);
\draw[edge, rounded corners] (4) -- ($(4) + (.2,.2)$) -- node[above,action]{$\mathsf{fee}\colon[\,\GREEN{s^w = o^r + 10}\,]$} ($(5) + (-.2,.2)$) -- (5);
\end{tikzpicture}
}\\
The timer $d$ is initialized to a number of days, and $t$ is fixed (action \textaction{init}).
Then, while the timer did not expire (\textaction{check}), bids are taken (\textaction{bid}) or the timer may be decremented (\textaction{dec}).
The auction ends if the timer expires and a bid was set (\textaction{exp}),
or the offer exceeds $t$ (\textaction{sell\:now}).
Finally, \textaction{fee} sets $s$ to the offer plus an auction fee. We will use our approach to verify that
$\psi = \Box(\mathsf{sold} \wedge d\,{>}\,0 \to o\,{>}\,t)$ holds, i.e., if the auction
ends
before the timer expires, the offer exceeds the threshold.
The meaning of the colors will be clarified later.
\end{example}

\noindent
Our contribution is as follows. 
\begin{inparaenum}[(1)]
\item First, we introduce the novel property of \emph{\property}, and show that the above restrictions studied in the literature (i.e. $(i)$ monotonicity constraints, $(ii)$ gap-order constraints, $(iii)$ feedback freedom) are instances of this property. We further generalize feedback freedom introducing a new, expressive property called $(iv)$ \emph{bounded lookback}. 
\item Second, we prove that finite summary guarantees the existence of a \emph{faithful, finite-state abstraction} for a \mydds by representing sets of states as logical constraints. This is used to show that checking existence of a witness for an LTL$_f$ property is decidable. 
\item Third, we illustrate a \emph{modularity} result: if a \mydds $\BB$ represents either the sequential, or parallel but variable-disjoint, execution of \myddss with finite summary (possibly according to the different criteria $(i)$-$(iv)$), then also $\BB$ enjoys this property and is thus amenable to our verification technique.
\end{inparaenum}

To the best of our knowledge, LTL$_f$ model checking of such  combinations of $(i)-(iv)$ is shown decidable for the first time (and the result is new for $(ii),(iv)$ individually).

To demonstrate feasibility, we implemented our approach in the tool \tool, which tests for finite summary using $(i)-(iv)$, computes finite-state abstractions, and handles LTL$_f$ model checking using an SMT solver as backend.

\smallskip
\noindent
\textbf{Related Work.}
Verification of transition systems with arith\-metic constraints
has been studied in many areas including 
formal methods, database theory, and BPM.
For monotonicity constraint (MC) systems, LTL model checking 
was proven decidable in~\cite{DD07}, even comparing variables multiple steps apart.
An extended language is studied in~\cite{Demri06}.
\mydds{s} with MCs are also considered in \cite{FLM19} from the perspective of
a finite-run semantics, giving an explicit procedure to compute finite, faithful abstractions.
For gap-order constraint (GC) systems, reachability was shown decidable~\cite{BGI09}.
Also the existential fragment of CTL$^*$ with \GC s is decidable,
while the universal one is not~\cite{BP14}.
A similar di\-cho\-tomy was discovered for the EF and EG fragments of CTL~\cite{MT16}.
We here consider LTL$_f$ model checking, 
a task suited to many applications~\cite{dGV13}:
For \mydds{s} with a finite summary, we prove decidability of our verification task, i.e., 
to check existence of a witness for an LTL$_f$ formula with constraints. 
\Property is based on the notion of \emph{history constraints} from \cite{DDV12},
and we show that it generalizes their feedback freedom property, though their
constraints may refer to a read-only database, a feature that we leave for future work.
\myddss generalize timed automata, and in fact our abstraction shares with the famous region graph the representation of a ``region'' of reachable states by a formula~\cite{AD94}.
The finite summary property does not cover timed automata with multiple clocks, whereas the one-clock-case is captured by MCs.
Abstracting reachable states by formulas is an approach that was also pursued in~\cite{BDD13}.
However, our results are incomparable to both of these works.
Our method can be seen as a form of predicate abstraction,
subject to a long line of research 
(e.g., \cite{ClarkeKSY04,ColonU98});
but in contrast to most works there, our abstraction is 
\emph{strongly} preserving, i.e., our verification task is decidable.

\smallskip
\noindent
\textbf{Paper structure.} 
In \secref{dds} we formalize \mydds{s} and our verification language and task. 
In \secref{finite summary} we develop the notion of \property and show how it yields finite state abstractions.
\secref{model checking} is devoted to our verification technique.
In \secref{conditions} we demonstrate four concrete classes implying finite summary, and in \secref{modularity} we present modulari\-ty results. \secref{implementation} describes our tool \tool and concludes with directions for future work.  
All proofs and further examples can be found in an extended version~\cite{adax}.

\section{DDSs with Arithmetic}
\label{sec:dds}

In this section we fix our model and verification language:
we enrich \emph{data-aware dynamic systems (DDSs)} from \cite{LFM20}
with linear arithmetic constraints, and extend the linear-time verification language correspondingly.

\smallskip
\noindent
\textbf{Model.}
We start by defining the set of 
arithmetic constraints over a domain $D$,
which may be $\mathbb Z$, $\mathbb Q$, or $\mathbb R$:
\begin{definition}
\label{def:constraint}
A \emph{constraint} $c$ over a set $V$ of variables is defined by the following grammar, where $k\in D$ and $v \in V$: \\
\begin{tabular}{l@{~}l}
$e$ &:= $v\ \mid\ k\ \mid\ e + e\ \mid\ e - e\ $ \\
$c$ &:= $e = e\ \mid\ e \neq e\ \mid\ e < e\ \mid\ e \leq e\ \mid\ c \wedge c$ 
\end{tabular}
\label{def:constraints}
\end{definition}
\noindent
The set of all constraints over domain $D$ is denoted by $\CC_D$.
E.g., $x \neq 1$, $x < y\,{-}\,z$, and $x\,{-}\,y = 2$ are constraints over 
$\{x, y, z\}$ for domain $\mathbb Z$, $\mathbb Q$, or $\mathbb R$. 
From now on, $V$ will be a fixed, finite set of variables. 
Two disjoint copies $V^r$ and $V^w$ of $V$, called the \emph{read} and \emph{write} variables,
denote the variable values before and after a transition, respectively. 
We also write $\vec V$ for a vector that contains the variables $V$ in an arbitrary but fixed order, and $\vec V^r$ and $\vec V^w$ for $V^r$ and $V^w$ ordered in the same way.
Throughout this paper, by a \emph{formula} $\phi$ we mean a boolean formula
whose atoms are either propositional or constraints as in \defref{constraint}.
We are thus in the realm of SMT with linear arithmetic, which is decidable and admits \emph{quantifier elimination}:
if $\phi$ is a formula with free variables $X \cup \{y\}$, 
and atoms in $\CC_D$ (cf. \defref{constraints}),
there is some $\phi'$ with free
vari\-ables $X$ that is equivalent to $\exists y. \phi$, 
i.e., $\phi'\,{\equiv}\,\exists y. \phi$~\cite{Presburger29}.
Here the relation $\equiv$ denotes logical equivalence.
For a set $C$ of constraints and a formula $\phi$, 
we sometimes write $\phi \wedge C$ for the formula $\phi \wedge \bigwedge C$.

A \emph{state variable assignment} $\alpha$ is a total function
$\alpha \colon V \mapsto D$;  we say that $\alpha$ \emph{satisfies} a constraint $c$ over $V$, 
written $\alpha \models c$, if the evaluation of $c$ under $\alpha$ is true in $D$.
\begin{definition}
\label{def:DDSA}
A \emph{DDS with arithmetic (\mydds)}  is
a labelled transition system $\langle B, b_0, \AA, T, F, V, \alpha_0, guard\rangle$,  where:
\begin{compactitem}[$\bullet$]
\item $B$ is a finite set of \emph{states}, with $b_0\in B$ the initial one;
\item $\AA$ is a finite set of \emph{actions};
\item $T\colon B \times \AA \mapsto B$ is a \emph{transition function};
\item $F \subseteq B$ is the set of \emph{final states};
\item $\alpha_0$ is the \emph{initial state variable assignment}; and
\item $guard\colon \AA \mapsto \CC_D$ specifies \emph{executability constraints} 
on actions over variables $V^r\cup V^w$.
\end{compactitem}
\end{definition}

\noindent
In \defref{DDSA} we restrict to conjunctive guards: disjunction can be captured by 
multiple transitions between the same states.
With this convention, the system in \exaref{auction} can be transformed into an
equivalent \mydds, and \figref{examples} shows further examples of \myddss.
Note that a guard simultaneously expresses a condition on the read variables, and an update on the written ones:
for instance, $v^r<7$ requires the current value of $v$ to be less than $7$, while $v^w - v^r \leq 7$ demands that the new value of $v$ exceeds the current value by at most $7$. 

We denote a transition from state $b$ to $b'$ by executing an action $a\inn \AA$ as 
$b \goto{a} b'$.
A \emph{configuration} of $\BB$ is a pair $(b, \alpha)$ where $b\inn B$
and $\alpha$ is a state variable assignment.

A \emph{guard assignment} $\beta$ is a function $\beta\colon V^r \cup V^w \mapsto D$.
As defined next, an action $a$ transforms a configuration $(b, \alpha)$ into a new configuration $(b', \alpha')$ by changing state as defined by  action $a$, and updating the state variable assignment in agreement with the action guard. In the new assignment $\alpha'$, variables that are not written keep their previous value as per $\alpha$, whereas written variables are updated according to the guard.
%
Let $write(a) = \{x \mid x^w\in V^w\text{ occurs in }guard(a)\}$.

\begin{definition}
A \mydds $\BB\,{=}\,\langle B, b_0, \AA, T, F, V, \alpha_0, \guard\rangle$
\emph{admits a step} from configuration $(b, \alpha)$ to 
$(b', \alpha')$ via action $a$,
denoted $(b, \alpha) \goto{a} (b', \alpha')$,
if $b \goto{a} b'$ and
the guard assignment $\beta$ given by
$\beta(v^r) = \alpha(v)$ and
$\beta(v^w) = \alpha'(v)$ for all $v \in V$
satisfies the guard of $a$, i.e., $\beta \models \guard(a)$ holds.
\end{definition}
\noindent
A \emph{run} of length $n$ is a sequence of steps
$\smash{\rho\colon(b_0, \alpha_0) 
\goto{a_1}}$
$\smash{ (b_1, \alpha_1)
\goto{a_2} \dots
\goto{a_n} (b_n, \alpha_n)}$,
and $\rho_i$ refers to $(b_i, \alpha_i)$.
Note that a run always starts in the initial state $(b_0, \alpha_0)$.

\smallskip
\noindent
\textbf{Specification language.}
\label{sec:lang}
For a constraint set $\CC$ over $V$ and
\mydds $\BB = \langle B, b_0, \AA, T, F, V, \alpha_0, guard\rangle$,
let $\LBC$ be the language defined by the following grammar: \\[.5ex]
%
\begin{tabular}{@{\quad}c@{\quad}}
$c \mid b \mid 
\psi {\wedge} \psi \mid  \psi {\vee} \psi \mid 
\langle a\rangle \psi \mid \langle\cdot\rangle \psi \mid 
\Diamond \psi \mid \Box \psi \mid \psi \until \psi$
\end{tabular}
where $a \inn\AA$, $c \inn \CC$, and $b\inn B$.
%
Note that $\LBC$ does not support negation as we will also consider fragments 
where decidability is lost if constraints can be negated.
However, if the set $\CC$ is closed under negation,\footnote{Here a constraint set $\CC$ is \emph{closed under negation} if for all $c\,{\in}\,\CC$ there is some $c'\,{\in}\,\CC$ such that $c' \equiv \neg c$.} $\LBC$ can express an
arbitrary formula in negation normal form. We adapt LTL$_f$ semantics~\cite{dGV13}:

\begin{figure}[t]
\centering
\resizebox{\columnwidth}{!}{
\begin{tikzpicture}[node distance=22mm,>=stealth']

\begin{scope}[shift={(0mm,0mm)}, node distance=24mm]
\node[caption] at (-.7,0) {$\BB_1$};
\node[state] (1) {$\mathsf 1$};
\node[state, right of=1, double] (2) {$\mathsf 2$};
\draw[edge] ($(1) + (-.4,0)$) -- (1);
\draw[edge, rounded corners] (1) -- node[above,action]{$\mathsf{a}_1\colon [x^w > y^r]$} (2);
\draw[edge, rounded corners] (2) -- ($(2) - (0,.3)$)
  -- node[below,action]{$\mathsf{a}_2\colon [y^w > x^r]$}  ($(1) - (0,.3)$)  -- (1);
\end{scope}

\begin{scope}[shift={(31mm,-2mm)},node distance=27mm]
\node[caption] at (-.7,0) {$\BB_2$};
\node[state] (1) {$\m 1$};
\node[state, right of=1,xshift=-5mm] (2) {$\m 2$};
\node[state, right of=2, double] (3) {$\m 3$};
\draw[edge] ($(1) + (-.4,0)$) -- (1);
\draw[edge] (1) -- node[pos=.3,above,action]{$[x^w > 0]$} (2);
\draw[edge] (2) to[loop left, out=120,in=60, looseness=6] node[action, xshift=-1mm,yshift=1mm]{$[y^w > x^r]$} (2);
\draw[edge] (2) -- node[above,action]{$[y^r \leq x^r+7]$} (3);
\end{scope}

\begin{scope}[shift={(0mm,-12mm)},node distance=24mm]
\node[caption] at (-.7,0) {$\BB_3$};
\node[state] (1) {$\m 1$};
\node[state, right of=1, double] (2) {$\m 2$};
\draw[edge] ($(1) + (-.4,0)$) -- (1);
\draw[edge] (1) -- node[above,action]{$[x^w - y^r \geq 2]$} (2);
\draw[edge, rounded corners] (2) -- ($(2) - (0,.4)$)
  -- node[above,action]{$[y^w - y^r \geq 3]$}  ($(1) - (0,.4)$)  -- (1);
\end{scope}

\begin{scope}[shift={(31mm,-12mm)}, node distance=27mm]
\node[caption] at (-.7,0) {$\BB_4$};
\node[state] (1)  {$\m 1$};
\node[state, right of=1, xshift=-5mm] (2) {$\m 2$};
\node[state, right of=2, double] (3) {$\m 3$};
\draw[edge] ($(1) + (-.4,0)$) -- (1);
\draw[edge] (1) to node[above,action] {$[s^w = a^r]$} (2);
\draw[edge] (2) to node[above,action]{$[s^w = s^r + b^r]$} (3);
\draw[edge, rounded corners] (3) -- ($(3) - (0,.4)$)
  -- node[above,action, pos=0.55]{$[a^w = 0 \wedge b^w = 0]$}  ($(1) - (0,.4)$)  -- (1);
\draw[edge] (1) to[loop above, out=120,in=60, looseness=6] node[action, yshift=-1mm]{$[a^w > 0]$} (1);
\draw[edge] (2) to[loop above, out=120,in=60, looseness=6] node[action, yshift=-1mm]{$[b^w > 0]$} (2);
\end{scope}
\end{tikzpicture}
}
\caption{Simple \mydds{s} (with \property).}
\label{fig:examples}
\end{figure}

\begin{definition}
\label{def:witness}
A run $\rho$ of length $n$ \emph{satisfies} $\psi \in \LBC$, denoted 
$\rho \models \psi$, iff $\rho,0 \models \psi$ holds, 
where for $0 \leq i \leq n$:

\noindent
\begin{tabular}{@{~}l@{ ~ }l}
$\rho,i \models c$  & iff $\rho_i = (b,\alpha)$ for some $b$
 and $\alpha\models c$\\
$\rho,i \models b$ & iff $\rho_i = (b,\alpha)$ for   
 some $\alpha$\\
$\rho,i \models \psi_1 \wedge \psi_2$ & iff 
$\rho,i \models \psi_1$ and $\rho,i \models \psi_2$\\
\end{tabular}
\noindent
\begin{tabular}{@{~}l@{ ~ }l}
$\rho,i \models \psi_1 \vee \psi_2$ & iff 
$\rho,i \models \psi_1$ or $\rho,i \models \psi_2$\\
$\rho,i \models \langle a\rangle\psi$ & iff $i<n$, $\exists \beta$
$\smash{\rho_i \goto{a,\beta} \rho_{i{+}1}}$ and $\rho,i{+}1 \models \psi$\\
$\rho,i \models \langle\cdot\rangle\psi$ & iff $i<n$ and 
 $\rho,i{+}1 \models \psi$\\
$\rho,i \models \Diamond\psi$ & iff 
$\rho,i \models \psi$ or ($i<n$ and
$\rho,i{+}1\models \Diamond\psi$)\\
$\rho,i \models \Box\psi$ & iff 
$\rho,i \models \psi$ and ($i=n$ or
$\rho,i{+}1\models \Box\psi$)\\
$\rho,i \models \psi_1 \until \psi_2$ & iff 
$\rho,i \models \psi_2$, or ($i\,{<}\,n$ and both\\
& \quad 
$\rho,i \models \psi_1$ and
$\rho,i{+}1\models \psi_1 \until \psi_2$)
\end{tabular}
\end{definition}

\noindent
\textbf{Verification problem.} We use $\LBC$ to express properties over the finite traces of a \mydds $\BB$. 
A run $\rho$ is a \emph{witness} for $\psi \in \LBC$ if $(i)$ $\rho$ ends in a final state of $\BB$ and $(ii)$ $\rho \models \psi$.

\begin{definition}[Verification task]
Given a \mydds $\BB$ and $\psi \in \LBC$, check whether there exists a witness $\rho$ for $\psi$ in $\BB$.
\end{definition}

\noindent If $\CC$ is closed under negation, one can model check $\psi$ by looking for a witness for $\neg \psi$, i.e., a counterexample.

Unsurprisingly, \mydds{s} can directly encode 2-counter Minsky machines, making the verification task undecidable.

\begin{remark}
\label{rem:undecidability}
It is undecidable to check whether there exists a witness for a property of the form $\Diamond b$ in a \mydds, for $b\in B$.
\end{remark}

\section{DDSAs with Finite Summary}
\label{sec:finite summary}

Instead of taming undecidability of verification by directly looking for decidable fragments, we introduce a \emph{semantic property} called \emph{\property}, and show that DDSAs with this property admit a faithful finite-state abstraction that preserves all properties expressible in our verification language. 
Throughout the section, we fix a
DDSA $\BB = \langle B, b_0, \AA, T, F, V, \alpha_0, guard\rangle$ and a finite constraint set $\CC$.
We first consider paths in $\BB$, called \emph{symbolic runs}:

\begin{definition}A \emph{symbolic run} $\sigma$ is a transition sequence
$\smash{b_0 \goto{a_1} b_1 \goto{a_2} \dots \goto{a_n} b_n}$ 
where $b_i\inn B$ and $a_i \inn \AA$;
it \emph{abstracts} any run of the form
$\smash{(b_0, \alpha_0) 
\goto{a_1} (b_1, \alpha_1)
\goto{a_2}}$
$\smash{\dots \goto{a_n} (b_n, \alpha_n)}$
i.e., a run with the same state and action sequence.
The prefix of $\sigma$ of $i$ steps is denoted $\sigma|_i$.
\end{definition}

\noindent
For instance, for the DDSA $\BB_1$ in \figref{examples} the sequence
$\m 1 \goto{\m a_1} \m 2 \goto{\m a_2} \m 1$ is a symbolic run. 
In this section, we aim to construct an abstract representation of the reachable configurations of $\BB$, where we capture a set of configurations
by a pair $(b,\varphi)$ of a system state $b \inn B$ and a formula $\varphi$
with free variables $V$ that describes the current state of the data. 
Our aim is to find a \emph{finite} set of such pairs that covers all reachable 
configurations while being precise enough to decide our verification task.
%
To that end, we next define the $\update$ function as a uniform way to express how the current state, captured by a formula $\phi$, changes by executing an action.

First, we define the \emph{transition formula} $\trans{a}$ of action $a$
as $\trans{a}(\vec V^r, \vec V^w)\:{=}\:
\guard(a) \wedge \bigwedge_{v\not\in write(a)} v^{w}\,{=}\,v^{r}$. 
Intuitively, this formula states the conditions on variables \emph{before and af\-ter} executing $a$: $\guard(a)$ must be true and the values of all variables that are not written are propagated by inertia. 
Note that $\trans{a}$ has free variables $\vec V^r$ and 
$\vec V^w$; for variable vectors $\vec X$ and $\vec Y$ of the same length,
let $\trans{a}(\vec X, \vec Y)$ be the formula obtained from $\trans a$ 
by replacing $\vec V^r$ by $\vec X$ and $\vec V^w$ by $\vec Y$.

\begin{definition}
\label{def:update}
For a formula $\phi$ with free variables $V$ and an action $a$, let 
$\update(\phi, a) = \exists \vec U. \phi(\vec U) \wedge \Delta_a(\vec U, \vec V)$,
where $\vec U$ is a variable vector of the same length as $\vec V$ such that
$U$ is disjoint from $V$ and variables in $\phi$, to avoid variable capture.
\end{definition}

For instance, for action $\m a_1$ in DDSA $\BB_1$ of \figref{examples},
$\Delta_{\m a_1} = (x^w\,{>}\,y^r) \wedge (y^w\,{=}\,y^r)$; and for 
$\phi = (x\,{>}\,0) \wedge(y\,{>}\,x)$ we get
$\update(\phi, \m a_1) = \exists x'\,y'. (x'\,{>}\,0) \wedge(y'\,{>}\,x') \wedge
(x\,{>}\,y') \wedge (y\,{=}\,y')$. Using quantifier elimination, we get an
equivalent, quantifier-free formula, for instance
$(y\,{>}\,0) \wedge(x\,{>}\,y)$.

A key notion for our approach are \emph{history constraints}:
formulas that sum up constraints collected along symbolic runs,
possibly in combination with additional constraints that are needed for verification
(i.e., constraints that occur in the property $\psi$ to be checked).
To express the latter, we consider \emph{verification constraint sequences} 
$\cseq C$ over constraint set $\CC$,
i.e., sequences
$\cseq C = \langle C_0,\dots, C_n\rangle$  of sets 
$C_i \subseteq \CC$.
A prefix $\langle C_0,\dots, C_m\rangle$ of $\cseq C$ is denoted by 
$\cseq C|_m$.
Moreover, we denote by $C_{\alpha_0} = \{v\,{=}\,\alpha_0(v) \mid v\inn V\}$
the set of \emph{initial constraints}, to capture in a formula the initial assignment.

\begin{definition}\label{def:history constraint}
For a symbolic run $\sigma\colon b_0 \goto{a_1} b_1 \goto{a_2} \dots \goto{a_n} b_n$,
and verification constraint sequence $\cseq C= \langle C_0,\dots, C_n\rangle$,
the \emph{history constraint} $\hist(\sigma, \cseq C)$ 
is inductively defined by setting
$\hist(\sigma, \cseq C) = \bigwedge (C_{\alpha_0} \cup C_0)$ if $n\,{=}\,0$, and 
$\hist(\sigma, \cseq C) = \update(\hist(\sigma|_{n-1}, \cseq C|_{n-1}), a_{n}) \wedge C_n$ if $n>0$.
\end{definition}

\noindent
Informally, the history constraint of a symbolic run is a formula that captures all variable constraints that must hold in the last state, i.e., it is a \emph{summary} of the symbolic run, taking into account
additional verification constraints $\cseq C$ that will become relevant in \secref{model checking}.
Note that symbolic runs may in fact feature a sequence of actions that is not executable
due to guard conditions. In these cases history constraints are unsatisfiable. 
For simplicity, in what follows we do not rule out these explicitly 
(as it does not affect our results), 
though it is possible and in fact done in our implementation.
We call $\hist(\sigma, \cseq C)$ a history constraint \emph{of} $\BB$ and $\CC$ 
if $\sigma$ is a symbolic run of $\BB$ and $\cseq C$ is a 
constraint sequence over $\CC$.
If no verification constraints are needed,
we write $\hist(\sigma)$ for $\hist(\sigma, \langle \emptyset, \dots, \emptyset\rangle)$.
\begin{example}
\label{exa:history constraints}
For $\BB_1$ in \figref{examples} with domain $\mathbb Q$
and $\alpha_0(x) = \alpha_0(y) = 0$,
let $\sigma_k$ be the (unique) symbolic run of $k$ steps,
e.g. $\sigma_2 \colon \m 1 \goto{\m a_1} \m 2 \goto{\m a_2} \m 1$.
We get the history constraints
\begin{footnotesize}
\begin{align*}
\hist(\sigma_0) &=x\,{=}\,0 \wedge y\,{=}\,0 &(\phi_0)\\
\hist(\sigma_1) &= 
 \exists x_0 y_0.\: x_0\,{=}\,0 \wedge y_0\,{=}\,0 \wedge x\,{>}\,y_0 \wedge y\,{=}\,y_0 \\[-.5ex]
&\equiv x\,{>}\,0 \wedge y\,{=}\,0 & (\phi_1)\\[-.5ex]
\hist(\sigma_2) &= \exists x_1 y_1 x_0 y_0.\: x_0\,{=}\,0 \wedge y_0\,{=}\,0 \wedge x_1\,{>}\,y_0 \wedge{}\\[-.5ex]
&\quad  y_1\,{=}\,y_0 \wedge y\,{>}\,x_1 \wedge x\,{=}\,x_1 \\
&\equiv x\,{>}\,0 \wedge y\,{>}\,x & (\phi_2)
\end{align*}
\end{footnotesize}
where $x_0$, $y_0$, $x_1$, $y_1$ are fresh variables.
The equivalence steps are obtained by simplification
and quantifier elimination. 
In a similar way, we get for
$\sigma_3 \colon \m 1 \goto{\m a_1} \m 2 \goto{\m a_2} \m 1 \goto{\m a_1} \m 2 $
the constraint $\hist(\sigma_3) \equiv (y\,{>}\,0) \wedge (x\,{>}\,y)$, and
for
$\sigma_4 \colon \m 1 \goto{\m a_1} \m 2 \goto{\m a_2} \m 1 \goto{\m a_1} \m 2 \goto{\m a_2} \m 1$ we get $\hist(\sigma_4) \equiv (x\,{>}\,0) \wedge (y\,{>}\,x) $. 
The fact that $\hist(\sigma_2)$ and $\hist(\sigma_4)$ are equivalent reflects that 
$\sigma_2$ and $\sigma_4$ are equivalent in our finite-state abstraction.
\end{example}

\noindent
Next we relate history constraints and
assignments in runs.

\begin{lemma}
\label{lem:abstraction}
For any symbolic run $\sigma$ of length $n$ and $\cseq C= \langle C_0,\dots, C_n\rangle$,
$\smash[t]{\hist(\sigma,\cseq C)}$
is satisfied by assignment $\alpha$ iff there is a run
$\smash{(b_0, \alpha_0) \goto{a_1} \dots 
\goto{a_n} (b_n, \alpha_n)}$
that is abstracted by $\sigma$ such that $\alpha=\alpha_n$ and
$\alpha_i \models C_i$ for all $i$, $0\leq i \leq n$.
\end{lemma}

\noindent
This shows that history constraints faithfully summarize accumulated constraints 
in symbolic runs, and their satisfying assignments correspond to the results
of actual runs.
Both directions 
are proven by straightforward induction proofs.
For instance, \lemref{abstraction} states that since the assignment $\alpha(x)=9$,
$\alpha(y) = 7$ satisfies $\hist(\sigma_3)$ in \exaref{history constraints},
there is a run abstracted by $\sigma_3$ ending with this assignment. This is true, e.g.,
for
$(\m 1, \domino{x}{0}{y}{0}) \goto{\m{a}_1}
(\m 2, \domino{x}{1}{y}{0}) \goto{\m{a}_2}
(\m 1, \domino{x}{1}{y}{7})\goto{\m{a}_1}
(\m 2, \domino{x}{9}{y}{7})$.

Our \property property will express that all 
(infinitely many) symbolic runs can be faithfully described by a
\emph{finite} set of states $(b,\varphi)$
of a system state $b\in B$ and a formula $\varphi$ that summarizes accumulated constraints.
To that end, we first define a \emph{history set} as a 
set of such states that contains a representative for every history constraint:

\begin{definition}
\label{def:history set}
A \emph{history set} $\Phi$ for $\BB,\CC$ is a 
set of pairs $(b,\phi)$ of $b\inn B$ and a formula $\phi$ such that
for every history constraint $\hist(\sigma,\cseq C)$ of $\BB,\CC$
where $\sigma$ has final state $b$,
there is a $(b, \phi) \inn \Phi$ with $\hist(\sigma,\cseq C)\,{\equiv}\, \phi$ 
and $\Phi$ contains no other pairs.
\end{definition}

\noindent
The next result turns out to be convenient in the sequel to characterize history sets:

\begin{lemma}
\label{lem:history set}
$\Phi$ is a history set iff
(1) for all $C \subseteq \CC$, there is some $(b_0,\phi_0) \in \Phi$ 
such that $\phi_0 \equiv \bigwedge (C_{\alpha_0} \cup C)$, and
(2) for all $(b,\phi)\in\Phi$, $b \goto{a} b'$, and $C \subseteq \CC$,
there is some $(b',\phi')\in \Phi$ such that $\phi' \equiv \update(\phi,a) \wedge C$.
\end{lemma}

\noindent
We will show that
some of the DDSA classes that we consider in this paper admit a \emph{finite} history set---systems with monotonicity constraints and bounded lookback---and this feature is sufficient to decide our verification problem.
For other systems (e.g., gap-constraint systems) it is not possible to find
finite history sets. However, we will prove that the verification problem is still 
decidable if the more liberal property of \emph{finite summary} holds.
Basically, this property expresses that there exists a suitable equivalence relation 
$\equivBC$ such that the \emph{quotient} of a history set with respect $\equivBC$ is finite.
Here, $\equivBC$ is considered \emph{suitable} if it is preserved under steps of $\BB$ and implies equisatisfiability; for practicality we also require decidability.
These requirements are made formal in the following definition.

\begin{definition}
\label{def:finite summary}
A \emph{summary} for $(\BB,\CC)$
is a pair $(\Phi, \equivBC)$ of a history set $\Phi$ for $\BB$, $\CC$
and equivalence relation $\equivBC$ s.t.
\begin{compactenum}[(1)]
\item $\equivBC$ contains $\equiv$ on $\Phi$ and is decidable,
\item for all $(b,\phi),(b,\psi) \in \Phi$ such that  $\phi \equivBC \psi$,
\begin{inparaenum}
\item $\phi$ and $\psi$ are equisatisfiable, and 
\item for all 
transitions $b \goto{a} b'$ and $C \subseteq \CC$,
$[\update(\phi,a) \wedge C] \equivBC [\update(\psi,a) \wedge C]$.
\end{inparaenum}
\end{compactenum}
We say that $(\BB,\CC)$ \emph{has finite summary} if it admits a summary 
$(\Phi, {\equivBC})$ where
$\equivBC$ has finitely many equivalence classes.
\end{definition}

\noindent
Here, $[\cdot]$ is a \emph{representative} function for the given history set:
if for a pair $(b,\psi)$ there is some $(b,\phi) \in \Phi$
with $\psi \equiv \phi$, we can assume that $[\psi]$ is such a formula $\phi$.
A formula equivalent to $\update(\phi,a) \wedge C$ exists in $\Phi$
because of \lemref{history set}.

Intuitively, a DDSA has \property if it admits a finite-state abstraction that is
expressive enough to account for all possible evolutions of $\BB$ and properties in $\CC$.
%
We next show that $(\BB,\CC)$ \hasproperty if it has a finite history set, so one can 
pick $\equiv$ as equivalence relation.

\begin{lemma}
\label{lem:simple finite summary}
If $\BB$ and $\CC$ admit a finite history set $\Phi$
then $(\BB,\CC)$ has finite summary $(\Phi, \equiv)$.
\end{lemma}
\begin{proof}
\defref{finite summary} (1) follows from decidability of linear arithmetic and finiteness of $\Phi$. For \defref{finite summary} (2), we have that (a)
$\phi\,{\equiv}\,\psi$ implies equisatisfiability, and (b)
we can write 
$\update(\phi,a) \wedge \bigwedge C = \exists \vec U. \phi(\vec U) \wedge \chi$ 
and
$\update(\psi,a) \wedge \bigwedge C = \exists \vec U. \psi(\vec U) \wedge \chi$
for some $\chi$, and these two formulas
are clearly again equivalent as $\phi\,{\equiv}\,\psi$.
\end{proof}

\begin{example}
\label{exa:finite summary}
Continuing \exaref{history constraints}, it can be shown that
$\hist(\sigma_{2i}) \equiv \varphi_2$ and 
$\hist(\sigma_{2i+1}) \equiv \varphi_3$ for all $i > 0$.
Thus the set $\Phi = \{(\mathsf 1, \varphi_0), (\mathsf 2,\varphi_1), (\mathsf 1, \varphi_2), (\mathsf 2, \varphi_3)\}$ is a finite history set, and by \lemref{simple finite summary}
the tuple $(\Phi, \equiv)$ is a finite summary 
for $(\BB_1, \emptyset)$.
It can be visualized in a constraint graph, as done in~\cite{FLM19}:

\centering
\begin{tikzpicture}[node distance=24mm,>=stealth']
\tikzstyle{node}=[rectangle, rounded corners, inner sep=3pt, draw, scale=.8, thick]
\node[node] (1) {$\m 1,\phi_0$};
\node[node, right of=1] (2) {$\m 2, \phi_1$};
\node[node, right of=2] (3) {$\m 1, \phi_2$};
\node[node, right of=3] (4) {$\m 2, \phi_3$};
\draw[edge,->] ($(1) + (-.8,0)$) -- (1);
\draw[edge,->] (1) -- node[above,action]{$\m a_1$} (2);
\draw[edge,->] (2) -- node[above,action]{$\m a_2$} (3);
\draw[edge,->] (3) to[bend left=10] node[above,action]{$\m a_1$} (4);
\draw[edge,->] (4) to[bend left=10] node[below,action]{$\m a_2$} (3);
\end{tikzpicture}
\end{example}

We conclude this section with another example where the history set is not finite
but a finite summary can be found.
\begin{example}
\label{exa:gap order}
Consider the DDSA $\BB_3$, 
and let $\sigma_k$ be the symbolic run of $k$ steps (there is only one).
We have e.g.
$\hist(\sigma_1) \equiv (x\,{-}\,y\,{\geq}\,2) \wedge (y \eqn 0)$, and
$\hist(\sigma_2) \equiv (x\,{\geq}\,2) \wedge (y\,{\geq}\,3)$.
In general, we obtain $\hist(\sigma_{2i}) \equiv (x\,{\geq}\,3i-1) \wedge (y\,{\geq}\,3i)$
and $\hist(\sigma_{2i+1}) \equiv (x\,{-}\,y\,{\geq}\,2) \wedge (y\,{\geq}\,3i)$
for all $i\geq 1$.
Since $\hist(\sigma_i) \not \equiv \hist(\sigma_j)$ for $i\neq j$, the history set 
$\Phi = \{\hist(\sigma_i) \mid i \geq 0\}$ is not finite.
However, in \secref{conditions} (subsection on gap-order constraints) we will show that $\BB_3$ admits a finite summary $(\Phi, \equivGC)$, where $\equivGC$ is the \emph{cutoff equivalence relation} that considers
formulas equivalent if they are syntactically equal after replacing all constants larger
than some bound $K$ by $K$ itself.
\end{example}

\section{Checking the Existence of Witnesses}
\label{sec:model checking}

In order to express the requirements on a run of a DDSA $\BB$ to satisfy an LTL$_f$ formula $\psi$,
we next define a nondeterministic automaton (NFA) $\NFApsi$.
Then we combine $\NFApsi$ with $\BB$ in a kind of product construction to check for the existence of witnesses for $\psi$.

To get the NFA, we perform a similar preprocessing step as in~\cite{LFM20}, and replace first all occurrences of subformulas $\langle a\rangle \psi'$ in $\psi$ by
$\langle \cdot\rangle (a \wedge \psi')$, adding a new proposition symbol for each action.
For a run $\rho$ of length $n$, we thus write $\rho,i \models a$ if 
$0\,{<}\,i\,{<}\,n$
and $\rho, i-1 \models \langle a \rangle \top$.
This modification allows us to consider fewer cases in the constructions and proofs below.

Technically, given $\psi \in \LBC$ we build the NFA $\NFApsi = (Q, \Sigma, \varrho, q_0, Q_F)$, where:
\begin{inparaenum}[\it (i)]
\item the set $Q$ of states is a set of quoted formulas; 
\item $\Sigma\,{=}\, 2^{S}$ is the alphabet, where 
$S = B\,{\cup}\,\AA\,{\cup}\,\CC $;
\item $\varrho \subseteq Q \times \Sigma \times Q$ is the transition relation; 
\item $q_0 \in Q$ is the initial state;
\item $Q_F\subseteq Q$ is the set of final states. 
\end{inparaenum}
Following~\cite{GMM14}, we define $\varrho$ using an auxiliary function $\delta$
and a new proposition $\last$ that marks the last element of the trace.
The input of $\delta$ is a 
(quoted) formula $\psi \in \LBC \cup \{\top,\bot\}$, and its output a set of tuples
$(\inquotes{\psi'},\varsigma)$ where $\psi'$ has the same type as $\psi$ and 
$\varsigma \in 2^{S\cup \{\last, \neg \last\}}$.
For two sets of such tuples $R_1$, $R_2$, and $\odot$ either $\wedge$ or $\vee$, let
$R_1 \odot R_2 = \{ (\inquotes{\psi_1 \odot \psi_2}, \varsigma_1 \cup \varsigma_2) \mid (\inquotes{\psi_1}, \varsigma_1) \inn R_1, (\inquotes{\psi_2}, \varsigma_2) \inn R_2 \}$, where we simplify $\psi_1 \odot \psi_2$ if possible. 
The function $\delta$ is as follows:

\noindent
\begin{tabular}{@{~}l@{~}l}
$\delta(\inquotes{\top})$ &= 
  $\{(\inquotes{\top},\emptyset)\}
  \text{ and }\delta(\inquotes{\bot}) = \{(\inquotes{\bot},\emptyset)\}$\\
$\delta(\inquotes{p})$ &= 
  $\{(\inquotes{\top},\{p\}),(\inquotes{\bot},\emptyset)\} 
  \text{ if $p \in \CC \cup B \cup \mathcal A$}$\\
$\delta(\inquotes{\psi_1 \vee \psi_2})$ &= 
  $\delta(\inquotes{\psi_1}) \vee \delta(\inquotes{\psi_2})$ \\
$\delta(\inquotes{\psi_1 \wedge \psi_2})$ &= 
  $\delta(\inquotes{\psi_1}) \wedge \delta(\inquotes{\psi_2})$ \\
$\delta(\inquotes{\langle\cdot\rangle \psi})$ &= 
  $\{(\inquotes{\psi},\{\neg \last\}), (\inquotes{\bot}, \{\last\})\}$ \\
$\delta(\inquotes{\Diamond \psi})$ &=
  $\delta(\inquotes{\psi}) \vee \delta(\inquotes{\langle\cdot\rangle\Diamond \psi})$ \\
$\delta(\inquotes{\Box \psi})$ &=
  $\delta(\inquotes{\psi}) \wedge (\delta(\inquotes{\langle\cdot\rangle\Box \psi}) \vee \delta_\lambda)$ \\
$\delta(\inquotes{\psi_1 \until \psi_2})$ &=
  $\delta(\inquotes{\psi_2}) \vee (\delta(\inquotes{\psi_1})
  \wedge \delta(\inquotes{\langle\cdot\rangle(\psi_1 \until \psi_2)}))$
\end{tabular}

\noindent
where $\delta_\lambda$ abbreviates 
$\smash{\{(\inquotes{\top},\{ \lambda\}),(\inquotes{\bot},\{\neg\lambda\})\}}$.
While the symbol $\last$ is needed for the construction, we can omit it
from the NFA,
and define $\NFApsi$ as follows:
\begin{definition}
\label{def:NFA}
Given a formula $\psi \inn \LBC$, let the NFA 
$\NFApsi\,{=}\,(Q, \Sigma, \varrho, q_0, \{q_f, q_e\})$
be given by
$q_0\,{=}\,\inquotes{\psi}$,
$q_f\,{=}\,\inquotes{\top}$ and $q_e$ is an additional final state,
and
$Q$, $\varrho$ are the smallest sets such that $q_0, q_f, q_e \in Q$ and whenever 
$q\in Q\setminus\{q_e\}$ and $(q', \varsigma)\in \delta(q)$
such that $\{\last,\neg \last\} \not\subseteq \varsigma$
then $q'\in Q$ and 
\begin{compactenum}[(i)]
\item if $\last \not\in \varsigma$ then
$(q, \varsigma \setminus\{\last, \neg \last\}, q') \in \varrho$, and
\item
if $\last \in \varsigma$ and $q' = \inquotes{\top}$ then
$(q, \varsigma \setminus\{\last, \neg \last\}, q_e) \in \varrho$.
\end{compactenum}
\end{definition}
This construction is similar to the one by \cite{GMM14}, but
reflects that our verification language does not include negation.
In fact it can be seen as a relaxation, in that 
$\delta(\inquotes{p})$ contains $(\bot,\emptyset)$ 
rather than $(\bot,\{\neg p\})$, for any atom $p$.
In this way, $\NFApsi$ cannot explicitly require atoms to be false;
instead, the transition labels intuitively state \emph{minimal requirements}
for $\psi$ to hold.

\begin{example}
\label{exa:simplenfa}
Let $\psi = \Diamond c$ for a constraint $c = (y > 5)$.
By the definition of $\delta$, we have
$\delta(\inquotes{\Diamond c}) = 
\delta(\inquotes{c}) \vee \delta(\inquotes{\langle\cdot\rangle \Diamond c}) = \{(\inquotes{\top},\{c\}),(\inquotes{\bot},\emptyset)\} \vee \{(\inquotes{\Diamond c},\{\neg \last\}), (\inquotes{\bot}, \{\last\})\} = 
\{(\inquotes{\top},\{c, \neg  \last\}), (\inquotes{\top},\{c,\last\}),
(\inquotes{\Diamond c},\{\neg \last\}), (\inquotes{\bot}, \{\last\}) \}$, so
the automaton $\NFApsi$ is as follows:\\[.5ex]
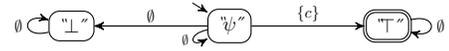
\tikz[node distance=28mm,>=stealth']{
\tikzstyle{formula}=[scale=.75, rectangle, rounded corners, inner sep=4pt, draw]
\node[formula] (psi) {$\inquotes{\psi}$};
\node[formula, left of=psi] (bot) {$\inquotes{\bot}$};
\node[formula, right of=psi, double] (top) {$\inquotes{\top}$};
\node at (-4,0) {}; 
\draw[->] (psi) to 
 node[scale=.6, above, yshift=-1pt]{$\{c\}$}
 (top);
\draw[->] (psi) to node[scale=.6, above, yshift=-1pt]{$\emptyset$} (bot);
\draw[->] (psi) to (bot);
\draw[->] (top) to[loop right, looseness=6] node[scale=.6, right]{$\emptyset$} (top);
\draw[->] (bot) to[loop left, looseness=6] node[scale=.6, left]{$\emptyset$} (bot);
\draw[->] (psi) to[loop, out=190,in=210, looseness=6] node[scale=.5, left]{$\emptyset$}(psi);
\draw[->] ($(psi.160) + (-.2,.2)$) to (psi);
}\\
Due to our relaxation, the self-loop on $\inquotes{\psi}$ is
labeled $\emptyset$ rather than $\{ \neg c \}$, but nonetheless
$\NFApsi$ works as expected:
if $c$ is true an accepting path exists, and
if $c$ is false no further possibilities arise.
\end{example}

To express correctness of $\NFApsi$,
we need some notions of consistency to express that a word $w$
and a symbolic run are not contradictory with respect to actions and states.
First,  we call a symbol $\varsigma\,{\in}\,\Sigma$ \emph{consistent} with transition 
$b \goto{a} b'$ if $\varsigma$ is disjoint from  $\AA\,{\setminus}\,\{a\}$ and 
$B\,{\setminus}\,\{b'\}$, namely if it contains no action symbol other than $a$ nor state symbol other than $b'$. 
Let $\constr(\varsigma) = \varsigma \cap \CC$. 
\noindent
\begin{definition}
A word $w=\varsigma_0 \varsigma_1 \cdots \varsigma_n\in \Sigma^*$ is \emph{consistent} with
\begin{inparaenum}
\item[(a)] 
a symbolic run
$\sigma \colon b_0 \goto{a_1} b_2 \goto{a_2} \dots \goto{a_n} b_n$
if
$\varsigma_0$ is disjoint from $B \setminus \{b_0\}$, and
$\varsigma_i$ is consistent with $b_{i-1} \goto{a_{i}} b_{i}$ for $0\,{<}\,i\,{\leq}\,n$.
\item[(b)]
a run $\rho$ if it is consistent with
the abstraction $\sigma$ of $\rho$
and $\alpha_i$ satisfies $\bigwedge \constr(\varsigma_i)$,
where $\rho_i = (b_i,\alpha_i)$.
\end{inparaenum}
\end{definition}

\noindent
These notions allow us to express correctness of $\NFApsi$:

\begin{lemma}
\label{lem:NFA acceptance}
$\NFApsi$ accepts a word that is consistent with a run $\rho$ iff 
$\rho \models \psi$.
\end{lemma}

\paragraph{Product construction.}
To check the existence of a witness for $\psi$ in \mydds $\BB$,
we combine $\NFApsi$ with $\BB$ to a cross-product automaton $\smash{\NN^\psi_\BB}$,
exploiting the notions from \secref{finite summary}.

First, for technical reasons we add a dummy initial state $b_0'$ to $\BB$ and update its states to $B' = B\cup\{b_0'\}$ and its transitions to $T' = T \cup \{(b_0', a_0, b_0)\}$ for a fresh action $a_0$ with $\guard(a_0) = \top$. 
We call the resulting \mydds  $\BB'$.

\begin{definition}
\label{def:product construction}
Let  $\BB' = \langle B', b_0', \AA, T', F, V, \alpha_0, \guard\rangle$ as above, 
$\CC$ a constraint set, and $(\Phi, \equivBC)$ a summary for $(\BB,\CC)$.
For a formula $\psi \in \LBC$ and $\NFApsi$ as above,
the \emph{product automaton} 
$\smash{\NN^\psi_\BB=(P, \Sigma, R, p_0, P_F)}$ is as follows:
\begin{compactitem}[$\bullet$]
\item
States in $P$ are triples $(b, q, \varphi)$ s.t. $b\,{\in}\,B'$,
$q\,{\in}\, Q$, $\varphi\in \Phi$;
\item
The initial state is $p_0=(b_0', q_0, \bigwedge C_{\alpha_0})$; 
\item
There is a transition  $(b,q, \varphi) \goto{a} (b',q', \varphi')$ in $R$ iff 
$b \goto{a} b'$ in $T'$, 
there is some $\varsigma \in \Sigma$ s.t. $q \goto{\varsigma} q'$ in 
$\NFApsi$, and
\begin{compactitem}[$-$]
\item 
formula $\chi = \update(\varphi, a) \wedge constr(\varsigma)$ 
is satisfiable, and $\varphi' \equivBC \chi$;
in this way, $\chi$ captures all current constraints that are either inherited from $\BB$ or stem from the transition of $\NFApsi$, given by $constr(\varsigma)$,
\item
$\varsigma$ is consistent with $b \goto{a} b'$, and
\item
$(b',q', \varphi') \in P_F$ iff $b'\in F$, $q'\in Q_F$.
\end{compactitem}
\end{compactitem}
\end{definition}

\noindent
Note that $R$ is well-defined in the sense that for every such formula
$\chi$ above, some $\phi'$ with $\varphi' \equivBC \chi$ and $(b',\phi') \in \Phi$ exists, because $\Phi$
is a history set (cf. \lemref{history set}).
Thus, if $(\Phi, \sim)$ is a \emph{finite} summary, the construction
in \defref{product construction} terminates.
The next result states properties of the 
product construction, the induction proofs of both directions are straightforward.
\begin{lemma}
\label{lem:PC sat}
Let $\sigma$ be a symbolic run of $\BB$ and $w \in \Sigma^*$.
There is a path $\pi$ with $\sigma = \sigma(\pi)$ to a node $(b,q,\phi)$
in $\smash{\NN_\BB^\psi}$ such that $\varphi \equivBC \hist(\sigma,w)$ iff 
$w$ is accepted by $\NN_\psi$, consistent with $\sigma$,
and $\hist(\sigma,w)$ is satisfiable.
\end{lemma}

\noindent
We next state our main result, where $h(\sigma,w)$ denotes 
$h(\sigma,\langle \constr(\varsigma_0), \dots, \constr(\varsigma_n)\rangle)$
for word $w = \varsigma_0 \cdots \varsigma_{n}$. 

\begin{theorem}
\label{thm:model checking}
Let $\psi \in \LBC$.
The language of $\smash{\NN_\BB^\psi}$ is non-empty 
iff there is a run of $\BB$ that is a witness for $\psi$.
\end{theorem}
\begin{proof}
($\Longrightarrow$)
Let $\pi$ be a path to a final state $p_f$ in $\smash{\NN_\BB^\psi}$.
By \lemref{PC sat}, there is an accepting transition sequence
in $\NN_\psi$ labeled $w=\varsigma_0 \dots\varsigma_n$, and a
symbolic run $\sigma(\pi)\colon b_0 \goto{a_1} b_1 \goto{*} b_n$
such that $\hist(\sigma(\pi),w)$ is satisfiable by some $\alpha$, and
$\varsigma_i$ is consistent with $b_{i-1} \goto{a_i} b_i$, 
for all $i$, so $w$ is consistent with $\sigma$.
By \lemref{abstraction} (2), there is a run 
$\rho\colon (b_0, \alpha_0) \goto{a_1} \dots \goto{a_n} (b_n, \alpha_n)$
abstracted by $\sigma$
such that $\alpha=\alpha_n$ and 
$\alpha_i  \models \bigwedge \constr(\varsigma_i)$ for all $i$, $0\,{\leq}\,i\,{\leq}\,n$. Thus $w$ is consistent with 
$\rho$, and it follows from \lemref{NFA acceptance} that $\rho$ is a witness.
($\Longleftarrow$) Let $\rho$ be a witness for $\psi$, and $\sigma$ its abstraction.
By \lemref{NFA acceptance}, $\NN_\psi$ accepts a word $w$ that is consistent with $\rho$.
Consistency of $w$ with $\rho$ implies that
$w$ is also consistent with $\sigma$, and that $\alpha_i \models \bigwedge \constr(\varsigma_i)$ for all $i$,
$0\,{\leq}\,i\,{\leq}\,n$. Thus
$\alpha_n$ satisfies $\hist(\sigma,w)$ by \lemref{abstraction} (1).
By \lemref{PC sat}, the run of $\NN_\psi$ labeled $w$ and
the symbolic run $\sigma$ give rise to a path $\pi$ in $\smash{\NN_\BB^\psi}$
such that $w$ is consistent with $\sigma$.
As $\NN_\psi$ accepts $w$, and the last state of $\sigma$ is final, also $\pi$ is accepting.
\end{proof}

We illustrate the product construction as well as the witness extraction on a simple example.
\begin{example}
Consider the DDSA $\BB_1$ from \figref{examples} and a formula $\psi = \Diamond (y > 5)$.
We use the NFA $\NN_\psi$ obtained in \exaref{simplenfa}, removing the leftmost deadlock state for compactness. Then, the product automaton is as follows:\\
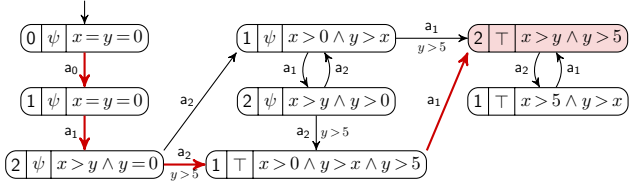
\begin{tikzpicture}[node distance=12mm,>=stealth']
\tikzstyle{node} = [draw,rectangle split, rectangle split parts=3,rectangle split horizontal, rectangle split draw splits=true, inner sep=3pt, scale=.7, rounded corners]
\tikzstyle{goto} = [->]
\tikzstyle{accepting path} = [red!80!black, line width=.8pt]
\tikzstyle{accepting state} = [fill=red!80!black!15]
\tikzstyle{action}=[scale=.6, black]
\tikzstyle{constr}=[scale=.5, black]
\node[node] (0)  {\pcnode{$\m 0$}{$\psi$}{$x\,{=}\,y\,{=}\,0$}};
\node[node, below of=0] (1) {\pcnode{$\m 1$}{$\psi$}{$x\,{=}\,y\,{=}\,0$}};
\node[node, below of=1] (2) {\pcnode{$\m 2$}{$\psi$}{$x\,{>}\,y \wedge y\,{=}\,0$}};
\node[node, right of=0, xshift=32mm] (3) {\pcnode{$\m 1$}{$\psi$}{$x\,{>}\,0 \wedge y\,{>}\,x$}};
\node[node, below of=3] (4) {\pcnode{$\m 2$}{$\psi$}{$x\,{>}\,y \wedge y\,{>}\,0$}};
\node[node, below of=4] (5) {\pcnode{$\m 1$}{$\top$}{$x\,{>}\,0 \wedge y\,{>}\,x \wedge y\,{>}\,5$}};
\node[node, right of=3, xshift=32mm, accepting state] (6) {\pcnode{$\m 2$}{$\top$}{$x\,{>}\,y \wedge y\,{>}\,5$}};
\node[node, below of=6] (7) {\pcnode{$\m 1$}{$\top$}{$x\,{>}\,5 \wedge y\,{>}\,x$}};
\draw[goto] ($(0) + (0,.5)$) -- (0);
\draw[goto,accepting path] (0) to node[action, left]{$\m{a_0}$} (1);
\draw[goto,accepting path] (1) to node[action, left]{$\m{a_1}$} (2);
\draw[goto] (2.north east) to node[action, left]{$\m{a_2}$} (3.south west);
\draw[goto,bend right] (3) to node[action, left]{$\m{a_1}$} (4);
\draw[goto,bend right] (4) to node[action, right]{$\m{a_2}$} (3);
\draw[goto] (4) to node[action, left]{$\m{a_2}$} node[right, constr]{$y\,{>}\,5$} (5);
\draw[goto,accepting path] (2) to node[action, above]{$\m{a_2}$} node[below, constr]{$y\,{>}\,5$} (5);
\draw[goto] (3) to node[action, above]{$\m{a_1}$} node[below, constr]{$y\,{>}\,5$} (6);
\draw[goto,accepting path] (5.north east) to node[action, left]{$\m{a_1}$} (6.south west);
\draw[goto,bend right] (6) to node[action, left]{$\m{a_2}$} (7);
\draw[goto,bend right] (7) to node[action, right]{$\m{a_1}$} (6);
\end{tikzpicture}
Since $\NN_\BB^\psi$ has a final state (shown shaded),
by \thmref{model checking} a witness for $\psi$ exists.
We obtain a witness from the accepting path drawn in red:
This path corresponds to the word 
$w = \langle\emptyset\:\emptyset\:\{y\,{>}\,5\}\:\emptyset\rangle \in \Sigma^*$ accepted by $\NN_\psi$,
and the symbolic run 
$\sigma\colon \m 1 \goto{\m a_1} \m 2 \goto{\m a_2} \m 1 \goto{\m a_1} \m 2$
of $\BB_1$.
The formula $\phi$ in the final state is satisfiable and equivalent to $\hist(\sigma, w)$, 
and for any satisfying assignment of $\phi$ we obtain a witness run for $\psi$ according to \lemref{abstraction}. For instance for $\alpha(x)=9$, $\alpha(y)=7$, one possible solution is
$(\m 1, \domino{x}{0}{y}{0}) \goto{\m{a}_1}
(\m 2, \domino{x}{1}{y}{0}) \goto{\m{a}_2}
(\m 1, \domino{x}{1}{y}{7})\goto{\m{a}_1}
(\m 2, \domino{x}{9}{y}{7})$.
\end{example}

\noindent

%
%
\savespace{
From \thmsref{model checking}{finite automaton} it follows that 
\mydds{s} with \property allow LTL$_f$ model checking in the following sense:
\begin{corollary}
\label{cor:decidability}
If $\BB$ has \property w.r.t.~to $\CC$ and a decidable relation $\equivBC$, it is decidable
whether $\BB$ admits a witness for $\psi \in \LBC$.
\end{corollary}
}

\section{Conditions for Finite Summary}
\label{sec:conditions}

Thanks to \thmref{model checking} we can check for witnesses over \mydds{s} that \haveproperty. Unfortunately, however:
\begin{lemma}
The finite summary property is undecidable.  
\end{lemma}

\begin{proof}[\textit{Proof (sketch)}]
 Consider a \mydds $\BB$ encoding a Minsky machine. $\BB$ \hasproperty iff the counter configurations are bounded, which is undecidable. 
\end{proof}

\noindent
In this section we identify relevant, sufficient conditions for \property.
In short, these are the following
(the equivalence relation used in \defref{finite summary} is given in parentheses).
\begin{compactenum}[\bfseries {C}1:]
\item \label{cri:mcs}
\makebox[61mm][l]{$\BB$ is over \emph{monotonicity constraints}
($\equiv$)} \thmref{mc}
\item \label{cri:gcs}
\makebox[61mm][l]{$\BB$ is over \emph{gap-order constraints} ($\equivGC$)}
\thmref{gap order}
\item \label{cri:feedback free}
\makebox[61mm][l]{$\BB$ is \emph{feedback free} ($\equiv$)}
\thmref{feedback free}
\item \label{cri:bounded lookback}
\makebox[61mm][l]{$\BB$ has \emph{bounded lookback} ($\equiv$)}
\thmref{bounded lookback}
\end{compactenum}

\smallskip
\noindent
While \criref{mcs} and \criref{gcs} restrict the constraint language,
\criref{feedback free} and \criref{bounded lookback} restrict the control
flow (i.e., the shape of the \mydds).
\criref{bounded lookback} generalizes \criref{feedback free} as well as the case where $\BB$ is acyclic.
\smallskip

\noindent
Before  explaining these conditions, we point out that the \myddss in~\figref{examples} \haveproperty.
$\BB_1$ can be seen as a monotonicity constraint (MC) system over $\mathbb Q$,
or a gap-order constraint (GC) system over $\mathbb Z$, 
but \criref{feedback free} and 
\criref{bounded lookback} do not apply.
$\BB_2$ is feedback free and it can be shown to have also 2-bounded lookback.
$\BB_3$ is a GC system over $\mathbb Z$ but no other condition applies.
$\BB_4$ models a shopping process where two products with prices $a$ and $b$
are chosen by a customer, and the sum is computed in the variable $s$. 
$\BB_4$ has 3-bounded lookback, but due to the self-reference of $s$ in
$s^w = s^r + b$, \criref{feedback free} does not apply, 
and neither do \criref{mcs} or \criref{gcs}.


\smallskip
\noindent
\textbf{Monotonicity constraints} (MCs) restrict \defref{constraints} as follows: MCs over variables $V$ and domain $D$ have the form $p \odot q$ where $p,q\in {D\,{\cup}\,V}$
and $\odot$ is one of $=, \neq, \leq, <, \geq$, or $>$.
For MCs, we consider $D$ to be $\mathbb R$ or $\mathbb Q$.
An \emph{MC-formula} is a boolean formula whose atoms are MCs.
A DDSA is an $\MC$-DDSA whose guards are conjunctions of MCs.

It is known that if $\phi$ is an MC-formula over constants $\Const$ and variables 
$V \cup \{x\}$,
then for a formula $\exists x.\, \phi$, we can find a
formula $\phi' \equiv \exists x.\, \phi$ such that $\phi'$ is an MC-formula over constants 
$\Const$ and variables $V$, using a quantifier elimination procedure \'{a} la Fourier-Motzkin ~\cite[Sec. 5.4]{KS16}. 
In particular the set of constants $\Const$ remains the same.
This fact is crucial for the next result:

\begin{theorem}\label{thm:mc}
If $\BB$ is an $\MC$-DDSA and $\CC$ a set of MCs then
$(\BB,\CC)$ \hasproperty. 
\end{theorem}
\begin{proof}
Let $\Const$ be the set of constants in $\CC$, ${\alpha_0}$, and guards of $\BB$,
and $\MC_\Const$ the set of quantifier-free formulas whose atoms are MCs over $V$, $\Const$, so $\MC_\Const$ is finite up to equivalence.
We use \lemref{history set} to show that $\Phi_\MC := B \times \MC_\Const$ is a finite history set.
First, for all $C\,{\subseteq}\,\CC$,
$\bigwedge (C_{\alpha_0} \cup C) \in \MC_\Const$.
If $(b,\phi) \in \Phi_\MC$ and $b \goto{a} b'$ then
$\update(\phi,a) \wedge C = 
\exists \vec U. \phi(\vec U) \wedge \chi$ for some  MC-formula $\chi \in \MC_\Const$. From quantifier elimination one obtains
some $\phi'$ in $\MC_\Const$ such that $\phi' \equiv \exists \vec U. \phi(\vec U) \wedge \chi$.
Thus $(\Phi_\MC,\equiv)$ is a finite summary by \lemref{simple finite summary}.
\end{proof}

\noindent
This result explains why the history set in \exaref{finite summary} is finite.
MCs over $\mathbb Q$ and $\mathbb R$ are closed under negation, so
\thmsref{model checking}{mc} imply decidability of LTL$_f$ model checking.
Note that the proof of \thmref{mc} fails for domain $\mathbb Z$,
as MCs over $\mathbb Z$
are \emph{not} closed under quantifier elimination. Instead, they are
covered by gap-order constraints, discussed next.
\smallskip


\smallskip
\noindent
\textbf{Gap-order constraints.}
Let $X$ be a set of variables and $\Const \subseteq \mathbb Z$
a finite set of constants such that $0 \in \Const$.
A \emph{gap-order constraint} (GC) over $X$ and $\Const$
restricts \defref{constraints} to constraints of the form
$x - y \geq k$ for $x,y \in X\cup \Const$ and $k\inn \mathbb N$.
We call a \emph{GC-formula} a quantifier-free formula whose atoms are GCs,
and a GC-DDSA a DDSA where all guards are conjunctions of GCs. 
GC-DDSAs are known to generalize MC-DDSAs over $\mathbb Z$~\cite{BP14}:
for instance, $x=3$ is expressible by $x\,{-}\,3 \geq 0 \wedge 3\,{-}\,x \geq 0$.
However, it is known that relaxing the GC definition to allow also
$x - y \leq k$  (or $k<0$ in $x - y \geq k$)
renders reachability in GC-DDSAs undecidable~\cite{BP14}.

In order to show that GC-DDSAs allow for a finite summary, we use the concept
of a bounded approximation:
Given the set of constants $\Const$ and $K := max\{|c - c'| + 1 \mid c, c'\in \Const\}$, 
 the \emph{$K$-bounded approximation} $\cutoff[K]{\phi}$ of a GC-formula $\phi$
is obtained from $\phi$ by replacing all constraints $x\,{-}\,y \geq k$ where 
$k \geq K$ by $x - y \geq K$.
The next lemma rephrases \cite[Props. 6 and 7]{BP14}:
\begin{lemma}
\label{lem:cutoff}
\begin{inparaenum}[(1)]
\item A GC formula $\phi$ over variables $X$ and constants  $\Const$ is satisfiable iff $\cutoff[K]{\phi}$ is.
\item GC formulas $\phi$ over variables $U$ and $\Const$, and 
$\psi$ over $U \cup V$ and $\Const$ satisfy
$\smash{\cutoff[K]{\exists \vec U. \phi \wedge \psi} \equiv
\cutoff[K]{\exists \vec U. \cutoff[K]{\phi} \wedge \cutoff[K]{\psi}}}$.
\end{inparaenum}
\end{lemma}

\noindent
In the remainder of this section, let $\BB$ be a GC-DDSA and 
$\CC$ a constraint set, such that $C_{\alpha_0}$ and $\CC$ consist of GCs over variables $V$ and constants $\Const$,
and all guards of $\BB$ are GCs over $V^r\cup V^w$ and $\Const$,
with the bound $K$ as above.
Below we use the fact that GC formulas are closed under quantifier elimination:
if $\phi$ is a GC formula over variables $V \cup \{x\}$ and constants $\Const$, 
we can find a GC formula $\phi'$ that is equivalent to $\exists x.\, \phi$, quantifier-free,
and over the variables $V$ (though the constants in $\phi'$ need not be $\Const$). For details, see 
\cite{Revesz93}, \cite[Thm. 2]{BGI09}.

Let $\GC_\Const$ be the set of quantifier-free formulas whose atoms are GCs over $V$ and $\Const$ and $\Phi_\GC = B \times \GC_\Const$.
As $\GC_\Const$ may be infinite, we consider finite summary
w.r.t. the equivalence relation $\equivGC$ defined as 
$\phi \equivGC \psi$ iff $\cutoff[K]{\phi} \equiv \cutoff[K]{\psi}$.
\begin{example}
For $\BB_3$ from \figref{examples} we have 
$\Const = \{ 0,2,3\}$, so $K = 4$ (if $\CC = \emptyset$, otherwise constraints in $\CC$ need to be included).
The history constraints 
$\hist(\sigma_4) \equiv (x\,{\geq}\,5) \wedge (y\,{\geq}\,6)$ and
$\hist(\sigma_6) \equiv (x\,{\geq}\,8) \wedge (y\,{\geq}\,9)$ 
from \exaref{gap order} hence satisfy
$\hist(\sigma_4) \equivGC \hist(\sigma_6)$ because their cutoff is equal, namely
$\cutoff{\hist(\sigma_4)} = \cutoff{\hist(\sigma_6)} = (x\,{\geq}\,4) \wedge (y\,{\geq}\,4)$.
\end{example}

\begin{theorem}\label{thm:gap order}
$(\Phi_\GC,\equivGC)$ is a finite summary for $\BB$ and $\CC$.
\end{theorem}
\begin{proof}[Proof (sketch)]
We use \lemref{history set} to show that $\Phi_\GC$ is a history set:
first, for all $C\,{\subseteq}\,\CC$, 
$\bigwedge (C_{\alpha_0} \cup C)$ is in $\GC_\Const$.
Next, for $(b,\phi) \in \Phi_\GC$ and $b \goto{a} b'$, there is
some GC-formula $\chi$ over $U \cup V$ and $\Const$ such that
$\smash{\update(\phi,a) \wedge C = \exists \vec U. \phi(\vec U) \wedge \chi}$. 
From quantifier elimination we get a GC-formula $\phi'$ over $V$ and $\Const$
with $\phi' \equiv \exists \vec U. \phi(\vec U) \wedge \chi$, so $\phi' \in \GC_\Const$.
It remains to check \defref{finite summary}:
Suppose  $\phi \equivGC \psi$, so $\cutoff[K]{\phi} \equiv \cutoff[K]{\psi}$. 
Equisatisfiability of $\phi$ and $\psi$ follows from \lemref{cutoff} (1).
We can write
$\update(\phi,a) \wedge C = \exists \vec U. \phi(\vec U) \wedge \chi$ and
$\update(\psi,a) \wedge C = \exists \vec U. \psi(\vec U) \wedge \chi$
for some GC-formula $\chi$ over $U \cup V$ and $\Const$.
Then 
$\cutoff[K]{\exists \vec U. \phi(\vec U) \wedge \chi}
\equiv \cutoff[K]{\exists \vec U. \psi(\vec U) \wedge \chi}
$
follows using \lemref{cutoff} (2).
Finally, $(\Phi_\GC,\equivGC)$ has finitely many equivalence classes 
as the number of $K$-bounded GCs is finite.
\end{proof}

\noindent
With \thmref{model checking} it follows that model checking of a formula $\psi$ is decidable if $\neg \psi$ is expressible in $\LBC$.
However, the latter is not guaranteed for GC-DDSAs since GCs are not closed under negation.
For instance, $\Box(x\,{\geq}\,y)$ can be checked as its negation is expressible 
as $\Diamond (y\,{-}\,x\,{\geq}\,1)$; but $\Box(x\,{-}\,y\,{\geq}\,2)$ cannot as its negation is not expressible in $\LBC$ with GCs.

\smallskip
\noindent
\textbf{Feedback freedom}~\cite{DDV12} achieves decidability by forbidding variable updates that depend on an unbounded history: 
it requires that for every dependency between two instances $x_i$, $x_j$ of a variable $x$ in a run, another ``guard'' variable $y$, keeps its value for the time span $[i,j]$ of the dependency.
More precisely,
let $\sigma$ be a symbolic run of length $n$ whose $k$-th action is $a_{k}$, and $\CC$ a constraint set.
The \emph{computation graph} $G_{\sigma,\CC}$ is the undirected graph with nodes 
$\mc V = \{v_i \mid v\in V\text{ and }0\,{\leq}\,i\,{\leq}\,n\}$
and an edge from $x_i$ to $y_j$ iff $x_i$ and $y_j$ occur in a common literal of 
$\trans{a_{k}}(\vec V_{k-1}, \vec V_k) \wedge C(\vec V_k)$, 
for some $C \subseteq \CC$ and $i,j,k \leq n$.
The subgraph of $G_{\sigma,\CC}$ of all edges corresponding to equality literals $x_i = y_j$ for $x_i, y_j \in \mc V$ is denoted $E_{\sigma,\CC}$.

Let $\equiv_E$ be the smallest equivalence relation on $\mathcal V$ containing $E_{\sigma,\psi}$, so that the equivalence classes of $\equiv_E$ are the connected components of $E_{\sigma,\psi}$.
The equivalence class of $x_i\in \mc V$ is denoted $\eqc{x_i}$, and the \emph{span} of 
$\eqc{x_i}$ is the set
of affected instants, i.e., $span(\eqc{x_i}) = \{j \mid \exists v\inn V \text{ with }v_j\in \eqc{x_i}\}$.

\begin{definition}
For a \mydds $\BB$ and constraint set $\CC$, the pair
$(\BB,\CC)$ is \emph{feedback-free} if for every symbolic run 
$\sigma$, every path in $G_{\sigma,\CC}$ from $x_i$ to $x_j$ 
contains a node $y$ such
that $span(\eqc{x_i}) \cup span(\eqc{x_j}) \subseteq span(\eqc{y})$.
\end{definition}

\noindent
The next example illustrates this concept.

\begin{example}
\label{exa:computation graphs}
For runs $\sigma_2$ of $\BB_2$ and $\sigma_4$ of $\BB_4$ (cf. \figref{examples})
and $\CC = \{x\,{>}\,5, s\,{>}\,0\}$, we get the following graphs $G_{\sigma_i,\CC}$:
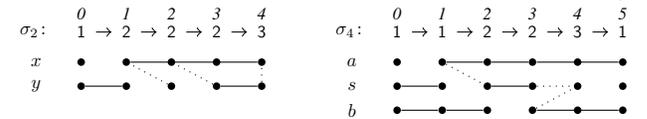
\begin{tikzpicture}[xscale=.6, yscale=.8]
\node[scale=.7] at (-1,.5) {$x$};
\node[scale=.7] at (-1,.1) {$y$};
\node[scale=.7] at (-1,1) {$\sigma_2\colon$};
\foreach \i/\l in {0/1,1/2,2/2,3/2,4/3} {
  \node[scale=.65] at (\i,1.3) {\i};
  \node[scale=.65] (state\i) at (\i,1) {$\m\l$};
  \node[fill, circle, inner sep=0pt, minimum width=1mm] (x\i) at (\i,.5) {};
  \node[fill, circle, inner sep=0pt, minimum width=1mm] (y\i) at (\i,.1) {};
  }
\foreach \s/\t in {0/1, 1/2, 2/3, 3/4} {
  \draw[->, shorten >=.6mm, shorten <=.6mm] (state\s) -- (state\t);
  }
\draw (y0) -- (y1);
\draw (x1) -- (x4);
\draw[dotted] (y2) -- (x1);
\draw[dotted] (y3) -- (x2);
\draw (y3) -- (y4);
\draw[dotted] (y4) -- (x4);
\begin{scope}[xshift=7cm]
\node[scale=.7] at (-1,.5) {$a$};
\node[scale=.7] at (-1,.1) {$s$};
\node[scale=.7] at (-1,-.3) {$b$};
\node[scale=.7] at (-1,1) {$\sigma_4\colon$};
\foreach \i/\l in {0/1,1/1,2/2,3/2,4/3,5/1} {
  \node[scale=.65] at (\i,1.3) {\i};
  \node[scale=.65] (state\i) at (\i,1) {$\m\l$};
  \node[fill, circle, inner sep=0pt, minimum width=1mm] (a\i) at (\i,.5) {};
  \node[fill, circle, inner sep=0pt, minimum width=1mm] (s\i) at (\i,.1) {};
  \node[fill, circle, inner sep=0pt, minimum width=1mm] (b\i) at (\i,-.3) {};
  }
\foreach \s/\t in {0/1, 1/2, 2/3, 3/4, 4/5} {
  \draw[->, shorten >=.6mm, shorten <=.6mm] (state\s) -- (state\t);
  }
\draw (b0) -- (b1);
\draw (s0) -- (s1);
\draw[dotted] (a1) -- (s2);
\draw (a1) -- (a2);
\draw (b1) -- (b2);
\draw (a2) -- (a3);
\draw (s2) -- (s3);
\draw[dotted] (s3) -- (s4);
\draw[dotted] (b3) -- (s4);
\draw (a3) -- (a5);
\draw (b3) -- (b5);
\end{scope}
\end{tikzpicture}
 where edges in $E_{\sigma_i,\CC}$ are drawn solid and others dotted. 
For $\sigma_2$, we have $span(\eqc{y_2})\allowbreak \cup span(\eqc{y_3}) \subseteq span(\eqc{x_1})$. The graph is similar for other runs, so that
$\BB_2$ is feedback free; but $\BB_4$ is not, as witnessed by the path from $s_3$ to $s_4$.
\end{example}
\noindent
We postpone the proof of the next theorem, to show 
below that feedback freedom is a special case of \emph{bounded lookback}.
\begin{theorem}\label{thm:feedback free}
Feedback freedom implies finite summary.
\end{theorem}


\smallskip
\noindent
\textbf{Bounded lookback.}
We next show that a DDSA $\BB$ has finite summary if, intuitively, 
at any point of a run of $\BB$ the values of $V$
depend on a bounded number of earlier steps.
Throughout this section, we consider a DDSA $\BB$ and constraint set $\CC$.
Moreover, we denote by $\eqc{G_{\sigma,\CC}}$ the graph obtained from $G_{\sigma,\CC}$ by collapsing
all edges in $E_{\sigma,\CC}$.

\begin{definition}
The pair $(\BB,\CC)$ has \emph{bounded lookback} if there is some $K$ such that for all
symbolic runs $\sigma$ of $\BB$,
all acyclic paths in $\eqc{G_{\sigma,\CC}}$ have length at most $K$.
\end{definition}

\noindent
For instance, after collapsing all solid (i.e., $E_{\sigma_4,\CC}$) edges of $G_{\sigma_4,\CC}$ in \exaref{computation graphs}, the longest path has length 3. In fact, one can show that $(\BB_4, \CC)$ has bounded lookback for $K=3$.

\begin{theorem}\label{thm:bounded lookback}
Bounded lookback implies \property. 
\end{theorem}
\begin{proof}[Proof (sketch)]
Let $\Psi$ be the set of formulas with free variables $V$, 
quantifier depth at most $K\cdot |V|$, and vocabulary $\CC$, $C_{\alpha_0}$, and
guards of $\BB$. As the quantifier depth is bounded and the set of atoms in the vocabulary is finite, $\Psi$ is finite up to equivalence.
Induction on $\sigma$ shows that $B \times \Psi$ is a history set%
:
If $\sigma$ is empty, $\hist(\sigma, \cseq C)$
is quantifier free and has all atoms in $C_{\alpha_0}$, hence it is in $\Psi$.
Otherwise, by induction hypothesis,
$\hist(\sigma|_n, \cseq C|_n)$ is equivalent to some $\phi \in \Psi$, so 
$\hist(\sigma, \cseq C) \equiv \exists \vec U. \phi(\vec U) \wedge \chi =:\phi'$ for some quantifier free $\chi$.
Let $\eqc{\phi'}$ be the formula that is obtained from $\phi'$ by eliminating all equality literals $x\,{=}\,y$, and substituting all variables in an equivalence class by a representative.
As $\eqc{\phi'}$ encodes $\eqc{G_{\sigma,\CC}}$ and
$(\BB,\CC)$ has $K$-bounded lookback, $\eqc{\phi'}$ is equivalent to a formula that has quantifier depth at most $K\cdot |V|$.
Hence, $\Psi$ must contain a formula equivalent to $\eqc{\phi'}$.
\end{proof}

\noindent
Note that all acyclic DDSAs have bounded lookback, for $K$ the number of states.
For feedback-free systems, \cite[Lem. 5.4]{DDV12} shows that 
$\eqc{G_{\sigma,\CC}}$ is a tree of depth at most $|V|$, so that \thmref{feedback free} follows from:

\begin{lemma}\label{lem:feedback free bounded lookback}
If $(\BB, \CC)$ is feedback-free then it has $2|V|$-bounded lookback.
\end{lemma}

\noindent
For a fixed $K$, bounded lookback is decidable in a similar way as
feedback freedom \cite[Sec. 4.4]{DDV12}, by
enumerating all possible variable dependencies in symbolic runs of $\BB$.
While \cite{DDV12} discovered that LTL model checking is decidable for feedback-free systems,
the respective result---implied by \thmsref{model checking}{bounded lookback}---for the
larger class of DDSAs with bounded lookback is new.

\section{Modularity}
\label{sec:modularity}

In this section we show that a DDSA admits a finite summary if it is suitably 
decomposable into smaller systems that enjoy this property. As finite summary of
the subsystems may be due to different criteria \criref{mcs}--\criref{bounded lookback},
modularity results substantially extend applicability of our approach.
%
As an arbitrary splitting of a DDSA $\BB$ into subsystems with finite summary does
not imply that $\BB$ inherits the property, 
we consider two specific ways of decomposition
for a \mydds $\BB = \langle B, b_0, \AA, T, F, V, \alpha_0, \guard\rangle$:

\begin{definition}
\label{def:seq decomposable}
Suppose $B = B_1 \cup B_2$, $B_1\cap B_2 = \{b\}$,
and $T$ contains neither edges from $B_2$ to  $B_1$, nor from $B_1 \setminus \{b\}$ to $B_2 \setminus \{b\}$; 
Let $T_1$ and $T_2$ be the projections of $T$ to $B_1\times \AA$ and $B_2 \times \AA$, respectively. 
Then $\BB$ is \emph{sequentially decomposable} into the \myddss
$\BB_1 = \langle B_1, b_0, \AA, T_1, \{b\}, V, \alpha_0, \guard\rangle$ and
$\BB_2 = \langle B_2, b, \AA, T_2, F, V \cup U, \alpha_U, \guard\rangle$, where $\alpha_U$
is the assignment such that $\alpha_U(\vec V) = \vec U$, for some set of variables $U$
such that $|U| = |V|$ and $U$ is disjoint from $V$.
\end{definition}

\begin{definition}
\label{def:var decomposable}
Let $V = V_1\,{\uplus}\,V_2$ such that
all constraints in $\{\guard(a)\mid a \inn \AA\}\cup \CC$
are over $V_1$ or $V_2$.\\
Then $(\BB, \CC)$ is \emph{variable-decomposable} into 
$(\BB_1, \CC|_{V_1})$ and $(\BB_2, \CC|_{V_2})$
where
$\BB_i = \langle B, b_0, \AA, T, F, V_i, \alpha_0|_{V_i}, guard_i\rangle$,
and $\guard_i(a)$ is $\guard(a)$ if it is over $V_i$, and
$\top$ otherwise.
\end{definition}

\noindent
Both ways of decomposition give rise to a modularity result:

\begin{theorem}
\label{thm:decompose}
Let $\BB$ be a \mydds admitting a decomposition into $\BB_1$ and $\BB_2$ 
that is either
$(a)$ sequential and so that $(\BB_i, \CC)$ has finite summary $(\Phi_i, \equiv)$, or 
$(b)$ variable 
and so that $(\BB_i, \CC|_{V_i})$ has finite summary $(\Phi_i, \sim_i)$,
for some $\CC$ and both $i\in \{1,2\}$. Then $(\BB,\CC)$ \hasproperty.
\end{theorem}

\begin{proof}[Proof (sketch)]
$(a)$ For
$\Phi = \Phi_1 \cup \{ \exists \vec U.\:\varphi_1(\vec U) \wedge \varphi_2 \mid 
\varphi_1 \in \Phi_1\text{ and }\varphi_2 \in \Phi_2 \}$,
the pair $(\Phi, \equiv)$ is a finite summary.
$(b)$ We show that $\Phi = \{ \varphi_1\wedge \varphi_2 \mid 
\varphi_1 \in \Phi_1\text{ and }\varphi_2 \in \Phi_2 \}$ with $\sim_1$ and $\sim_2$ combined is a finite summary.
\end{proof}

\noindent
We conclude this section by showing that \thmref{decompose} allows
us to handle our motivating example \exaref{auction}. Note that 
decidability does not follow by any of the criteria 
\criref{mcs}--\criref{bounded lookback} alone.

\begin{example}
\label{exa:auction2}
The system $\BB$ of \exaref{auction} is variable decomposable
into a red GC-\mydds $\BB_1$ over $\{b,d\}$,
and a blue/green system $\BB_2$ over $\{o,s,t\}$. 
$\BB_2$ can in turn be sequentially split into a 
blue MC-\mydds $\BB_{21}$, and the green single-step
system $\BB_{22}$ having 1-bounded lookback.
By \thmref{decompose}, $\BB$ has finite summary because so do $\BB_{1}$, $\BB_{21}$, and $\BB_{22}$.
Then \thmref{model checking} applies to check that there is no witness for $\Diamond(\mathsf{sold} \wedge d\,{>}\,0 \wedge o\,{\leq}\,t)$ 
(so property $\psi$ in \exaref{auction} holds).
On the other hand, we can obtain a witness for 
$\Diamond(b\,{=}\,1 \wedge o\,{>}\,t \wedge \Diamond(\mathsf{sold} \wedge b\,{\neq}\,1))$,
showing that a bid above the threshold $t$ need not win.
\end{example}

\section{Conclusion}
\label{sec:implementation}

\smallskip
\noindent
\textbf{Implementation.}
We implemented our approach in the prototype \tool (arithmetic DDS analyzer),
available via a web interface (\texttt{\url{https://ltl.adatool.dev}})
where also source code and examples can be found.
\tool takes a \mydds $\BB$ and an LTL$_f$ formula $\psi$ and checks whether $\BB$
and the constraints $\CC$ in $\psi$
admit a finite summary
according to \criref{mcs}--\criref{bounded lookback}, or if $\BB$, $\CC$ is suitably
decomposable (cf. \secref{modularity}).
If \property is detected, \tool
visualizes the constraint graph, the NFA $\NN_\psi$,
and $\smash{\NN_\BB^\psi}$, then 
extracts a witness for $\psi$ if it exists (cf. \thmref{model checking}).
In the extended version~\cite{adax} we show results 
for 
relevant examples, including \exaref{auction}
and processes converted from Petri nets with data~\cite{MannhardtLRA16}.
\tool is written in Python and uses the Z3 SMT solver~\cite{Z3}.

\smallskip
\noindent
\textbf{Future work.}
We see many possibilities for extensions: 
we expect finite summary to cover further known decidable cases,
e.g. \myddss with integer periodicity constraints \cite{Demri06}; and flat systems with Presburger-definable loop
effects~\cite{BDD13}.
For the criteria \criref{mcs}--\criref{bounded lookback}, it would
be interesting to investigate the complexity bounds implied by our method.
Further decomposition results would be useful, too, e.g. 
forms of parallel execution.
Next, we want to study whether our techniques apply to branching-time properties,
as well as transition systems with full-fledged 
relational databases in the vein of \cite{DHLV18,CGGM20}.

\clearpage
\bibliography{references}
\clearpage
\appendix

\section{Proofs}

\subsection{\mydds{s} with Finite Summary}

The next result relates history constraints and symbolic runs to actual runs.
For a run $\rho$ of length $n$ and $i < n$, let 
$\rho|_i$ denote the $i$-step prefix of $\rho$, for $i \leq n$.

\begin{numberedlemma}{\ref{lem:abstraction}}
For any symbolic run $\sigma\colon b_0 \goto{a_1} b_1 \goto{a_2} \dots \goto{a_n} b_n$ and $\cseq C= \langle C_0,\dots, C_n\rangle$,
\begin{enumerate}[{label=(\arabic*)}]
\item
If $\sigma$ abstracts a run 
$\rho\colon (b_0, \alpha_0) \goto{a_1} \dots \goto{a_n} (b_n, \alpha_n)$
and $\alpha_i \models C_i$ for all $i$, $0\,{\leq}\,i\,{\leq}\,n$,
then $\alpha_n$ satisfies $\smash[t]{\hist(\sigma, \cseq C)}$; 
\item
If $\smash[t]{\hist(\sigma,\cseq C)}$
is satisfied by assignment $\alpha$ then there is a run 
$\rho\colon (b_0, \alpha_0) \goto{a_1} \dots \goto{a_n} (b_n, \alpha_n)$
that is abstracted by $\sigma$ such that $\alpha=\alpha_n$ and
$\alpha_i \models C_i$ for all $i$, $0\,{\leq}\,i\,{\leq}\,n$.
\end{enumerate}
\end{numberedlemma}
\begin{proof}
\mbox{}
\begin{enumerate}[{label=(\arabic*)}]
\item 
By induction on $n$.
If $n\,{=}\,0$ then $\sigma$ and $\rho$ must be empty.
As $\alpha\,{=}\,\alpha_0$ satisfies $C_0$ by assumption, 
$\alpha$ also satisfies 
$\hist(\sigma, \langle C_0\rangle) = \bigwedge (C_{\alpha_0} \cup C_0)$.
Otherwise, suppose $\sigma$ is a symbolic run 
$\sigma \colon b_0 \goto{*} b_n \goto{a} b_{n+1}$
that abstracts
$\rho \colon
(b_0, \alpha_0) 
\goto{*} (b_n, \alpha_n)
\goto{a} (b_{n+1}, \alpha_{n+1})
$, and let
$\cseq C\,{=}\, \langle C_0,\dots, C_n, C_{n+1}\rangle$.
Then $\sigma|_n$ also abstracts $\rho|_n$, 
so by the induction hypothesis $\alpha_n$ satisfies $\hist(\sigma|_n, \cseq C|_n)$.
By definition of a step,
the guard assignment $\beta$ given by $\beta(\vec V^r) = \alpha_n(\vec V)$ and
$\beta(\vec V^w) = \alpha_{n+1}(\vec V)$
satisfies $\guard(a)$.
For $X = V \setminus \vwrite(a)$, $g = \guard(a)$, and the formula
$\varphi = \hist(\sigma|_n, \cseq C|_n)$, we thus have
\begin{align*}
\hist(\sigma, \cseq C) &= \update(\varphi, a) \wedge C_{n+1}\\
&= \exists \vec U.\:\varphi(\vec U) \wedge \trans{a}(\vec U, \vec V) \wedge C_{n+1}\\
&= \exists \vec U.\:\varphi(\vec U) \wedge
(g \wedge\!\bigwedge_{v\in X}\! v^{w}\,{=}\,v^{r})(\vec U, \vec V) \wedge C_{n+1}
\end{align*}
Since $\alpha_{n+1}\models C_{n+1}$ by assumption, it follows that
$\alpha_{n+1}$ satisfies $\hist(\sigma, \cseq C)$, using the values 
$\alpha_n(\vec V)$ as
witnesses for the existentially quantified variables $\vec U$.
\item
By induction on $n$.
For $n\,{=}\,0$ we have
$\hist(\sigma,\cseq C) = \bigwedge (C_{\alpha_0} \cup C_0)$, 
which is only satisfied by $\alpha_0$ 
since $\alpha_0$ fixes all variables in $V$.
Any empty run has variable assignment $\alpha_0$, and $\alpha_0$ satisfies $C_0$, 
so the claim holds.
Now consider a symbolic run 
$\sigma \colon b_0 \goto{*} b_n \goto{a} b_{n+1}$ 
and $\cseq C= \langle C_0,\dots, C_n, C_{n+1}\rangle$
such that
$\hist(\sigma, \cseq C)$ is satisfied by an assignment $\alpha$. 
Since 
\begin{align*}
\hist(\sigma, \cseq C)
&= \update(\hist(\sigma|_n, \cseq C|_n), a) \wedge C_{n+1} \\
&= \exists \vec U. \hist(\sigma|_n, \cseq C|_n)(\vec U) \wedge \trans{a}(\vec U, \vec V) \wedge C_{n+1} 
\end{align*}
we have $\alpha \models C_{n+1}$ and
there must be an assignment $\alpha'$ with domain $U$ such that $\alpha' \cup \alpha$ satisfies 
$\hist(\sigma|_n, \cseq C|_n)(\vec U)$ and $\trans{a}(\vec U, \vec V)$. 
Let $\alpha_n$ be the assignment with domain $\vec V$ such that $\alpha_n(\vec V) = \alpha'(\vec U)$,
so $\alpha_n$ satisfies $\hist(\sigma|_n, \cseq C|_n)$.
Therefore, by the induction hypothesis $\sigma|_n$ abstracts a run
$\rho \colon
(b_0, \alpha_0) 
\goto{a_1} (b_1, \alpha_1)
\goto{a_2} \dots
\goto{a_n} (b_n, \alpha_n)
$ with final assignment $\alpha_n$, and such that 
$\alpha_i \models C_i$ for all $i$, $0\,{\leq}\,i\,{\leq}\,n$.
Let $\beta$ be the guard assignment such that $\beta(\vec V^r) = \alpha_n(\vec V) = \alpha'(\vec U)$ and
$\beta(\vec V^w) = \alpha(\vec V)$.
Since $\alpha' \cup \alpha$ satisfies $\trans{a}(\vec U, \vec V)$, $\beta$ satisfies $\trans{a}(\vec V^r, \vec V^w)$ and hence $guard(a)$.
Thus $\rho$ can be extended with a step $(b_n, \alpha_n)\goto{a} (b_{n+1}, \alpha_{n+1})$.
By definition of a step the assignment $\alpha_{n+1}$ coincides with $\alpha$.
Moreover, as $\alpha$ satisfies $C_{n+1}$,
$\alpha_i$ satisfies $C_i$ for all $i$, $0\,{\leq}\,i\,{\leq}\,{n+1}$.
This proves the claim.
\qedhere
\end{enumerate}
\end{proof}

\begin{numberedlemma}{\ref{lem:history set}}
$\Phi$ is a history set iff
(1) for all $C \subseteq \CC$, there is some $(b_0,\phi_0) \in \Phi$ 
such that $\phi_0 \equiv \bigwedge (C_{\alpha_0} \cup C)$, and
(2) for all $(b,\phi)\in\Phi$, $b \goto{a} b'$, and $C \subseteq \CC$,
there is some $(b',\phi')\in \Phi$ such that $\phi' \equiv \update(\phi,a) \wedge C$.
\end{numberedlemma}
\begin{proof}
($\Longleftarrow$)
Suppose $\Phi$ satisfies (1) and (2).
We show that for every history constraint $\hist(\sigma,\cseq C)$ of $\BB$ and $\CC$
where $\sigma$ has final state $b$ there is some $(b, \phi) \in \Phi$ with 
$\hist(\sigma,\cseq C) \equiv \phi$, by induction on $\sigma$.
If $\sigma$ is empty and $\cseq C = \langle C_0\rangle$ then by Condition (1) there is some $(b_0,\phi_0) \in \Phi$ such that
$\phi_0 \equiv \bigwedge (C_{\alpha_0} \cup C_0)$.
Otherwise, $\sigma$ is of the form $b_0 \goto{*} b_n \goto{a} b_{n+1}$, and 
$\cseq C= \langle C_0,\dots, C_n, C_{n+1}\rangle$.
By the induction hypothesis there is some $(b_n,\phi) \in \Phi$ such that
$\hist(\sigma|_n,\cseq C|_n) \equiv \phi$.
By Condition (2), there is hence some $(b_{n+1},\phi')\in \Phi$ such that $\phi' \equiv \update(\phi,a) \wedge C_{n+1} \equiv h(\sigma, \cseq C)$.

\noindent
($\Longrightarrow$)
Suppose $\Phi$ is a history set. We verify that Conditions (1) and (2) hold.
For (1), let $\sigma$ be the empty run.
For all $C\subseteq \CC$, there must be some $(b_0,\phi_0) \in \Phi$
such that $\hist(\sigma,\langle C\rangle) = \bigwedge (C_{\alpha_0} \cup C) \equiv \phi_0$, so Condition (1) is satisfied.
Next, let $(b,\phi)\in\Phi$.
By the assumption that every pair in $\Phi$ corresponds to a history constraint  of $\BB$ and $\CC$, there must be some $\hist(\sigma,\cseq C)$ such that $\sigma$ ends in $b$ and
$\hist(\sigma,\cseq C)\equiv \phi$.
Let $\sigma'$ be $\sigma$ extended with $b \goto{a} b'$, and $\cseq C'$ be $\cseq C$ with $C$ appended.
As $\Phi$ is a history set, there must be some $(b',\phi') \in \Phi$ such that
$\hist(\sigma',\cseq C')\equiv \phi'$.
Since $\hist(\sigma',\cseq C') = \update(\hist(\sigma, \cseq C), a) \wedge C \equiv \update(\phi, a) \wedge C$, also Condition (2) holds.
\end{proof}

\subsection{Checking the Existence of Witnesses}

Before proving correctness of our approach in \thmref{model checking}, we 
establish relevant properties of our NFA construction (\defref{NFA}).
To that end, we define more precise consistency notions.
Let $\Sigma' =  2^{S \cup\{\last, \neg \last\}}$.
\begin{definition}
\begin{compactenum}
\item
$\varsigma \inn \Sigma$ is \emph{consistent with step $i$} of a symbolic run $\sigma \colon b_0 \goto{a_1} b_2 \goto{a_1} \dots \goto{a_n} b_n$ if
either $i=0$ and $\varsigma$ is disjoint from $B \setminus \{b_0\}$, or
$\varsigma$ is consistent with $b_{i-1} \goto{a_{i}} b_{i}$.
\item
$\varsigma \inn \Sigma$ is \emph{consistent with step $i$} of a run
\begin{equation}
\label{eq:therun}
\rho\colon
(b_0, \alpha_0) \goto{a_1}
(b_1, \alpha_1) \goto{a_2} \dots 
\goto{} (b_n, \alpha_n)
\end{equation}
if $\alpha_i$ satisfies $constr(\varsigma)$, and $\varsigma$ is consistent
with step $i$ of the abstraction $\sigma$ of $\rho$.
\item 
$\varsigma \inn \Sigma'$ is \emph{$\lambda$-consistent with step $i$} of a run $\rho$
if $\varsigma$ is consistent with step $i$ of $\rho$, 
if $i<n$ then $\last \not\in \varsigma$, and 
if $i=n$ then $\neg \last \not\in \varsigma$.
\end{compactenum}
\end{definition}

\noindent
By definition, a word $\varsigma_0 \varsigma_1 \cdots \varsigma_n\in \Sigma^*$ is consistent
with a symbolic run $\sigma$ (run $\rho$) if $\varsigma_i$ is consistent with step $i$ of $\sigma$ ($\rho$) for all $i$, $0\,{\leq}\,i\,{\leq}\,n$.
\smallskip

We first note that the function $\delta$ is total in the sense that its result set
contains an entry that is consistent with any pair of an assignment $\alpha$ and
a run $\sigma$:

\begin{lemma}
\label{lem:delta total}
For every run $\rho$ of the form \eqref{therun}, every $i$, $0\leq i \leq n$,
and $\varphi \in \LBC \cup \{\top,\bot\}$,
there is some $(\inquotes{\varphi'}, \varsigma) \in \delta(\inquotes{\varphi})$ 
such that 
$\varsigma$ is $\lambda$-consistent with step $i$ of $\rho$.
\end{lemma}
\begin{proof}
By induction on the structure of $\varphi$. The claim is easy to check for every base
case of the definition of $\delta$, and in all other cases it follows from the induction hypothesis.
\end{proof}

\smallskip
We next show a crucial feature of the $\delta$ function,
namely that it preserves and reflects the property of a run satisfying a formula.
Both directions are proven by laborious but straightforward induction proofs on the formula
structure.

\begin{lemma}
\label{lem:delta}
Let $\varphi \in \LBC \cup \{\top,\bot\}$,
$\rho$ a run of the form \eqref{therun}, and $0\,{\leq}\,i\,{\leq}\,n$.
Then
$\rho,i \models \varphi$ holds if and only if
there is some $(\inquotes{\varphi'}, \varsigma)\in \delta(\inquotes{\varphi})$ such that\\
\noindent
\begin{tabular}{@{\ }r@{\ }p{8cm}}
$(a)$ & $\varsigma$ is $\lambda$-consistent with step $i$ of $\rho$, \\
$(b)$ &either $i<n$ and $\rho,i{+}1 \models \varphi'$, or $i\,{=}\,n$ and $\varphi' = \top$.
\end{tabular}
\end{lemma}
\begin{proof}
($\Longrightarrow$)
We first note that if $\varphi' = \top$ then $(b)$ holds for both $i<n$ and $i=n$ ($\star$).
The proof is by induction on $\varphi$.
\begin{compactitem}
\item
If $\varphi = \top$, we can choose $(\inquotes{\varphi'},\varsigma) = (\top, \emptyset)$.
\item
If $\rho,i \models p$ for some $p \inn \AA \cup B$, 
we must have either $p = b_i$, or 
$p = a_i$ and $i>0$. By choosing
$(\inquotes{\top},\{p\})\in\delta(\inquotes{p})$, consistency holds and we use ($\star$)
for $(b)$.
\item
If $\rho,i \models c$ for some $c \in \CC$, we may take
$(\inquotes{\top}, \{c\})\in \delta(\inquotes{c})$.
As $\rho,i \models c$, $\alpha_i$ satisfies 
$c$, so consistency holds and we use ($\star$)
for $(b)$.
\item 
Suppose $\rho,i \models \langle \cdot\rangle \varphi$, so that
$i\,{<}\,n$ and $\rho,i{+}1 \models \varphi$.
For
$(\inquotes{\varphi}, \{\neg \last\}) \in \delta(\inquotes{\langle \cdot\rangle \varphi})$,
part $(a)$ holds since 
$\neg \last \in \varsigma$ and $i\,{<}\,n$, and
$(b)$ because of $\rho,i{+}1 \models \varphi$.
\item 
Suppose $\varphi = \varphi_1\wedge \varphi_2$.
By assumption $\rho,i \models \varphi_1\wedge \varphi_2$, and hence 
$\rho,i \models \varphi_1$ and $\rho,i \models \varphi_2$.
By the induction hypothesis, there are
$(\inquotes{\varphi_1'}, \varsigma_1) \in \delta (\inquotes{\varphi_1})$ and
$(\inquotes{\varphi_2'}, \varsigma_2) \in \delta (\inquotes{\varphi_2})$ such that 
for both $k\in \{1,2\}$,
$(a')$ $\varsigma_k$ is consistent with step $i$ of $\rho$ and
$(b')$ either $\rho,i{+}1 \models \varphi_k'$, or $i\,{=}\,n$ and $\varphi_k' = \top$.
By definition of $\delta$,  we can choose 
$(\inquotes{\varphi_1' \wedge \varphi_2'}, \varsigma_1 \cup \varsigma_2) \in 
\delta(\inquotes{\varphi_1 \wedge \varphi_2})$.
Then
$(a)$ follows from $(a')$ and $\varsigma = \varsigma_1 \cup \varsigma_2$, and
$(b)$ if $i=n$ then $(b')$ implies $\varphi' = \top$, and otherwise
$\rho,i{+}1 \models \varphi_1' \wedge \varphi_2'$.
\item 
Suppose $\varphi = \varphi_1\vee \varphi_2$.
By assumption $\rho,i \models \varphi_1\vee \varphi_2$, and hence 
$\rho,i \models \varphi_1$ or $\rho,i \models \varphi_2$.
We assume the former.
By the induction hypothesis, there is some
$(\inquotes{\varphi_1'}, \varsigma_1) \in \delta (\inquotes{\varphi_1})$ such that 
$(a')$ $\varsigma_1$ is consistent with step $i$ of $\rho$, and
$(b')$ $\rho,i{+}1 \models \varphi_1'$, or $i\,{=}\,n$ and $\varphi_1' = \top$.
As $\delta$ is total (\lemref{delta total}), there must be some 
$(\inquotes{\varphi_2'}, \varsigma_2) \in \delta (\inquotes{\varphi_2})$ such that
$\varsigma_2$ is $\lambda$-consistent with step $i$ of $\rho$.
By definition of $\delta$,  we can choose 
$(\inquotes{\varphi'},\varsigma)$ as 
$(\inquotes{\varphi_1' \vee \varphi_2'}, \varsigma_1 \cup \varsigma_2) \in 
\delta(\inquotes{\varphi_1 \vee \varphi_2})$.
Then
$(a)$ follows from $(a')$ and $\varsigma_2$ being $\lambda$-consistent with step $i$ of $\rho$, and
$(b)$ if $i\eqn n$ then $(b')$ implies $\varphi_1'\eqn \top$, hence $\varphi' \eqn \top$; otherwise $\rho,i{+}1 \models \varphi_1' \vee \varphi_2'$.
\item 
Suppose $\rho,i \models \Diamond \varphi$, so either 
(1) $\rho,i \models \varphi$, or 
(2) $i\,{<}\,n$ and $\rho,i{+}1 \models \Diamond \varphi$.
We have $\delta(\inquotes{\Diamond \varphi}) = \delta(\inquotes{\varphi}) \vee \delta(\inquotes{\langle\cdot\rangle \Diamond \varphi})$.
In case (1), by the induction hypothesis there is some
$(\inquotes{\varphi'}, \varsigma') \in \delta (\inquotes{\varphi})$ such that
$(a_1)$ $\varsigma'$ is $\last$-consistent with step $i$ of $\rho$, and
$(b_1)$ $\rho,i{+}1 \models \varphi'$, or $i\,{=}\,n$ and $\varphi' = \top$.
In case (2), let $(\inquotes{\varphi'}, \varsigma')$ be 
$(\inquotes{\Diamond \varphi}, \{\neg \last\}) \in 
\delta(\inquotes{\langle\cdot\rangle \Diamond \varphi})$, so that
$(a_2)$ $\varsigma'$ is $\last$-consistent with step $i$ of $\rho$ because $i<n$, and
$(b_2)$ $\rho,i{+}1 \models \varphi$.
In either case, the claim follows from $(a_k)$ and $(b_k)$ by a similar
reasoning as in the case for disjunction.
\item
Suppose $\rho,i \models \Box \varphi$, so $\rho,i \models \varphi$ and either 
(1) $i=n$, or 
(2) $\rho,i{+}1 \models \Box \varphi$.
We have 
$\delta(\inquotes{\Box \varphi}) = \delta(\inquotes{\varphi}) \wedge (\delta(\inquotes{\langle\cdot\rangle \Box \varphi}) \vee \delta_\last) = 
(\delta(\inquotes{\varphi}) \wedge \delta(\inquotes{\langle\cdot\rangle \Box \varphi}))
\vee (\delta(\inquotes{\varphi}) \wedge \delta_\last)$.
In either case, by the induction hypothesis there is some
$(\inquotes{\varphi'}, \varsigma') \in \delta (\inquotes{\varphi})$ such that
$(a')$ $\varsigma'$ is $\last$-consistent with step $i$ of $\rho$, and
$(b')$ $\rho,i{+}1 \models \varphi'$, or $i\,{=}\,n$ and $\varphi' = \top$.

(1) Let
$(\inquotes{\varphi_1},\varsigma_1)$ be
$(\inquotes{\varphi' \wedge \top}, \varsigma' \cup \{ \last \}) \in 
\delta(\inquotes{\varphi}) \wedge \delta_\last$. We have
$(a_1)$ $\varsigma_1$ is $\last$-consistent with step $i$ of $\rho$ because of $(a')$ and $i\eqn n$, and
$(b_1)$ $\rho,i{+}1 \models \varphi_1 = \top$ by $(b')$.

(2) Let 
$(\inquotes{\varphi_2},\varsigma_2)$ be
$(\inquotes{\varphi' \wedge \Box \varphi}, \varsigma' \cup \{ \neg \last \}) \in 
\delta(\inquotes{\varphi}) \wedge \delta(\inquotes{\langle\cdot\rangle \Box \varphi})$.
Then
$(a_2)$ $\varsigma_2$ is $\last$-consistent with step $i$ of $\rho$ by $(a')$ and $i<n$, and
$(b_2)$ $\rho,i{+}1 \models \varphi_2 = \varphi' \wedge \Box \varphi$, using $(b')$ and $\rho,i{+}1 \models \Box \varphi$.
Thus the two cases can be combined as in the case for disjunction,
using $(a_1)$, $(b_1)$ and $(a_2)$, $(b_2)$.
\item The case for the $\until$ operator is similar.
\end{compactitem}
\noindent
($\Longleftarrow$)
Note that the assumptions exclude $\phi' = \bot$.
We apply induction on $\varphi$, and use the definition of $\delta$ for each case.
\begin{compactitem}
\item 
If $\varphi = \top$ then $(\inquotes{\varphi'}, \varsigma)\in \delta(\inquotes{\varphi})$ implies $\phi' = \top$, and 
$\rho, i \models \top$ holds.
\item
If $\varphi\,{=}\,c\in \CC$, we must have
$\varphi' = \top$ and $\varsigma =\{c\}$.
As $\alpha_{i}$ satisfies $constr(\varsigma)$ by $\last$-consistency, $\rho,i \models c$.
\item
If $\varphi = p$ for $p \in B \cup \AA$, we must have
$\varphi' = \top$ and $\varsigma = \{p\}$.
By consistency, either 
$i=0$ and $\varsigma$ is disjoint from $B \setminus \{b_0\}$, or
$i>0$ and $\varsigma$ is consistent with $b_{i-1} \goto{a_{i}} b_{i}$.
In the former case, we must have $p=b_0$, 
so $\rho,0 \models p$.
Otherwise, by consistency either $p = b_{i}$ or $p = a_i$.
In the former case $\rho,i \models b_i$ clearly holds.
If $p = a_i$, this occurrence of $a_i$ must be the result of replacing a subformula
$\langle a_i\rangle \psi'$ in $\psi$ by
$\langle \cdot\rangle (a_i \wedge \psi')$, for some $\psi'$.
We have $\rho,i \models a_i$, equivalent to $\rho,i-1 \models \langle a_i\rangle \psi'$.
\item 
Let $\varphi = \langle \cdot \rangle \chi$.
As $\varphi'=\top$ or $\rho, i{+}1 \models \varphi'$, 
by definition of $\delta$ the only possibility is
$\varphi' = \inquotes{\chi}$ and $\varsigma = \{\neg \last\}$.
As $\varsigma$ is consistent with step $i$ of $\rho$ and $\neg \last \in \varsigma$,
we must have $i\,{<}\,n$, so $\rho, i{+}1 \models \varphi'$
and hence $\rho, i \models \varphi$ by \defref{witness}.
\item
If $\varphi = \varphi_1 \wedge \varphi_2$ then 
by $(\inquotes{\varphi'}, \varsigma) \in \delta(\inquotes{\varphi})$
and the definition of $\delta$ there are $\varphi_1'$ and $\varphi_2'$ such that
$(\inquotes{\varphi_1'}, \varsigma_1')\in \delta(\inquotes{\varphi_1})$ and $(\inquotes{\varphi_2'}, \varsigma_2')\in \delta(\inquotes{\varphi_2})$, and
$\varphi' = \varphi_1' \wedge \varphi_2'$ and $\varsigma = \varsigma_1' \cup \varsigma_2'$.
Therefore, either $i=n$ and $\varphi' = \varphi_1' = \varphi_2' = \top$, or
$i<n$ and $\rho, i{+}1 \models \varphi'$, which implies
$\rho, i{+}1 \models \varphi_1'$ and $\rho, i{+}1 \models \varphi_2'$.
In either case, 
$\rho, i \models \varphi_1$ and $\rho, i \models \varphi_2$ 
hold by the induction hypothesis,
so $\rho, i \models \varphi_1 \wedge \varphi_2$.
\item
Similarly,
if $\varphi = \varphi_1 \vee \varphi_2$ then there are $\varphi_1'$ and $\varphi_2'$ such that
$(\inquotes{\varphi_1'}, \varsigma_1')\in \delta(\inquotes{\varphi_1})$ and $(\inquotes{\varphi_2'}, \varsigma_2')\in \delta(\inquotes{\varphi_2})$,
$\varphi' = \varphi_1' \vee \varphi_2'$ and $\varsigma = \varsigma_1' \cup \varsigma_2'$.
If $i=n$ and $\varphi' = \top$, then $\varphi_1' = \top $ or $\varphi_2' = \top$.
If otherwise $i<n$ then $\rho, i{+}1 \models \varphi'$ implies
$\rho, i{+}1 \models \varphi_1'$ or $\rho, i{+}1 \models \varphi_2'$.
From the induction hypothesis we obtain in either case
$\rho, i \models \varphi_1$ or $\rho, i \models \varphi_2$,
so $\rho, i \models \varphi_1 \vee \varphi_2$.
\item
If $\varphi = \Diamond \chi$ then by the definition of $\delta$, 
there are $\varphi_1'$ and $\varphi_2'$ such that
$(\inquotes{\varphi_1'}, \varsigma_1')\in \delta(\inquotes{\chi})$, 
$(\inquotes{\varphi_2'}, \varsigma_2')\in \delta(\inquotes{\langle\cdot\rangle\Diamond \chi})$,
$\varphi' = \varphi_1' \vee \varphi_2'$, and $\varsigma = \varsigma_1' \cup \varsigma_2'$.
If $i=n$ and $\varphi' = \top$, we must have $\varphi_1' = \top$, and
by the induction hypothesis, $\rho, i \models \chi$.
If $i<n$, since $\rho, i{+}1 \models \varphi' = \varphi_1' \vee \varphi_2'$, by \defref{witness} either
(1) $\rho, i{+}1 \models \varphi_1'$ or (2) $\rho, i{+}1 \models \varphi_2'$.
In case (1), by the induction hypothesis, $\rho, i \models \chi$.
In case (2),
$(\inquotes{\varphi_2'}, \varsigma_2')\in \delta(\inquotes{\langle\cdot\rangle\Diamond \chi})$
and $\rho, i{+}1 \models \varphi_2'$
implies $\varphi_2' = \Diamond \chi$ and we have $i < n$, so
$\rho, i \models \langle\cdot\rangle\Diamond \chi$ by \defref{witness}.
Either way, $\rho, i \models \Diamond \chi$ holds.
\item
If $\varphi = \Box \chi$ then we can distinguish two cases:

(1) There are $\varphi_1$ and $\varphi_2$ such that
$(\inquotes{\varphi_1'}, \varsigma_1)\in \delta(\inquotes{\chi})$, $(\inquotes{\varphi_2'}, \varsigma_2)\in \delta_{\last}$, 
$\varphi' = \varphi_1' \wedge \varphi_2'$ and $\varsigma = \varsigma_1 \cup \varsigma_2$.
As $(\inquotes{\varphi_2'}, \varsigma_2)\in \delta_{\last}$, we must have $\varphi_2' = \top$ and
$\varsigma_2 = \{\last\}$ (otherwise, we would have $\varphi'=\bot$).
By consistency, $\last \in \varsigma$ implies $i = n$, so $\varphi' = \top$
by assumption and therefore we must have $\varphi_1' = \top$.
From the induction hypothesis and
$(\inquotes{\varphi_1'}, \varsigma_1)\in \delta(\inquotes{\chi})$ 
we conclude $\rho, i \models \chi$, so by \defref{witness}
$\rho, i \models \Box \chi$.

(2)
There are $\varphi_1'$ and $\varphi_2'$ such that
$(\inquotes{\varphi_1'}, \varsigma_1)\in \delta(\inquotes{\chi})$, 
$(\inquotes{\varphi_2'}, \varsigma_2)\in \delta(\inquotes{\langle\cdot\rangle\Box \chi})$,
$\varphi' = \varphi_1' \wedge \varphi_2'$ and $\varsigma = \varsigma_1 \cup \varsigma_2$.
We have $\neg\last \in \varsigma_2$, so by consistency $i<n$.
As $\rho, i{+}1 \models \varphi' = \varphi_1' \wedge \varphi_2'$, by \defref{witness}  
$\rho, i{+}1 \models \varphi_1'$ and $\rho, i{+}1 \models \varphi_2'$.
By the induction hypothesis, 
$(\inquotes{\varphi_1'}, \varsigma_1)\in \delta(\inquotes{\chi})$ and
$\rho, i{+}1 \models \varphi_1'$
imply $\rho, i \models \chi$.
Moreover,
$(\inquotes{\varphi_2'}, \varsigma_2)\in \delta(\inquotes{\langle\cdot\rangle\Box \chi})$
and 
$\rho, i{+}1 \models \varphi_2'$ imply 
$\varphi_2' = \Box \chi$ by \defref{witness}, so we have
$\rho, i \models \langle\cdot\rangle\Box \chi$.
Thus $\rho, i \models \Box \chi$.
\item
The case for $\until$ is similar. \qedhere
\end{compactitem}
\end{proof}

Let a word $\varsigma_0 \varsigma_1 \cdots \varsigma_{n}\in \Sigma'^*$ be \emph{well-formed}
if $\last \not\in \varsigma_i$ for all $0\leq i < n$, and $\neg \last \not\in \varsigma_n$.

\begin{lemma}
\label{lem:deltastar}
A well-formed word $w\in \Sigma'^*$ that is consistent with a run $\rho$
satisfies $\inquotes{\top} \in \delta^*(\inquotes{\psi},w)$
iff
$\rho \models \psi$.
\end{lemma}
\begin{proof}
($\Longrightarrow$)
Let $w = \varsigma_0 \varsigma_1 \cdots \varsigma_{n}$ and
$\chi_0,\chi_1,\dots, \chi_{n+1}$
be the sequence of formulas witnessing
$\inquotes{\top} \in \delta^*(\inquotes{\psi},w)$,
so that $\chi_0 = \psi$, $\chi_{n+1} = \top$,
and $(\inquotes{\chi_{i{+}1}}, \varsigma_{i}) \in \delta(\inquotes{\chi_i})$
for all $i$, $0\leq i \leq n$.
As $w$ is well-formed and consistent with $\rho$, 
by definition $\varsigma_i$ is $\lambda$-consistent with 
$\rho$ at $i$ for all $i$, $0\leq i \leq n$.
In order to show that $\rho$ is a witness for $\psi$, we verify that
$\rho, i \models \chi_i$
for all $i$, $0\,{\leq}\,i\,{\leq}\,n$, by induction on $n-i$.
In the base case $i\,{=}\,n$.
We have $\chi_{n+1} = \top$ and 
$(\inquotes{\chi_{n{+}1}}, \varsigma_{n}) \in \delta(\inquotes{\chi_n})$,
and from \lemref{delta}\:($\Longrightarrow$) it follows that $\rho, n \models \chi_n$.
If $i < n$, we assume by the induction hypothesis that $\rho, i{+}1 \models \chi_{i{+}1}$.
We have
$(\inquotes{\chi_{i{+}1}}, \varsigma_{i}) \in \delta(\inquotes{\chi_i})$, so
$\rho, i \models \chi_i$ follows again from \lemref{delta}\:($\Longrightarrow$), which concludes the induction step.
Finally, the claim follows for the case $i\,{=}\,0$ because 
$\chi_0 = \psi$.

($\Longleftarrow$)
Let $\rho$ be of the form \eqref{therun}.
We show that for all $i$, $0 \leq i \leq n$,
and every formula $\chi$,
if $\rho, i \models \chi$ then 
there is a word $w_i = \varsigma_i \varsigma_{i+1} \cdots \varsigma_{n}$
of length $n-i+1$
such that
$\inquotes{\top} \in \delta^*(\inquotes{\chi},w_i)$,
and 
$\varsigma_j$ is $\lambda$-consistent with step $j$ of $\rho$ for all $j$, $i\leq j \leq n$.
The proof of is by induction on $n-i$.

In the base case where $i=n$, we assume that $\rho,n \models \chi$. 
By \lemref{delta}\:($\Longleftarrow$) there is some $\varsigma_n$ such that
$(\inquotes{\top}, \varsigma_n)\in \delta(\inquotes{\chi})$,
and $\varsigma_n$ is $\lambda$-consistent with step $n$ of $\rho$.
For the induction step, assume $i\,{<}\,n$ and $\rho, i \models \chi$.
By \lemref{delta}\:($\Longleftarrow$) there is some
$(\inquotes{\chi'}, \varsigma_i) \in \delta(\inquotes{\chi})$ such that
$\rho, i{+}1 \models \chi'$, and moreover
$\varsigma_i$ is $\lambda$-consistent with $\rho$ at step $i$.
By the induction hypothesis, 
there is a word $w_{i+1} = \varsigma_{i+1} \cdots \varsigma_{n}$
such that
$\inquotes{\top} \in \delta^*(\inquotes{\chi'},w_{i+1})$,
and 
$\varsigma_j$ is $\lambda$-consistent with $\rho$ at instant $j$ for all $j$, $i\,{<}\,j\,{\leq}\, n$.
Thus, we can define $w_{i} = \varsigma_i\varsigma_{i+1} \cdots \varsigma_{n}$,
which satisfies
$\inquotes{\top} \in \delta^*(\inquotes{\chi},w_{i})$ and 
$\varsigma_j$ is $\lambda$-consistent with $\rho$ at instant $j$ for all $j$, $i\,{\leq}\,j\,{\leq}\, n$.
This concludes the induction step.

By assumption, $\rho,0 \models \psi$  holds.
From the case $i = 0$ of the above statement, we obtain a word $w=w_0$
such that
$\inquotes{\top} \in \delta^*(\inquotes{\psi},w)$ and $w$ is
$\lambda$-consistent with all steps of $\rho$, i.e., $w$ is well-formed and consistent with $\rho$.
\end{proof}

\begin{lemma}
\label{lem:delta last}
Let $\psi \in \LBC$ and
$(\chi,\varsigma) \in \delta(\inquotes{\psi})$.
\begin{compactenum}
\item[(1)] If $\last \in \varsigma$ and $\neg \last \not\in\varsigma$
then $\chi = \top$ or $\chi = \bot$.
\item[(2)] Suppose $\neg \last \in \varsigma$, $\last \not\in\varsigma$,
and $\chi = \top$, and $\varsigma$ is consistent with some step $i$ of some run $\rho$.
Then there is some $(\top,\varsigma') \in \delta(\inquotes{\psi})$
such that $\neg \last\not\in \varsigma'$ and $\varsigma'$ is consistent with step $i$ of $\rho$ as well.
\item[(3)] If $\chi$ is not $\top$ or $\bot$ then
$\varsigma$ contains $\last$ or $\neg \last$.
\end{compactenum}
\end{lemma}
\begin{proof}
All three statements are shown simultaneously by induction on $\psi$.
\begin{compactitem}
\item
If $\psi$ is $\top$, $\bot$, or an atom $p$ then $\chi$ is $\top$ or $\bot$, so (1) and (3) hold, and $\neg \last \not \in \varsigma$, so also (2) is satisfied.
\item
If $\psi = \langle \cdot\rangle \psi'$ then 
$\delta(\inquotes{\psi}) = \{(\inquotes{\psi'},\{\neg \last\}), (\inquotes{\bot}, \{\last\})\}$.
(1) is satisfied by $(\inquotes{\bot}, \{\last\})$,
(2) holds because $\psi'$ cannot be $\top$ since $\top$ does not occur in $\LBC$, and
(3) is satisfied anyway.
\item
If $\psi = \psi_1 \vee \psi_2$ then we must have $\chi = \chi_1\vee \chi_2$
such that $(\chi_i,\varsigma_i) \in \delta(\inquotes{\psi_i})$ for both 
$i\in \{1,2\}$, and $\varsigma = \varsigma_1 \cup \varsigma_2$.

(1)
Suppose $\last \inn \varsigma$ and $\neg \last \not\in \varsigma$.
First, assume $\last \in \varsigma_1$, 
$\neg \last \not\in \varsigma_1$,
and $\neg \last \not\in \varsigma_2$.
By the induction hypothesis (1), $\chi_1$ is either $\top$ or $\bot$.
In the former case, $\chi = \top$, so the claim holds.
Otherwise, $\chi = \chi_2$.
Then, if $\last \in \varsigma_2$ we can again use the induction hypothesis
to conclude that $\chi = \chi_2$ is $\top$ or $\bot$.
Otherwise, we have $\last \not\in \varsigma_2$ and
$\neg \last \not\in \varsigma_2$, so
$\chi_2$ must be $\top$ or $\bot$
by the induction hypothesis (3).

(2)
Suppose $\neg \last \inn \varsigma$, $\last \not\in \varsigma$, and $\chi=\top$,
and $\varsigma$ is consistent with step $i$ of $\rho$.
W.l.o.g., we can assume $\neg \last \in \varsigma_1$, 
$\last \not\in \varsigma_1$, $\last \not\in \varsigma_2$,
and $\chi_1 = \top$.
By the induction hypothesis (2) applied to $\chi_1$, there is some
$(\top,\varsigma_1') \in \delta(\inquotes{\psi_1})$
such that $\neg \last \not\in\varsigma_1'$ and $\varsigma_1'$ is consistent with step $i$ of $\rho$.
By \lemref{delta total}, there is some 
$(\chi_2',\varsigma_2') \in \delta(\inquotes{\psi_2})$ such that
$\varsigma_2'$ is consistent with step $i$ of $\rho$, and such that $\neg \last \not\in \varsigma_2'$.
Thus $(\top,\varsigma') \in \delta(\inquotes{\psi})$
with $\varsigma' = \varsigma_1'\cup \varsigma_2'$ satisfies the claim.

(3)
If $\chi$ is not $\top$ or $\bot$ then at least one of $\chi_1$ or
$\chi_2$ is not $\top$ or $\bot$, so by the induction hypothesis (3),
$\varsigma_1$ or $\varsigma_2$ contains $\last$ or $\neg \last$, hence so
does $\varsigma$.
\item
If $\psi = \psi_1 \wedge \psi_2$ then we must have $\chi = \chi_1\wedge \chi_2$
such that $(\chi_i,\varsigma_i) \in \delta(\inquotes{\psi_i})$ for both 
$i\in \{1,2\}$.

(1)
Suppose $\last \in \varsigma$ and $\neg \last \not\in \varsigma$.
W.l.o.g., we can assume $\last \in \varsigma_1$, 
$\neg \last \not\in \varsigma_1$,
and $\neg \last \not\in \varsigma_2$.
By the induction hypothesis (1), $\chi_1$ is either $\top$ or $\bot$.
In the latter case, $\chi = \bot$, so the claim holds.
Otherwise, $\chi = \chi_2$, and due to the assumption that
$\neg \last \not\in \varsigma_2$, by the induction hypothesis (3), 
$\chi_2$ must be $\top$ or $\bot$.

(2)
Suppose $\neg \last \in \varsigma$, $\last \not\in \varsigma$, and $\chi=\top$,
and $\varsigma$ is consistent with step $i$ of $\rho$.
We can assume  $\last \not\in \varsigma_1$, $\last \not\in \varsigma_2$ and
$\chi_1 = \chi_2=\top$.
We must have $\last \in \varsigma_1$, $\last \in \varsigma_2$, or both.
However, for each $i\in \{1,2\}$ such that $\last \in \varsigma_i$, by
the induction hypothesis (2) there is some
$(\top,\varsigma_i') \in \delta(\inquotes{\psi_i})$
such that $\last \not\in\varsigma_i'$ and $\varsigma_i'$ is consistent with step $i$ of $\rho$.
If $\last \not\in \varsigma_i$, set $\varsigma_i' = \varsigma_i$.
Hence $(\top,\varsigma_1'\cup \varsigma_2') \in \delta(\inquotes{\psi})$ such that
$\varsigma_1'\cup \varsigma_2'$ is consistent with step $i$ of $\rho$ and 
$\neg \last\not\in \varsigma_1'\cup \varsigma_2'$.

(3)
If $\chi$ is not $\top$ or $\bot$ then at least one of $\chi_1$ or
$\chi_2$ is not $\top$ or $\bot$, so by the induction hypothesis (3),
$\varsigma_1$ or $\varsigma_2$ contains $\last$ or $\neg \last$, hence so
does $\varsigma$.
\item For $\psi = \Box \psi'$, note that
$\delta_\lambda = \{(\inquotes{\top},\{ \lambda\}),(\inquotes{\bot},\{\neg\lambda\})\}$ satisfies the properties.
The result then follows from the cases for $\vee$ and $\wedge$.
\item All other cases follow from the cases for $\vee$ and $\wedge$.
\qedhere
\end{compactitem}
\end{proof}

\begin{numberedlemma}{\ref{lem:NFA acceptance}}
$\NFApsi$ accepts a word $w$ that is consistent with a run $\rho$ iff 
$\rho \models \psi$.
\end{numberedlemma}
\begin{proof}
($\Longrightarrow$)
Let $w = \varsigma_0 \varsigma_1 \cdots \varsigma_{n}$, and
$q_0 \goto{\varsigma_0} q_1 \goto{\varsigma_1} \dots \goto{\varsigma_{n}} q_{n+1}$  ($\star$) 
be the respective accepting run of $\NFApsi$.
By \defref{NFA}, there are $\varsigma_i'$, 
such that $\varsigma_i = \varsigma_i' \setminus \{\last, \neg \last\}$
and $\{\last, \neg \last\} \not\subseteq \varsigma_i'$
for all $i$, $0\leq i \leq n$.
Let $w'$ be the word $w' = \varsigma_0' \varsigma_1' \cdots \varsigma_{n}'$.
Then $w'$ is consistent with $\rho$ because so is $w$.
Moreover,
by \lemref{delta last} (2) we can choose $\varsigma_{n}'$ such that
$\neg \last \not\in \varsigma_{n}'$, and $\varsigma_{n}'$ is consistent with $\rho$ at $n$.
Then $w'$ is well-formed: indeed,
since edges to $\inquotes{\top}$ labeled $\last$ are redirected to $q_e$
and $\inquotes{\bot}$ cannot occur in ($\star$), by \lemref{delta last} (1)
we have $\last\not\in \varsigma_i'$ for $i<n$.
Thus by \defref{NFA} we have
$\inquotes{\top} \in \delta^*(\inquotes{\psi},w')$.
According to \lemref{deltastar}, $\rho \models \psi$.

($\Longleftarrow$)
If $\rho \models \psi$ then by \lemref{deltastar} there is a 
well-formed word 
$w = \varsigma_0 \varsigma_1 \cdots \varsigma_{n}$ 
that is consistent with $\rho$ such that $\inquotes{\top} \in \delta^*(\inquotes{\psi},w)$.
As $w$ is well-formed, no $\varsigma_i$ contains both $\last$
and $\neg \last$. Hence all $\delta$-steps are reflected by transitions in 
$\NFApsi$. Since moreover $\inquotes{\psi}$ is the initial state,
by \defref{NFA}, there is an accepting run
in $\NFApsi$, leading either to $\inquotes{\top}$ or $q_e$.
\end{proof}

These preliminary results about $\NN_\psi$ allow us to prove
correctness of our product construction.
First, we relate paths in $\smash{\NN_\BB^\psi}$ 
to symbolic runs $\sigma$ and consistent words $w$ accepted by $\NN_\psi$.
Below, given a path $\pi$ in $\NN_\BB^\psi$ of the form
\begin{equation}
\label{eq:pi}
\pi\colon
(b_0',q_0,\varphi_0)
\goto{a_0} (b_0, q_1,\varphi_1)
\goto{*} (b_n, q_{n+1},\varphi_{n+1})
\end{equation}
we write $\sigma(\pi)$ for the symbolic run
$\sigma\colon b_0 \goto{a_1} b_1 \goto{*} b_n$
(ignoring the initial step in $\pi$).

\begin{numberedlemma}{\ref{lem:PC sat}}
Let $\psi \in \LBC$.
\begin{enumerate}
\item
For a path $\pi$ of the form \eqref{pi}
in $\NN_\BB^\psi$ there is a run in $\NN_\psi$ labeled $w$ such that 
$\varphi_{n+1} \equivBC h(\sigma(\pi),w)$,
$\hist(\sigma(\pi),w)$ is satisfiable, and
$w$ is consistent with $\sigma(\pi)$.
\item
If a word $w$ is accepted by $\NN_\psi$ and consistent with symbolic run
$\sigma$ such that
$\hist(\sigma,w)$ is satisfiable, 
there is a path $\pi$ of the form \eqref{pi}
in $\smash{\NN_\BB^\psi}$ such that $\sigma = \sigma(\pi)$ and $\varphi_{n+1} \equivBC \hist(\sigma,w)$.
\end{enumerate}
\end{numberedlemma}
\begin{proof}
\begin{enumerate}[{label=(\arabic*)}]
\item
By induction on $n$.
If $n\,{=}\,0$ then $\pi$ consists of the single step
$(b_0',q_0,C_{\alpha_0}) \goto{a_0} (b_0,q_1,\varphi_1)$
and $\sigma$ consists only of state $b_0$. 
By \defref{product construction},
this step labeled $a_0$ exists because $\update(\bigwedge C_{\alpha_0}, a_0) \wedge constr(\varsigma_0) = \bigwedge (C_{\alpha_0} \wedge constr(\varsigma_0))$ is satisfiable,
for some $\varsigma_0$ such that $q_0 \goto{\varsigma_0} q_1$ and
$\varsigma_0$ is consistent with $b_0' \goto{a_0} b_0$.
The formula $\varphi_1$ must satisfy
$\varphi_1 \equivBC \bigwedge (C_{\alpha_0} \wedge constr(\varsigma_0))$, 
as $\guard(a_0)=\top$.
For $w = \varsigma_0$ we indeed have 
$\hist(\sigma,w) = \bigwedge (C_{\alpha_0} \wedge constr(\varsigma_0))$, and
$\varsigma_0$ is consistent with $b_0' \goto{a_0} b_0$ by construction,
so the claim holds.

In the inductive step, consider a path
$\pi \colon p_0 \to^* p_{n+1} \goto{a} p_{n+2}$
for $p_0$ the initial node of $\smash{\NN_\BB^\psi}$
and $p_i = (b_{i-1}, q_{i},\varphi_{i})$ for all $i$, $1\,{\leq}\,i\,{\leq}\,n\,{+}\,2$.
Let $\sigma = \sigma(\pi)$ 
be the symbolic run $b_0 \goto{*} b_{n} \goto{a} b_{n+1}$. 
By the induction hypothesis, 
there is a run
$q_0 \goto{\varsigma_0} q_1 \goto{\varsigma_2} \dots \goto{\varsigma_n} q_{n+1}$ 
in $\NN_\psi$ such that 
$w' = \varsigma_0 \dots \varsigma_n$ is consistent with $\sigma|_n$, 
$\hist(\sigma|_n,w')$ is satisfiable, and
$\varphi_{n+1} \equivBC \hist(\sigma|_n,w')$ ($\star$).
Since there is an edge 
$(b_n, q_{n+1},\varphi_{n+1}) \goto{a} (b_{n+1}, q_{n+2},\varphi_{n+2})$, 
by \defref{product construction} there must be a transition
$q_{n+1} \goto{\varsigma_{n+1}} q_{n+2}$ in $\NN_\psi$, 
such that $\varphi_{n+2} \equivBC \update(\varphi_{n+1}, a) \wedge constr(\varsigma_{n+1})$ ($\star\star$), 
the formula $\varphi_{n+2}$ is satisfiable,
and
$\varsigma_{n+1}$ is consistent with $b_n \goto{a} b_{n+1}$.
Thus, as $w'$ is consistent with $\sigma|_n$,
$w$ is consistent with $\sigma$.
Let $C = constr(\varsigma_{n+1})$.
Note that as $\psi \in \LBC$, we must have $C \subseteq \CC$, so that
by \defref{finite summary} and ($\star$),
\begin{align*}
\update(\varphi_{n+1}, a) \wedge C 
& \equivBC \update(\hist(\sigma|_n,w'), a) \wedge C\\
&= \hist(\sigma,w)
\end{align*}
and the two formulas are equisatisfiable.
From ($\star\star$) it follows that
$\varphi_{n+2} \equivBC \hist(\sigma,w)$ and $\hist(\sigma,w)$ is satisfiable because so is $\varphi_{n+2}$.
\item 
By induction on the length $n$ of $\sigma$.
If $n\,{=}\,0$ then $\sigma$ is empty and
$w = \varsigma$ for some $\varsigma \in \Sigma$.
By assumption $w$ is consistent with the step $b_0' \goto{a_0} b_0$, and 
$\hist(\sigma,w) = \bigwedge (C_{\alpha_0} \cup constr(\varsigma))$ is satisfiable. 
Thus, by \defref{product construction} there is a step
$(b_0', q_0, C_{\alpha_0}) \goto{a_0} (b_0, q_f, \varphi_1)$ and we have
$\varphi_1 \equivBC \update(\bigwedge C_{\alpha_0},a_0) \wedge constr(\varsigma) = \bigwedge (C_{\alpha_0} \cup constr(\varsigma))$, using the definition of $a_0$.

In the inductive step, $\sigma$ has the form $b_0 \goto{*} b_n \goto{a} b_{n+1}$, and 
$w = \varsigma_0 \cdots \varsigma_n\varsigma_{n+1}$
is accepted by $\NN_\psi$, such that 
$\hist(\sigma,w)$ is satisfiable and $w$ is consistent with $\sigma$.
Hence $w'= \varsigma_0 \cdots \varsigma_n$ is consistent with $\sigma|_n$.
By the induction hypothesis, $\NN_\BB^\psi$ has a 
node $p_{n+1} = (b_n, q_{n+1},\varphi_{n+1})$
and a path 
$\pi \colon p_0 \to^* p_{n+1}$
such that $\varphi_{n+1} \equivBC \hist(\sigma|_n,w')$ ($\star$).
The constraint set $C = constr(\varsigma_{n+1})$ satisfies $C \subseteq \CC$ 
because $\psi \in \LBC$.
Therefore, by the properties of $\equivBC$ (\defref{finite summary}) and ($\star$),
\begin{align*}
\update(\varphi_{n+1}, a) \wedge C &\equivBC \update(\hist (\sigma|_n,w'), a) \wedge C \\
&= \hist(\sigma,w)
\end{align*}
and as $\hist(\sigma,w)$ is satisfiable, so is $\update(\varphi_{n+1}, a) \wedge C$.
By assumption
$\varsigma_{n+1}$ is consistent with $b_n \goto{a} b_{n+1}$.
Therefore, $\NN_\BB^\psi$ has a node $p' = (b_{n+1}, q_{n+2}, \varphi_{n+2})$ 
such that $\varphi_{n+2} \equivBC \update(\varphi_{n+1}, a) \wedge C$
and an edge $p_n \goto{a} p'$ can be appended to $\pi$.
\qedhere
\end{enumerate}
\end{proof}

\subsection{Gap-order Constraints}
\label{sec:gc appendix}

The following result about quantifier elimination of GC formulas is known
(see e.g., \cite[Sec. 3]{Revesz93}), but we sketch the procedure for the sake of self-containedness.
\begin{lemma}
\label{lem:gc qe}
For a GC formula $\phi$ over variables $X \cup \{y\}$ and constants $\Const$, 
there some GC formula $\phi'$ over $X$ and $\Const$
such that $\exists y. \phi \equiv \phi'$.
\end{lemma}
\begin{proof}
Initially, one can convert the input formula to disjunctive normal form, and perform quantifier elimination separately on every disjunct (which is a conjunction of GCs).
Thus, let $\phi$ be a conjunction of GCs. One can apply the following procedure:
let $\psi$ be the conjunction of GCs that is obtained by adding all upper bound inequalities for $y$ in $\phi$
to all lower bound inequalities for $y$. Moreover, let $\phi'$ be obtained from $\phi$
by removing all GCs from $\phi$ that mention $y$. Then $\phi' \wedge \psi$ is a quantifier-free GC-formula over variables $X$ and constants $\Const$, and satisfies
$\phi' \wedge \psi \equiv \phi$. 
\end{proof}

\begin{numberedtheorem}{\ref{thm:gap order}}
$(\Phi_\GC,\equivGC)$ is a finite summary for $\BB$ and $\CC$.
\end{numberedtheorem}
\begin{proof}
To show that $\Phi_\GC$ is a history set, we first note that for all
$C \subseteq \CC$, 
$\bigwedge (C_{\alpha_0} \cup C)$ is in $\GC_\Const$.
Now suppose $(b,\phi) \in \Phi_\GC$ and $b \goto{a} b'$.
We can write 
$\update(\phi,a) \wedge C = 
\exists \vec U. \phi(\vec U) \wedge \chi$ for some GC-formula $\chi$ over $U \cup V$ and $\Const$, and from quantifier elimination we obtain
a GC-formula $\phi'$ with $\phi' \equiv \exists \vec U. \phi(\vec U) \wedge \chi$ over $V$ and $\Const$, with $\phi' \in \GC_\Const$.

It remains to verify \defref{finite summary}:
Condition (1) holds by definition of $\equivGC$ and decidability of linear arithmetic.
For Condition (2), suppose  $\phi \equivGC \psi$, so $\cutoff[K]{\phi} \equiv \cutoff[K]{\psi}$. 
Equisatisfiability of $\phi$ and $\psi$ follows from  \lemref{cutoff} (1).
We can write
$\update(\phi,a) \wedge \bigwedge C = \exists \vec U. \phi(\vec U) \wedge \chi$ and
$\update(\psi,a) \wedge \bigwedge C = \exists \vec U. \psi(\vec U) \wedge \chi$
for some GC-formula $\chi$ over $U \cup V$ and $\Const$.
With \lemref{cutoff} (2) we thus obtain
\begin{align*}
\cutoff[K]{\exists \vec U. \phi(\vec U) \wedge \chi}
 &\equiv \cutoff[K]{\exists \vec U. \cutoff[K]{\phi(\vec U)} \wedge \cutoff[K]{\chi}}\\
 &\equiv \cutoff[K]{\exists \vec U. \cutoff[K]{\psi(\vec U)} \wedge \cutoff[K]{\chi}}\\
&\equiv \cutoff[K]{\exists \vec U. \psi(\vec U) \wedge \chi}
\end{align*}
Finally, $(\Phi_\GC,\equivGC)$ has finitely many equivalence classes 
because there are only finitely many GC-formulas that differ in their
$K$-bounded approximation.
\qedhere
\end{proof}

\subsection{Bounded Lookback}

\begin{numberedtheorem}{\ref{thm:bounded lookback}}
If $(\BB,\CC)$ has $K$-bounded lookback for some $K$, it has \property. 
\end{numberedtheorem}
\begin{proof}
Let $\Psi$ be the set of formulas with free variables $V$, 
quantifier depth at most $K\cdot |V|$, and using as vocabulary all constraints in 
$\CC$, $C_{\alpha_0}$, and
guards of $\BB$.
Note that as a set of formulae with bounded quantifier depth over a finite set of constraints, $\Psi$ is finite up to equivalence.
To prove that $B \times\Psi$ is a history set, 
we show that for every history constraint $h(\sigma,\cseq C)$ of $\BB$ and $\CC$
such that $\sigma$ has final state $b$
there is some $(b,\phi) \in B \times\Psi$ with $\phi \equiv h(\sigma,\cseq C)$.
This implies the claim by \lemref{simple finite summary}.
The proof is by induction of $\sigma$.
If $\sigma$ is empty,  $h(\sigma,\cseq C) = \bigwedge (C_{\alpha_0} \cup C_0)$ for some $C_0$, so that $h(\sigma,\cseq C) \in \Psi$, and hence
$(b_0, h(\sigma,\cseq C)) \in \Psi$.
Otherwise, $\sigma$ is of the form $b_0 \to^* b_n \goto{a} b_{n+1}$, and
$\hist(\sigma, \cseq C) = \update(\hist(\sigma|_{n}, \cseq C|_{n}), a) \wedge C$ for some $a\in \AA$ and $C\subseteq \CC$.
By induction hypothesis there is some $\phi \in \Psi$ with $\phi \equiv \hist(\sigma|_{n}, \cseq C|_{n}), a)$.
By \defref{update}, we can thus write $\hist(\sigma, \cseq C)$ as 
$\phi' = \exists \vec U. \phi(\vec U) \wedge \chi$ 
for some quantifier free formula  $\chi$.
Let $\eqc{\phi'}$ be obtained from $\phi'$ by eliminating all equality literals $x=y$ in $\phi'$ and uniformly substituting all variables in an equivalence class by a representative.
Since $(\BB,\CC)$ has $K$-bounded lookback, and $\eqc{\phi'}$ encodes a part of
$\eqc{G_{\sigma,\CC}}$, 
$\eqc{\phi'}$ is equivalent to a formula $\psi$ of quantifier depth at most $K\cdot |V|$ over the same vocabulary, 
obtained by dropping irrelevant literals and existential quantifiers.
Hence $(b_{n+1}, \psi) \in B \times \Psi$.
\end{proof}

\begin{numberedlemma}{\ref{lem:feedback free bounded lookback}}
If $(\BB, \CC)$ is feedback-free then it has $2|V|$-bounded lookback.
\end{numberedlemma}
\begin{proof}
The proof uses the main ideas from \cite[Lem. 5.4]{DDV12}.
We consider the graph $\eqc{G_{\sigma,\CC}}$ obtained from $G_{\sigma,\CC}$ 
by collapsing equivalence classes of $\equiv_E$.
Let $\sigma$ be a symbolic run of $\BB$ of length $n$, and $\cseq C$ a
verification constraint sequence.
We can write $h(\sigma, \cseq C)$ as $\exists V_0 \dots V_{n-1}.\,\chi$, for some 
quantifier-free formula $\chi$ (using the variables $V_i$ as the set $U$ in \defref{update}). Let $\eqc{\chi}$ be obtained from $\chi$
by eliminating all equality literals $x_i = y_j$ for $x_i, y_j \in \mc V$
and uniformly substituting all variables in $\mc V$ by their representative
with respect to $\equiv_E$.
Consider now a conjunction $\nu$ of a subset of literals in $\eqc{\chi}$, 
such that the restriction $\eqc{G_\nu}$ of $\eqc{G_{\sigma,\CC}}$ to variable equivalence
classes in $\nu$ is connected, and let $X \subseteq V$ be the set of variables 
$X = \{x \mid x_i \text{ occurs in $\nu$ for some $i$}\}$.
We show by induction on $|X|$ that $\eqc{G_\nu}$ is a tree of depth at most $|X|$. 
Then it follows that the length of all paths in $\eqc{G_\nu}$ is bounded by $2|X|$.
If $|X|=0$, i.e., $X = \emptyset$, there can be no paths in $\eqc{G_\nu}$, so the claim is vacuously true.
Otherwise, we distinguish two cases: if the length of the longest path in  $\eqc{G_\nu}$ is at most $|X|$,
the claim is obvious as well. Otherwise, $\eqc{G_\nu}$ contains a path of length greater than $|X|$, 
so there must be a variable $x\in X$ and instants $i,j \leq n$ with $i \neq j$ 
such that both $x_i$ and $x_j$ occur in $\nu$. 
Since $\nu$ is connected, there must be a path from $x_i$ to $x_j$ in $\eqc{G_\nu}$.
By the property of feedback freedom, there must be a variable $y_k$ such that
$span(\eqc{x_i}) \cup span(\eqc{x_j}) \subseteq span(\eqc{y_k})$.
We may choose $y_k$ such that $span(\eqc{y_k})$ is maximal with this property.
Then one observes that there is no $y_m$ in $\nu$ with $k \neq m$ and $\eqc{y_k} \neq \eqc{y_m}$ $(\star)$:
indeed, since $\eqc{G_\nu}$ is connected, 
there would be a path between $\eqc{y_k}$ and $\eqc{y_m}$, 
and by feedback freedom such a path must contain a node $\eqc{z}$ such that
$span(\eqc{y_k}) \cup span(\eqc{y_m}) \subseteq span(\eqc{z})$, contradicting maximality of $span(\eqc{y_k})$.
Let $H_1, \dots, H_p$ be the connected components of $\eqc{G_\nu}$ without $\eqc{y_k}$.
By $(\star)$, the set of variables occurring in $H_1, \dots, H_p$ is $X\setminus\{y\}$.
Therefore we can apply the induction hypothesis to conclude that $H_q$ is a tree of depth at most $|X|-1$ for all $q$,
$1\leq q \leq p$, and it follows that $\eqc{G_\nu}$ is a tree of depth at most $|X|$.
\end{proof}
%

\subsection{Modularity}


\begin{lemma}
\label{lem:seq decompose}
Suppose $\BB$ is sequentially decomposable into a {\mydds}s $\BB_1$ 
with finite summary $(\Phi_1, \equiv)$
and $\BB_2$ with finite summary $(\Phi_2, \equiv)$, both with respect to $\CC$.
Then $(\BB, \CC)$ has finite summary.
\end{lemma}
\begin{proof}
Since $(\Phi_1, \equiv)$ and $(\Phi_2, \equiv)$ are finite summaries,
$\Phi_1$ and $\Phi_2$ are finite history sets.
Let the given DDSAs be
$\BB_1 = \langle B_1, b_0, \AA_1, T_1, \{b\}, V, \alpha_0, \guard\rangle$ and
$\BB_2 = \langle B_2, b, \AA_2, T_2, F, V\cup U, \alpha_U, \guard\rangle$.
Let $\Phi_1^b \subseteq \Phi_1$ be the set of all $(b,\varphi) \in \Phi_1$ 
for some $\varphi$, and the set $\Phi$ given by
\[\Phi_1 \cup \{ (b',\exists \vec U.\varphi_1(\vec U) \wedge \varphi_2) \mid 
(b,\varphi_1) \in \Phi_1^b\text{, }(b',\varphi_2) \in \Phi_2 \}.
\]
We show that $\Phi$ is a finite history set for $(\BB,\CC)$,
so that the claim follows from \lemref{simple finite summary}.
Finiteness is immediate from finiteness of $\Phi_1$ and $\Phi_2$.

To verify that $\Phi$ is a history set for $\BB$, we use \lemref{history set}.
First, for all $C \subseteq \CC$, there is some $(b_0,\phi_0) \in \Phi_1$ 
such that $\phi_0 \equiv \bigwedge (C_{\alpha_0} \cup C)$ by finite summary of $\BB_1$,
and by definition of $\Phi$, we have $(b_0,\phi_0) \in \Phi$.
Second, let $(b_1,\phi)\in\Phi$, $b_1 \goto{a} b_2$, and $C \subseteq \CC$.
Two further cases can be distinguished. If $a \in \AA_1$ then by finite summary of
$\BB_1$, there is some $(b_2,\phi')\in \Phi_1$ such that 
$\phi' \equiv \update(\phi,a) \wedge C$, hence $(b_2,\phi')\in \Phi$.
Otherwise $a \in \AA_2$, and we can write
$\phi = \exists \vec U.\varphi_1(\vec U) \wedge \varphi_2$ such that
$\varphi_1 \in \Phi_1^b$ and $\varphi_2 \in \Phi_2$.
Then there is a formula $\chi$ with free variables $X$ and $V$ such that we can write
\begin{align*}
\update(\phi,a) &\wedge C = \exists \vec X. \phi(\vec X) \wedge \chi \\
&= \exists \vec X. 
(\exists \vec U.\varphi_1(\vec U) \wedge \varphi_2)(\vec X) \wedge \chi \\
&\equiv 
\exists \vec U.\varphi_1(\vec U) \wedge \exists \vec X.(\varphi_2(\vec X) \wedge \chi) \\
&\equiv 
\exists \vec U.\varphi_1(\vec U) \wedge \update(\phi_2,a) \wedge C & =:\psi
\end{align*}
where the quantifier shift in the third step is allowed because $\vec X$ does not occur in $\varphi_1$.
Since $\Phi_2$ is a history set, there must be some $(b_2,\phi_2')\in \Phi_2$
such that $\phi_2' \equiv \update(\phi_2,a) \wedge C$.
Therefore, we have $\psi \equiv \exists \vec U.\varphi_1(\vec U) \wedge \phi_2' =:\phi'$.
By definition of $\Phi$, it contains $(b_2,\phi')$, which proves the claim.
\end{proof}

\begin{lemma}
\label{lem:var decompose}
Suppose $(\BB,\CC)$ is variable-decomposable into {\mydds}s 
$(\BB_1, \CC|_{V_1})$ and $(\BB_2, \CC|_{V_2})$,
such that
$(\BB_i, \CC|_{V_i})$ has finite summary $(\Phi_i, \sim_i)$, 
for $i\in \{1,2\}$.
Then $(\BB, \CC)$ has finite summary.
\end{lemma}
\begin{proof}
Suppose $\BB$ is variable-decomposable into the {\mydds}s
$\BB_1 = \langle B_1, b_0, \AA, T, F, V_1, \alpha_0|_{V_1}, guard_1\rangle$ and
$\BB_2 = \langle B_2, b_0, \AA, T, F, V_2, \alpha_0|_{V_2}, guard_2\rangle$.
By assumption, $\BB_1$ has a finite summary $(\Phi_1, \sim_1)$ and 
$\BB_2$ has a finite summary $(\Phi_2, \sim_2)$.
Let
\[\Phi = \{ \varphi_1\wedge \varphi_2 \mid 
\varphi_1 \in \Phi_1\text{ and }\varphi_2 \in \Phi_2 \}.
\]
and $\sim$ be defined as 
$\varphi_1\wedge \varphi_2 \sim \psi_1\wedge \psi_2$ iff  $\phi_1 \sim_1 \psi_1$ and $\phi_2 \sim_2 \psi_2$.
To show that $(\Phi, \sim)$ is a \property for $\BB$ with respect to $\CC$,
i.e., for every symbolic run $\sigma$ of $\BB$ and state constraint sequence $\cseq C$
over $\CC$, there is some $\varphi \in \Phi$ such that
$\hist(\sigma, \cseq C) \sim \varphi$.
The proof is by induction on $\sigma$.
If $\sigma$ is empty then $\hist(\sigma, \cseq C) = \bigwedge (C_{\alpha_0} \cup C_0)$
can be split into $ \bigwedge (C_{\alpha_0|_{V_1}} \cup C_0|_{V_1}) \wedge  \bigwedge (C_{\alpha_0|_{V_2}} \cup C_0|_{V_2}) = \hist(\sigma_1, \cseq C|_{V_1}) \wedge \hist(\sigma_2, \cseq C|_{V_2})$ for $\sigma_i$ the empty run in $\BB_i$.
Otherwise, 
$\hist(\sigma, \cseq C) = \update(\hist(\sigma|_{n-1}, \cseq C|_{n-1}), a) \wedge C$ for some $a\in \AA$ and $C \subseteq \CC$.
By the induction hypothesis, there is some $\varphi_1\wedge \varphi_2 \in \Phi$ such that
$\hist(\sigma|_{n-1}, \cseq C|_{n-1}) \equiv \varphi_1\wedge \varphi_2$, 
$\varphi_1 \in \Phi_1$, and $\varphi_2 \in \Phi_2$.
For a suitable formula $\chi$ 
we can hence write $\update(\hist(\sigma|_{n-1}, \cseq C|_{n-1}), a) \wedge C$ as
\begin{align*}
\exists \vec U.&\hist(\sigma|_{n-1}, \cseq C|_{n-1})(\vec U) \wedge \chi\\
&\equiv\exists \vec U. \varphi_1(\vec U) \wedge \varphi_2(\vec U) \wedge \chi\\
&\equiv\exists \vec U_1. \varphi_1(\vec U_1) \wedge \chi_1 \wedge 
\exists \vec U_2. \varphi_2(\vec U_2) \wedge \chi_2 &(\star)
\end{align*}
We have
$\equiv\exists \vec U_i. \varphi_i(\vec U_i) \wedge \chi_i = 
\update(\varphi_i, a) \wedge C|_{V_i}$ for both $i\in \{1,2\}$,
and by finite summary there is some $\varphi_i' \in \Phi_i$ such that
$\varphi_i' \equiv \update(\varphi_i, a) \wedge C|_{V_i}$.
Therefore, $(\star)$ can be written as $\varphi_1' \wedge \varphi_2' \in \Phi$.

Condition (1) of \defref{finite summary} follows from the respective properties of $\sim_1$ and $\sim_2$.
For Condition (2), suppose $\phi_1 \wedge \phi_2 \sim \psi_1 \wedge \psi_2$, 
so $\phi_1 \sim_1 \psi_1$ and $\phi_2 \sim_2 \psi_2$.
If $\phi_1 \wedge \phi_2$ is satisfiable then so are both conjuncts,
so that by finite summary of $\BB_i$ also $\psi_i$ is satisfied by a valuation
$\alpha_i$. As $\alpha_1$ has domain $V_1$ and $\alpha_2$ has domain $V_2$, 
we can form their union and this valuation satisfies $\psi_1 \wedge \psi_2$

Now let $a\in \AA$ and $C \subseteq \CC$, 
which can be split into $C = C|_{V_1} \wedge C|_{V_2}$.
By the assumption of variable decomposition also 
$\trans a(\vec U, \vec V) = \Delta_1 \wedge \Delta_2$ for some formulas $\Delta_1$ over $U_1 \cup V_1$ and $\Delta_2$ over $U_2 \cup V_2$.
Hence $\update(\phi_1 \wedge \phi_2, a) \wedge C 
\equiv\exists \vec U. \varphi_1(\vec U) \wedge \varphi_2(\vec U) \wedge \Delta_1 \wedge \Delta_2 \wedge C_1 \wedge C_2 
\equiv\exists \vec U_1. \varphi_1(\vec U_1) \wedge \Delta_1 \wedge C_1 \wedge 
\exists \vec U_2. \varphi_2(\vec U_2) \wedge \Delta_2 \wedge C_2 = 
\update(\phi_1, a) \wedge C_1 \wedge \update(\phi_2, a) \wedge C_2$.
Similarly, $\update(\psi_1 \wedge \psi_2, a) \wedge C \equiv 
\update(\psi_1, a) \wedge C_1 \wedge \update(\psi_2, a) \wedge C_2$.
By finite summary of $\BB_1$ and $\BB_2$,
$\update(\phi_i, a) \wedge C_i \sim_i \update(\psi_i, a) \wedge C_i$ for $i\in \{1,2\}$. Thus 
$\update(\phi_1 \wedge \phi_2, a) \wedge C \equiv \update(\psi_1 \wedge \psi_2, a) \wedge C$ holds by definition.
\end{proof}

\noindent
\lemsref{seq decompose}{var decompose} then imply the decomposition theorem
(\thmref{decompose}).


\section{Examples}

\subsection{NFA for a Formula}
\label{sec:example:NFA}

\begin{example}
\label{exa:nfa}
Let $\psi = \Diamond (b \wedge \langle \cdot\rangle(x-y \geq 2))$.
We abbreviate the constraint $x-y \geq 2$ by $c$.

\medskip
\noindent
\textit{--first, note that for any $\phi$:}

\noindent
\begin{tabular}{@{~}l@{~}l}
& $\delta(\inquotes{\langle\cdot\rangle \phi})$ = $\{(\inquotes{\phi}, \{\neg \last\}), (\inquotes{\bot}, \{\last\})\}$ \\
\end{tabular}

\smallskip
\noindent
\textit{--this captures the fact that either $\phi$ is true next and the word on $2^\Sigma$ didn't end (i.e. the symbol $\neg \last$ does not appear), or $\phi$ is false as the word ended. It is used below for $\phi=c$:}
\bigskip

\noindent
\begin{tabular}{@{~}l@{~}l}
& $\delta(\inquotes{b \wedge \langle \cdot\rangle c})$ = $\delta(\inquotes{b}) \wedge \delta(\inquotes{\langle \cdot\rangle c})$ \\
&= $\{(\inquotes{\top}, \{b\}), (\inquotes{\bot}, \{\neg b\})\} \wedge \{(\inquotes{c}, \{\neg \last\}), (\inquotes{\bot}, \{\last\})\}$ \\
&= $\{(\inquotes{c}, \{b, \neg \last\}), (\inquotes{\bot}, \{\neg b, \neg \last\}), (\inquotes{\bot},$ \\
&\qquad  $\{\neg b,\last\}),(\inquotes{\bot}, \{b,\last\})\}$
\end{tabular}

\smallskip
\noindent
\textit{--the four elements above are obtained as
$R_1\,{\wedge}\, R_2 = \{ (\psi_1\,{\wedge}\,\psi_2, \varsigma_1\,{\cup}\, \varsigma_2) \mid (\psi_1, \varsigma_1) \inn R_1, (\psi_2, \varsigma_2) \inn R_2 \}$ by definition. The result for 
$\delta(\inquotes{b \wedge \langle \cdot\rangle c})$ is used below: }
\smallskip

\noindent
\begin{tabular}{@{~}l@{~}l}
& $\delta(\inquotes{\psi})$ = $\delta(\inquotes{b \wedge \langle \cdot\rangle c}) \vee \delta(\inquotes{\langle\cdot\rangle \psi})$\\
&= $\{(\inquotes{c \vee \psi}, \{b, \neg \last\}), (\inquotes{\psi}, \{\neg b, \neg \last\}), (\inquotes{\bot}, \{\neg b, \last\}),$ \\
&\qquad  $(\inquotes{\bot}, \{b,\last\})\}$ \\
\end{tabular}

\smallskip
\noindent
\textit{--these four elements are obtained by applying the definition of $R_1 \vee R_2 = \{ (\psi_1 \vee \psi_2, \varsigma_1 \cup \varsigma_2) \mid (\psi_1, \varsigma_1) \in R_1, (\psi_2, \varsigma_2) \in R_2 \}$. They correspond to four edges from $(\inquotes{\psi})$. }
\bigskip

\noindent
\begin{tabular}{@{~}l@{~}l}
& $\delta(\inquotes{c{\vee}\psi})$ = $\delta(\inquotes{c}) \vee \delta(\inquotes{\psi})$\\
&= $\{(\inquotes{\top}, \{c\}), (\inquotes{\bot}, \emptyset) \} \vee \delta(\psi)$ \\
&= $\{(\inquotes{\top}, \{c, b, \last\}),
 (\inquotes{\top}, \{c, b, \neg \last\}), 
 (\inquotes{\top}, \{c, \neg b, \last\}) $\\
&\qquad $
 (\inquotes{\top}, \{c, \neg b, \neg \last\})\} \cup \delta(\psi)$
\end{tabular}

\smallskip
\noindent
\textit{--these eight elements are obtained by applying the definition of $R_1 \vee R_2$. They correspond to eight edges from $(\inquotes{c{\vee}\psi})$, of which four reach $q_f=(\inquotes{\top})$.}
\bigskip

\noindent
\textit{--the remaining edges (self loops) are due to $\delta(\inquotes{\top}) = \{(\inquotes{\top},\emptyset)\}$ and $\delta(\inquotes{\bot}) = \{(\inquotes{\bot},\emptyset)\}$.}
The computed edges could be combined in the following NFA:

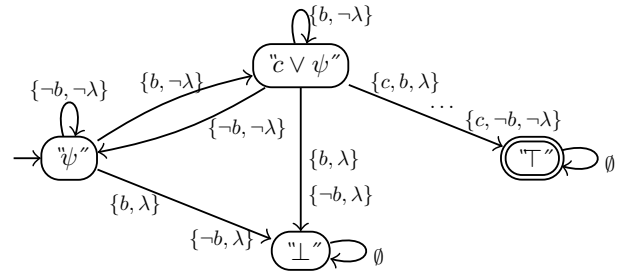
\resizebox{\columnwidth}{!}{
\centering
\begin{tikzpicture}[scale=.7]
\tikzstyle{formula} = [draw,rectangle, rounded corners, inner sep=4pt, scale=.55]
\tikzstyle{goto} = [->]
\tikzstyle{action} = [scale=.45]
\node[formula] (1)  {$\inquotes{\psi}$};
\node[formula] at (2.5, 1) (2) {$\inquotes{c\vee \psi}$};
\node[formula] at (2.5, -1) (bot) {$\inquotes{\bot}$};
\node[formula, double] (top) at (5,0) {$\inquotes{\top}$};
\draw[goto] ($(1) + (-.6,0)$) -- (1);
\draw[goto] (1) to[loop above] node[action, yshift=-1mm]{$\{\neg b, \neg \last\}$} (1);
\draw[goto, bend left=10] (1) to node[above,action] {$\{b,\neg \last\}$} (2);
\draw[goto, bend left=10] (2) to node[below,action, near start, xshift=4mm] {$\{\neg b,\neg \last\}$} (1);
\draw[goto] (1) to node[below,action, near start,xshift=-1mm]{$\{b,\last\}$} node[below,action, near end,xshift=-1mm]{$\{\neg b, \last\}$} (bot);
\draw[goto] (2) to node[right,action]{$\{b,\last\}$} 
 node[right,action, near end]{$\{\neg b,\last\}$} (bot);
\draw[goto] (2) to 
 node[above,action, near start,xshift=3mm]{$\{c, b, \last\}$}
 node[above,action,xshift=3mm]{$\dots$}
 node[above,action, very near end,xshift=5mm]{$\{c, \neg b, \neg \last\}$}
 (top);
\draw[goto] (2) to[loop above] node[action, right,yshift=-1mm]{$\{b, \neg \last\}$} (2);
\draw[goto] (bot) to[loop right] node[action, yshift=-1mm]{$\emptyset$} (bot);
\draw[goto] (top) to[loop right] node[action, yshift=-1mm]{$\emptyset$} (top);
\end{tikzpicture}
}
In order to obtain the NFA $\NN_\psi$ without $\last$, we use \defref{NFA}. If we moreover omit the non-accepting sink state $\inquotes{\bot}$,
we obtain the following automaton:\\
\resizebox{\columnwidth}{!}
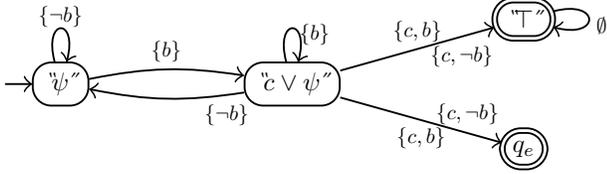
{
\centering
\begin{tikzpicture}[scale=.7]
\tikzstyle{formula} = [draw,rectangle, rounded corners, inner sep=4pt, scale=.55]
\tikzstyle{goto} = [->]
\tikzstyle{action} = [scale=.45]
\node[formula] (1)  {$\inquotes{\psi}$};
\node[formula] at (2.5, 0) (2) {$\inquotes{c\vee \psi}$};
\node[formula, double] (top) at (5,.7) {$\inquotes{\top}$};
\node[formula, double] (qe) at (5,-.7) {$q_e$};
\draw[goto] ($(1) + (-.6,0)$) -- (1);
\draw[goto] (1) to[loop above] node[action, yshift=-1mm]{$\{\neg b\}$} (1);
\draw[goto, bend left=10] (1) to node[above,action] {$\{b\}$} (2);
\draw[goto, bend left=10] (2) to node[below,action, near start, xshift=4mm] {$\{\neg b\}$} (1);
\draw[goto] (2) to 
 node[above,action]{$\{c, b\}$}
 node[below,action, near end, xshift=1mm]{$\{c, \neg b\}$}
 (top);
\draw[goto] (2) to 
 node[below,action]{$\{c, b\}$}
 node[above,action, near end, xshift=1mm]{$\{c, \neg b\}$}
 (qe);
\draw[goto] (2) to[loop above] node[action, right,yshift=-1mm]{$\{b\}$} (2);
\draw[goto] (top) to[loop right] node[action, yshift=-1mm]{$\emptyset$} (top);
\end{tikzpicture}
}
While the additional final state does not add anything in this case,
it becomes relevant for formulas $\psi$ containing $\Box$.
\end{example}

\subsection{Product Construction}
\label{sec:example:pc}

\newcommand{\sstart}{start}
\newcommand{\smain}{main}
\newcommand{\schange}{change}
\newcommand{\send}{end}
\newcommand{\ssold}{sold}
We consider the auction system $\BB$ from \exaref{auction}, 
and add names to identify states in the following construction:\\
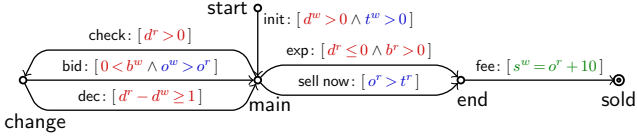
\begin{tikzpicture}[node distance=40mm]
\tikzstyle{action}=[scale=.6]
\node[state] (1)  {};
\node[state, below of=1, yshift=24mm] (2) {};
\node[state, left of=2, xshift=-12mm] (3) {};
\node[state, right of=2, xshift=5mm] (4) {};
\node[state, right of=4, double, xshift=-5mm] (5) {};
\node[scale=.8, xshift=-5mm] at (1) {$\m{\sstart}$};
\node[scale=.8, xshift=2mm, yshift=-4mm] at (2) {$\m{\smain}$};
\node[scale=.8, xshift=2mm, yshift=-7mm] at (3) {$\m{\schange}$};
\node[scale=.8, yshift=-3mm, xshift=2mm] at (4) {$\m{\send}$};
\node[scale=.8, yshift=-3mm] at (5) {$\m{\ssold}$};
\draw[edge] (1) to node[right,action, very near start]
{$\m{init}\colon[\,\RED{d^w\,{>}\,0}\wedge \BLUE{t^w\,{>}\,0}\,]$} (2);
\draw[edge] (3) to node[above,action] 
 {$\m{bid}\colon[\,\RED{0\,{<}\,b^w} \wedge \BLUE{o^w\,{>}\,o^r}\,]$} (2);
\draw[edge, rounded corners] (2) -- ($(2) + (-.2,.4)$) -- node[above,action]{$\m{check}\colon[\,\RED{d^r\,{>}\,0}\,]$} ($(3) + (.2,.4)$) -- (3);
\draw[edge, rounded corners] (3) -- ($(3) + (.2,-.4)$) to node[above,action]
{$\m{dec}\colon[\,\RED{d^r\,{-}\,d^w\,{\geq}\,1}\,]$} ($(2) + (-.2,-.4)$) -- (2);
\draw[edge, rounded corners] (2) -- ($(2) + (.2,.2)$)
  -- node[above,action]{$\m{exp}\colon[\,\RED{d^r\,{\leq}\,0} \wedge \RED{b^r\,{>}\,0}\,]$}  ($(4) + (-.2,.2)$)  -- (4);
\draw[edge, rounded corners] (2) -- ($(2) + (.2,-.2)$)
  -- node[above,action]{$\m{sell\: now}\colon[\,\BLUE{o^r\,{>}\,t^r}\,]$}  ($(4) - (.2,.2)$)  -- (4);
\draw[edge, rounded corners] (4) -- node[above,action]{$\m{fee}\colon[\,\GREEN{s^w\,{=}\,o^r\,{+}\,10}\,]$} (5);
\end{tikzpicture}
We want to verify that even if a bidder places a bid above the threshold, she need not win. To that end, we intend to find a witness for
$\psi = \Diamond(b\,{=}\,1 \wedge o\,{>}\,t \wedge \Diamond(\m{sold} \wedge b\,{\neq}\,1))$.
To cover all constraints in $\psi$, we set
$\CC = \{ b\,{=}\,1,\, o\,{>}\,t,\, b\,{\neq}\,1\}$.
We obtain the following NFA $\NN_\psi$ according to \defref{NFA}:
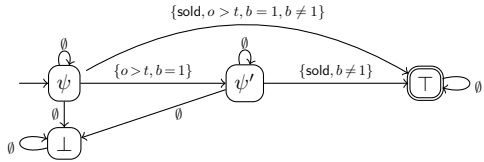
\begin{center}
\begin{tikzpicture}[node distance=30mm]
\tikzstyle{formula} = [draw,rectangle, rounded corners, inner sep=4pt, scale=.8]
\tikzstyle{goto} = [->]
\node[formula] (1)  {$\psi$};
\node[formula, right of=1] (2) {$\psi'$};
\node[formula, below of=1, yshift=20mm] (bot) {$\bot$};
\node[formula, right of=2, double] (top) {$\top$};
\draw[goto] ($(1) + (-.6,0)$) -- (1);
\draw[goto] (1) to[loop above,looseness=5] node[action, yshift=-1mm]{$\emptyset$} (1);
\draw[goto] (1) -- node[action, above]{$\{o\,{>}\,t, b\,{=}\,1\}$} (2);
\draw[goto] (2) to[loop above,looseness=5] node[action, yshift=-1mm]{$\emptyset$} (2);
\draw[goto] (2) -- node[action, above]{$\{\m{sold}, b\,{\neq}\,1\}$} (top);
\draw[goto] (1) -- node[action, left]{$\emptyset$} (bot);
\draw[goto] (2) -- node[action, right, xshift=4mm]{$\emptyset$} (bot);
\draw[goto] (bot) to[loop left] node[action, yshift=-1mm]{$\emptyset$} (bot);
\draw[goto] (top) to[loop right] node[action, yshift=-1mm]{$\emptyset$} (top);
\draw[goto] (psi) to[bend left=30] node[action, above]{$\{\m{sold}, o>t, b=1, b \neq 1\}$} (top);
\end{tikzpicture}
\end{center}
where $\psi' = \psi \vee \Diamond(\m{sold} \wedge b\,{\neq}\,1)$.
As an optimization, for the product construction below both the edge
from $\psi$ to $\top$ (having an unsatisfiable label) and the state $\bot$
(not contributing to accepting runs) can be omitted, so that we obtain
$\NFApsi'$:
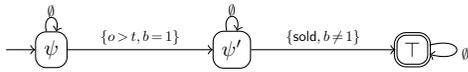
\begin{center}
\begin{tikzpicture}[node distance=30mm]
\tikzstyle{formula} = [draw,rectangle, rounded corners, inner sep=4pt, scale=.8]
\tikzstyle{goto} = [->]
\node[formula] (1)  {$\psi$};
\node[formula, right of=1] (2) {$\psi'$};
\node[formula, right of=2, double] (top) {$\top$};
\draw[goto] ($(1) + (-.6,0)$) -- (1);
\draw[goto] (1) to[loop above,looseness=5] node[action, yshift=-1mm]{$\emptyset$} (1);
\draw[goto] (1) -- node[action, above]{$\{o\,{>}\,t, b\,{=}\,1\}$} (2);
\draw[goto] (2) to[loop above,looseness=5] node[action, yshift=-1mm]{$\emptyset$} (2);
\draw[goto] (2) -- node[action, above]{$\{\m{sold}, b\,{\neq}\,1\}$} (top);
\draw[goto] (top) to[loop right] node[action, yshift=-1mm]{$\emptyset$} (top);
\end{tikzpicture}
\end{center}
We construct $\smash{\NN_\BB^\psi}$
by combining $\NFApsi'$ with $\BB$ as in
\defref{product construction}.
To that end, we first need to determine a summary for $\BB$:
As explained in \exaref{auction2}, finite summary of $\BB$ holds by a
variable decomposition into a GC-\mydds $\BB_1$ over $\{b,d\}$ and a \mydds
$\BB_2$ over $\{o,t,s\}$. Note that also $\CC$ can be split accordingly
into $\CC = \{b\,{=}\,1, b\,{\neq}\,1\} \uplus \{o\,{>}\,t\}$ 
over the respective variables.
$\BB_2$ is subsequently split sequentially into two subsystems, but both
have finite history sets, so that also $\BB_2$ has a finite history set by \thmref{decompose}, and hence a finite summary $(\Phi_2,\equiv)$, for some
$\Phi_2$. On the other hand, being a GC-\mydds with maximal constant 0, 
$\BB_1$ has finite summary $(\Phi_1, \equivGC[1])$.
Hence, following the proof of \thmref{decompose}, for $\BB$ we use the equivalence relation $\sim$ on formulas defined as follows:
$\phi \sim \phi'$ if $\phi|_{\{b,d\}} \equivGC[1] \phi'|_{\{b,d\}}$ and
$\phi|_{\{o,t,s\}} \equiv \phi'|_{\{o,t,s\}}$, where $\phi|_{\{b,d\}}$
is the subformula of $\phi$ mentioning only variables $\{b,d\}$, and
similar for $\phi'$ and $\{o,t,s\}$ (see the proof of \lemref{var decompose} for details).

Our tool computes a product construction with 35 states, of which
\figref{product} shows a part.
For the sake of readability, we omitted all edge labels $\emptyset\in \Sigma$.
\begin{figure}[t]
\begin{tikzpicture}[node distance=12mm]
\tikzstyle{node} = [draw,rectangle split, rectangle split parts=3,rectangle split horizontal, rectangle split draw splits=true, inner sep=3pt, scale=.7, rounded corners]
\tikzstyle{goto} = [->]
\tikzstyle{accepting path} = [red!80!black, line width=.8pt]
\tikzstyle{accepting state} = [fill=red!80!black!15]
\tikzstyle{action}=[scale=.6, black]
\node[node] (0)  {\pcnode{$b_0'$}{$\psi$}{$b\,{=}\,d\,{=}\,o\,{=}\,s\,{=}\,t\,{=}\,0$}};
\node[node, below of=0] (1)
 {\pcnode{$\m{\sstart}$}{$\psi$}{$b\,{=}\,d\,{=}\,o\,{=}\,s\,{=}\,t\,{=}\,0$}};
\node[node, below of=1] (2)
 {\pcnode{$\m{\smain}$}{$\psi$}{$d\,{\geq}\,1\wedge t\,{>}\,0 \wedge b\,{=}\,o\,{=}\,s\,{=}\,0$}};
\node[node, below of=2] (3) {\pcnode{$\m{\schange}$}{$\psi$}{$d\,{\geq}\,1\wedge t\,{>}\,0 \wedge b\,{=}\,o\,{=}\,s\,{=}\,0$}};
\node[node, below of=3, yshift=-5mm] (4) {\pcnode{$\m{\smain}$}{$\psi$}{$d\,{\geq}\,1\wedge b\,{\geq}\,1 \wedge o\,{>}\,0 \wedge s\,{=}\,0\wedge t\,{>}\,0$}};
\node[node, below of=3, xshift=28mm, yshift=5mm] (5) {\pcnode{$\m{\smain}$}{$\psi$}{$t\,{>}\,0\wedge o\,{=}\,b\,{=}\,s\,{=}\,0$}};
\node[node, below of=3, xshift=-28mm, yshift=-24mm] (6) {\pcnode{$\m{\smain}$}{$\psi'$}{$b\,{=}\,1\wedge o\,{>}\,t \wedge d\,{\geq}\,1 \wedge t\,{>}\,0\wedge s\,{=}\,0$}};
\node[node, below of=6, xshift=0mm, yshift=0mm] (7) {\pcnode{$\m{\schange}$}{$\psi'$}{$b\,{=}\,1\wedge o\,{>}\,t \wedge d\,{\geq}\,1 \wedge t\,{>}\,0\wedge s\,{=}\,0$}};
\node[node, below of=7, xshift=0mm, yshift=0mm] (8) {\pcnode{$\m{\smain}$}{$\psi'$}{$b\,{>}\,0\wedge o\,{>}\,t \wedge d\,{\geq}\,1 \wedge t\,{>}\,0\wedge s\,{=}\,0$}};
\node[node, below of=8, xshift=0mm, yshift=-5mm] (9) {\pcnode{$\m{\send}$}{$\psi'$}{$b\,{\geq}\,1\wedge o\,{>}\,t \wedge d\,{\geq}\,1 \wedge t\,{>}\,0\wedge s\,{=}\,0$}};
\node[node, below of=9, xshift=0mm, yshift=0mm,accepting state] (10) {\pcnode{$\m{\ssold}$}{$\top$}{$b\,{\geq}\,2\wedge o\,{>}\,t \wedge d\,{\geq}\,1 \wedge t\,{>}\,0\wedge s\,{=}\,o+10$}};
\node[node, below of=8, xshift=46mm, yshift=3mm] (11) {\pcnode{$\m{\schange}$}{$\psi'$}{$b\,{>}\,0\wedge o\,{>}\,t \wedge d\,{\geq}\,1 \wedge t\,{>}\,0\wedge s\,{=}\,0$}};
\node[node, below of=4, xshift=23mm, yshift=2mm] (12) {\pcnode{$\m{\schange}$}{$\psi$}{$d\,{\geq}\,1\wedge b\,{\geq}\,1\wedge o\,{>}\,0 \wedge t\,{>}\,0 \wedge s\,{=}\,0$}};
\node[node, below of=10, xshift=20mm, yshift=4mm,] (13) {\pcnode{$\m{\ssold}$}{$\psi'$}{$b\,{>}\,1\wedge o\,{>}\,t \wedge d\,{\geq}\,1 \wedge t\,{>}\,0\wedge s\,{=}\,o+10$}};
\node[node, below of=13, xshift=0mm, yshift=0mm] (14) {\pcnode{$\m{\smain}$}{$\psi'$}{$b\,{\geq}\,1\wedge o\,{>}\,t \wedge t\,{>}\,0\wedge s\,{=}\,0$}};
\node[node, below of=14] (15) {\pcnode{$\m{\send}$}{$\psi'$}{$b\,{\geq}\,1\wedge o\,{>}\,t \wedge d\,{\geq}\,0 \wedge t\,{>}\,0\wedge s\,{=}\,0$}};
\node[node, below of=15, accepting state] (16) {\pcnode{$\m{\ssold}$}{$\top$}{$b\,{\geq}\,2\wedge o\,{>}\,t \wedge d\,{\geq}\,0 \wedge t\,{>}\,0\wedge s\,{=}\,o+10$}};
\draw[goto] ($(0) + (0,.5)$) -- (0);
\draw[goto,accepting path] (0) to 
 node[action, left]{$\m{a_0}$} 
 (1);
\draw[goto,accepting path] (1) to 
 node[action, left]{$\m{init}$} 
 (2);
\draw[goto,accepting path] (2) to 
 node[action, left]{$\m{check}$} 
 (3);
\draw[goto] ($(3.south) + (-.7,0)$) to 
 node[action, left, very near end, yshift=1mm]{$\m{bid}$} 
 ($(4.north) + (-.7,0)$);
\node[yshift=-4mm, scale=.7, inner sep=5pt] (4dots) at (4.185) {\dots};
\draw[goto] (4.185) -- (4dots);
\draw[goto] (3.east) to[bend left]
 node[action, left, very near end]{$\m{dec}$} 
 (5.10);
\node[yshift=-4mm, scale=.7, inner sep=5pt] (5dots) at (5.352) {\dots};
\draw[goto] (5.352) -- (5dots);
\draw[goto,accepting path] (3.west) to[bend right]
 node[action, left, near start, yshift=-2mm, xshift=-2mm]{$\m{bid}$} 
 node[action, right, near start, yshift=-2mm, xshift=-1mm]{$\{o>t, b=1\}$} (6.170);
\draw[goto,accepting path] (6) to 
 node[action, left]{$\m{check}$} 
 (7);
\node[xshift=6mm, scale=.7, inner sep=5pt] (6dots) at (6.east) {\dots};
\draw[goto] (6) -- (6dots);
\draw[goto,accepting path] (7) to 
 node[action, left]{$\m{bid}$} 
 (8);
\node[xshift=6mm, scale=.7, inner sep=5pt] (7dots) at (7.east) {\dots};
\draw[goto] (7) -- (7dots);
\draw[goto,accepting path] ($(8.south) + (-1,0)$) to 
 node[action, left]{$\m{sell\:now}$} 
 ($(9.north) + (-1,0)$);
\draw[goto,accepting path] (9) to 
 node[action, left]{$\m{fee}$} 
 node[action, right]{$\{\m{\ssold}, b\,{\neq}\,1\}$} (10);
\draw[goto] (9.east) to[bend left] 
 node[action, left, yshift=0mm]{$\m{fee}$} 
 node[action, right, xshift=1mm]{$\{\m{\ssold}, b\,{\neq}\,1\}$} (13.6);
\draw[goto] (8.east) to[bend left]
 node[action, left, near end, xshift=-1mm]{$\m{check}$} 
 ($(11.north) + (.5,0)$);
\draw[goto] (11.west) to[bend left=20]
 node[action, below, near start]{$\m{bid}$} 
 (8.194);
\draw[goto] (4.east) to[bend left]
 node[action, left, near end]{$\m{check}$} 
 ($(12.north) + (1.3,0)$);
\draw[goto] (12.west) to[bend left]
 node[action, left, near end]{$\m{bid}$} 
 (4.190);
\draw[goto] (11.352) to[bend left]
 node[action, left]{$\m{dec}$} 
 (14.east);
\draw[goto] (14) to 
 node[action, left]{$\m{sell\:now}$} 
 (15);
\draw[goto] (15) to 
 node[action, left]{$\m{fee}$} 
 node[action, right]{$\{\m{\ssold}, b\,{\neq}\,1\}$} (16);
\end{tikzpicture}
\caption{Product construction for the auction system.\label{fig:product}}
\end{figure}
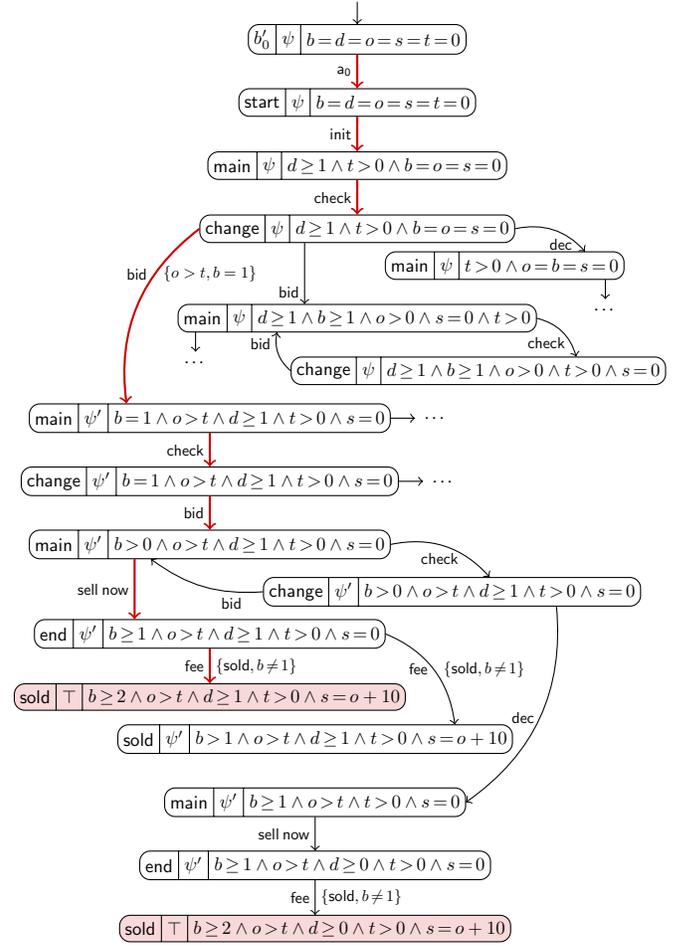
Since $\NN_\BB^\psi$ has final states (shown shaded),
by \thmref{model checking} a witness exists.
We consider the accepting path $\pi$ drawn in red,
and use the soundness part (direction $\Longrightarrow$) of the proof
to extract a respective witness.
By \lemref{PC sat} (1), $\pi$ contains an
accepting transition sequence with word $w$ in $\NN_\psi$ that is consistent with a symbolic run $\sigma$, such that $h(\sigma,w)$ is satisfiable.
Both can be directly read off the path,
here the former is given by
the state sequence $\vec q = (\psi^5\:\psi'^3\:\top)$ with
\\[.5ex]
$w = \emptyset\:\emptyset\:\emptyset\:\{o>t, b=1\}\:
\emptyset\:\emptyset\:\emptyset\:\{\m{\ssold}, b\,{\neq}\,1\}$\\[.5ex]
in $\Sigma^*$ and the symbolic run is
\begin{align*}
\sigma\colon
\m{\sstart} 
 &\goto{\m{init}} \m{\smain} 
  \goto{\m{check}} \m{\schange}
  \goto{\m{bid}} \m{\smain} 
  \goto{\m{check}} \m{\schange} \\
 &\goto{\m{bid}} \m{\smain}
  \goto{\m{sell\:now}} \m{\send}
  \goto{\m{fee}} \m{\ssold}
\end{align*}
Consistency of $w$ and $\sigma$ simply means that they are not contradictory with respect to states and actions.
This is indeed the case since only the last symbol of $w$ mentions a state ($\m{\ssold}$), which matches the last state of $\sigma$.

We now follow the proof of \thmref{model checking} ($\Longrightarrow$) to obtain a witness:
$w$ corresponds to the verification constraint sequence
$\cseq C = \langle \emptyset, \emptyset, \emptyset, \{o\,{>}\,t, b\,{=}\,1\}, 
\emptyset, \emptyset, \emptyset, \{b\,{\neq}\,1\}\rangle$.
Let $C_i$ be the $i$th element of $\cseq C$.
The formula $\hist(\sigma,w)\,{=}\,\hist(\sigma,\cseq C)$ is satisfied
e.g. by the assignment $\alpha$ that maps $(b,d,t,o,s)$ to $(2,3,100,100.5,110.5)$. By the correspondence of satisfying assignments for history constraints
and runs (\lemref{abstraction}),
there exists a run\\[.5ex]
$\rho\colon 
(\m{\sstart}, \alpha_0) \goto{\m{init}} 
(\m{\smain}, \alpha_1) \goto{\m{check}} \dots 
\goto{\m{fee}} (\m{\ssold}, \alpha_7)
$\\[.5ex]
of length $7$ abstracted by $\sigma$
such that $\alpha=\alpha_7$ and 
$\alpha_i$ satisfies $C_i$ for all $i$, $0\,{\leq}\,i\,{\leq}\,7$.
To show existence of $\rho$, the proof of \lemref{abstraction} works its way
back from $\alpha_7$ along the actions in $\sigma$ to construct suitable 
intermediate assignments, for instance the following (values are given for $(b,d,t,o,s)$):\\[.5ex]
$
\begin{array}{r@{\ }l@{\quad}r@{\ }l}
\alpha_0 \colon& (0,0,0,0,0) &
\alpha_3=\alpha_4 \colon & (1,3,100,100.1,0) \\
\alpha_1 = \alpha_2 \colon & (0,3,100,0,0) &
\alpha_5=\alpha_6 \colon & (2,3,100,100.5,0)
\end{array}
$\\[.5ex]
In \lemref{NFA acceptance} it is proven that $\rho\models \psi$. Indeed, the run $\rho$ has the property that bidder 1 submitted
a bid above the threshold but did not win the auction.

\subsection{Counterexamples for Modularity}
\label{app:example:mod}

We first illustrate that an arbitrary composition of {\mydds}s with
\property does not enjoy \property.
\begin{example}
\label{exa:butterfly}
Consider the following \mydds $\BB$, where $x_1, x_2$
are counters and $t$ is an auxiliary variable:
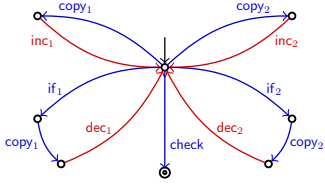
\begin{center}
\begin{tikzpicture}[xscale=.9]
\tikzstyle{mc}=[blue!80!black]
\tikzstyle{ff}=[red!80!black]
\node[state] (1) {};
\node[state] at (20:2cm) (2b) {};
\node[state] at (160:2cm) (2a) {};
\node[state] at (-20:2cm) (3b) {};
\node[state] at (-40:2cm) (4b) {};
\node[state] at (-160:2cm) (3a) {};
\node[state] at (-140:2cm) (4a) {};
\node[state, double] (f) at (0, -1.4) {};
\draw[edge] ($(1) + (0,.4)$) -- (1);
\draw[edge,mc] (1) to[bend right=20] node[above,action, near end]{$\mathsf{copy}_1$} (2a);
\draw[edge,ff] (2a) to[bend right=20] node[below,action, near start, anchor=east, xshift=-2mm]{$\mathsf{inc}_1$} (1);
\draw[edge,mc] (1) to[bend left=20] node[above,action, near end]{$\mathsf{copy}_2$} (2b);
\draw[edge,ff] (2b) to[bend left=20] node[below,action, near start, anchor=west, xshift=2mm]{$\mathsf{inc}_2$} (1);
\draw[edge,mc] (1) to[bend right=20] node[above,action, near end, anchor=east, yshift=1mm]{$\mathsf{if}_1$} (3a);
\draw[edge,mc] (3a) to[bend right=20] node[left,action, anchor=east]{$\mathsf{copy}_1$} (4a);
\draw[edge,mc] (1) to[bend left=20] node[above,action, near end, anchor=west, yshift=1mm]{$\mathsf{if}_2$} (3b);
\draw[edge,mc] (3b) to[bend left=20] node[right,action, anchor=west]{$\mathsf{copy}_2$} (4b);
\draw[edge,ff] (4b) to[bend left=20] node[right,action, anchor=west, xshift=1mm]{$\mathsf{dec}_2$} (1);
\draw[edge,ff] (4a) to[bend right=20] node[left,action, anchor=east, xshift=-1mm]{$\mathsf{dec}_1$} (1);
\draw[edge,mc] (1) --node[right,action, near end]{$\mathsf{check}$}  (f);
\end{tikzpicture}
\end{center}
The action guards are given by
$t^w = x_i^r$ for $\mathsf{copy}_i$,
$x_i^w = t^r + 1$ for $\mathsf{inc}_i$,
$x_i^r > 0$ for $\mathsf{if}_i$,
$x_i^w = t^r - 1$ for $\mathsf{dec}_i$, for both $i \in \{1,2\}$, and
$x_1^r = x_2^r$ for $\mathsf{check}$. 
All the actions can be seen as single-edge \myddss,
have bounded lookback because they are obviously acyclic.
(While bounded lookback is actually a property of a \mydds $\BB$ together with
a constraint set $\CC$, it suffices here to take $\CC=\emptyset$.)
However, $\BB$ models a two-counter system, so that even reachability of control states is undecidable, hence $\BB$ cannot enjoy finite summary.
\end{example}

The next example shows that the parallel composition of {\mydds}s with finite summary need not yield a system with finite summary if the subsystems share variables.

\begin{example}
\label{exa:parallel}
Consider the following two {\mydds}s:\\
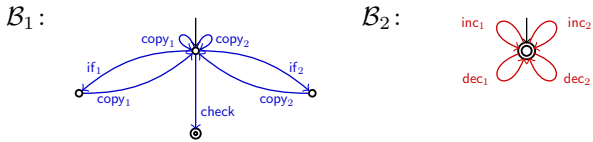
\begin{tikzpicture}[scale=1.1]
\begin{scope}
\tikzstyle{mc}=[blue!80!black]
\node at (-2, .4) {$\BB_1\colon$};
\node[state] (1) {};
\node[state] at (-20:1.5cm) (3b) {};
\node[state] at (-160:1.5cm) (3a) {};
\node[state, double] (f) at (0, -1) {};
\draw[edge] ($(1) + (0,.4)$) -- (1);
\draw[edge,mc] (1) to[loop, out=110, in=160, looseness=20] node[left,action, anchor=east, near end]{$\mathsf{copy}_1$} (1);
\draw[edge,mc] (1) to[loop, out=70, in=20, looseness=20] node[right,action, anchor=west, near end]{$\mathsf{copy}_2$} (1);
\draw[edge,mc] (1) to[bend right=20] node[above,action, anchor=east,  near end,yshift=1mm]{$\mathsf{if}_1$} (3a);
\draw[edge,mc] (3a) to[bend right=20] node[action, below, near start]{$\mathsf{copy}_1$} (1);
\draw[edge,mc] (1) to[bend left=20] node[above,action,anchor=west, near end, yshift=1mm]{$\mathsf{if}_2$} (3b);
\draw[edge,mc] (3b) to[bend left=20] node[below,action, near start]{$\mathsf{copy}_2$} (1);
\draw[edge,mc] (1) --node[right,action, near end]{$\mathsf{check}$}  (f);
\end{scope}
\begin{scope}[xshift=40mm]
\node at (-1.7, .4) {$\BB_2\colon$};
\tikzstyle{ff}=[red!80!black]
\node[state, inner sep=3pt, double] (1) {};
\draw[edge] ($(1) + (0,.4)$) -- (1);
\draw[edge,ff] (1) to[loop, out=110, in=160, looseness=20] node[left,action, anchor=east, xshift=-2mm]{$\mathsf{inc}_1$} (1);
\draw[edge,ff] (1) to[loop, out=70, in=20, looseness=20] node[right,action, anchor=west, xshift=2mm]{$\mathsf{inc}_2$} (1);
\draw[edge,ff] (1) to[loop, out=-20, in=290, looseness=20] node[right,action, anchor=west, xshift=1mm]{$\mathsf{dec}_2$} (1);
\draw[edge,ff] (1) to[loop, out=200, in=250, looseness=20] node[left,action, anchor=east, xshift=-1mm]{$\mathsf{dec}_1$} (1);
\end{scope}
\end{tikzpicture}
Both $\BB_1$ and $\BB_2$ on their own have 2-bounded lookback (they are even feedback-free), and hence have finite summary.
However, a suitable interleaving of the actions of $\BB_1$ and $\BB_2$ emulates the system in \exaref{butterfly}, hence such a parallel composition cannot preserve finite summary.
\end{example}

\subsection{Tool Results on Examples}
\label{sec:examples:tool}

We report on results obtained with \tool on some examples.
However, the tool is a prototype implementation,
and we intend to provide a proof of feasibility
rather than a comprehensive evaluation, which is left for future work.
\begin{compactitem}
 \item For the auction system $\BB_a$ in \exaref{auction}, we considered several relevant properties.
 For each of the properties $\psi_{1i}$ below, \tool
 can detect finite summary with respect to the set $\CC$ of constraints in $\psi$
 by using our decomposition results, as explained in \exaref{auction2}.
 Here, the unsatisfiability of $\psi_{11}$ shows that if the item is sold before 
 the auction expires, the price exceeds the threshold, and $\psi_{12}$ expresses
 that the first to bid above the threshold need not win (cf. \exaref{auction2}). 
Unsatisfiability of $\psi_{13}$ shows that when the $\mathsf{sold}$ state is reached, the bidder variable $b$ was set.
Satisfiability of $\psi_{14}$ shows that there is a run where
$s$ is not set until the auction expired or the threshold is reached;
and in fact we can show that its negation $\psi_{15}$ is unsatisfiable.
\item In \cite{MannhardtLRA16} a Petri net with data (DPN) is used
to describe a road fine management process by the Italian police.
This net has a single token at any point of time, and is hence easily transformed
into a \mydds that we call $\BB_r$ here.
It has nine states, 19 transitions, eight integer and rational variables,
and uses monotonicity constraints as well as more
complex constraints like
$\mathit{total} \geq \mathit{amount} + \mathit{expenses}$ as guards.
However, the system enjoys 2-bounded lookback and hence finite summary
with respect to all constraints occurring in the formulas below.
To verify that the variables $\mathit{delaySend}$, $\mathit{delayPrefecture}$,
and $\mathit{delayJudge}$ (abbreviated $\mathit{dS}$, $\mathit{dP}$,
and $\mathit{dJ}$) remain within the desired intervals $[0,2160[$,
$[0,1440[$, and $[0,1440[$, respectively, we show that 
$\psi_{21}$, $\psi_{22}$, and $\psi_{23}$ are unsatisfiable.
Satisfiability of $\psi_{24}$ shows that when the $\mathsf{end}$ state is 
reached, it might be the case that the total amount paid is less than the
expenses, and satisfiability of $\psi_{25}$ that the total amount paid
can remain 0 over the entire run.
To verify the liveness property
$\Box\,(\langle\mathsf{notify}\rangle\top \to \Diamond\, (\mathit{dismissal}\,{=}\, 0))$, we can check that its negation $\psi_{26}$
is unsatisfiable.
(Though $\top \not\in \LBC$, it can be encoded, e.g. as $x=x$).
\item
In~\cite{Mannhardt18} a hospital billing process is given as a DPN that is one-bounded, but has markings with more than one tokens. By enumerating all possible markings and representing string constants as integers, the DPN is transformed into a \mydds $\BB_h$ with 16 states, 40 transitions, and four variables. The system has bounded lookback with respect to the constraints in the formulas below.
We can check pure state reachability of $\mathsf{p16}$ with formulas like $\psi_{31}$.
To verify the liveness property
$\Box\,(\mathsf{p40} \wedge \mathit{closed}\,{=}\, 1 \to \Diamond\, (\mathit{ccode}\,{>}\, 0))$, we can check that its negation $\psi_{32}$
is unsatisfiable.
Finally, unsatisfiability of $\psi_{33}$ shows that a $\mathsf{reopen}$ action
never occurs right after a $\mathsf{storno}$ action.
We observe that
\tool is less efficient for $\BB_h$ than for other \myddss, likely because the summary set is comparatively large.
\end{compactitem}
\smallskip
 The following table lists \myddss with the checked properties, 
and indicates whether a witness exists, as well as the analysis time
 in seconds.\\[1ex]
\noindent
\begin{tabular}{@{}l@{\ }l@{$\colon\:$}ll@{\:\:}r@{}}
$\BB_a$ 
 & $\psi_{11}$ & $\Diamond(\mathsf{sold} \wedge d\,{>}\,0 \wedge o\,{\leq}\,t)$ & \no & 2.5\\
 & $\psi_{12}$ & $\Diamond(b\,{=}\,1 \wedge o\,{>}\,t \wedge \Diamond(\mathsf{sold} \wedge b\,{\neq}\,1))$ & \yes & 4.3\\
 & $\psi_{13}$ & $\Diamond(\mathsf{sold} \wedge b\,{=}\,0)$ & \no & 2.1\\
 & $\psi_{14}$ & $s\,{=}\,0 \until (d\,{\leq}\,0 \vee o\,{>}\,t)$ & \yes & 5.1\\
 & $\psi_{15}$ & $\Box\,(s\,{=}\,0) \vee ((d\,{>}\,0 \wedge o\,{\leq}\,t) \until s\,{\neq}\,0)$ & \no & 3.1\\[.5ex]
$\BB_r$ 
 & $\psi_{21}$ & $\Diamond\,(\mathit{dS}\,{\leq}\,0 \vee \mathit{dS}\,{\geq}\, 2160)$ & \no & 1.1\\
 & $\psi_{22}$ & $\Diamond\,(\mathit{dP}\,{\leq}\,0 \vee \mathit{dP}\,{\geq}\, 1440)$ & \no & 1.2\\
 & $\psi_{23}$ & $\Diamond\,(\mathit{dJ}\,{\leq}\,0 \vee \mathit{dJ}\,{\geq}\, 1440)$ & \no & 1.1\\
 & $\psi_{24}$ & $\Diamond\,(\mathsf{end} \wedge \mathit{total}\,{<}\, \mathit{amount}\,{+}\,\mathit{expense})$ & \yes & 1.0\\
 & $\psi_{25}$ & $\Box\,(\mathit{total}\,{=}\, 0)$ & \yes & 1.5\\
 & $\psi_{26}$ & $\Diamond\,(\langle\mathsf{notify}\rangle \top \wedge \Box\, (\mathit{dismissal} \,{\neq}\,0))$& \no & 2.1\\[.5ex]
$\BB_h$ 
 & $\psi_{31}$ & $\Diamond\,\mathsf{p16}$ & \yes & 27.1\\
 & $\psi_{32}$ & $\Diamond\,(\mathsf{p40} \wedge \mathit{closed}\,{=}\, 1 \wedge \Box\, (\mathit{ccode} \leq 0))$ & \no & 19.4 \\
 & $\psi_{33}$ & $\Diamond\,(\langle\mathsf{storno}\rangle\,\langle\mathsf{reopen}\rangle\,\top)$ & \no & 21.7\\
\end{tabular}\\[1ex]
\tool comes with a test script that checks 43 properties of 14 systems,
including the above.

We emphasize that the tool is a proof of concept implementation, and many
improvements can be seemlessly incorporated to optimize its performance. These range from
constructing more efficient automata for LTL$_f$ formulas~\cite{XiaoL0SPV21,GiacomoF21} to
a more succinct encoding when searching for witnesses (e.g., reusing formula
parts).

\end{document}